\documentclass[a4paper,10.5pt]{article}

\usepackage[top=1in, bottom=1in, left=1in, right=1in]{geometry}
\usepackage{setspace}
\usepackage{graphicx}
\usepackage{amsmath, amsfonts, amssymb, bbm}
\usepackage{booktabs, longtable, multirow, threeparttable}
\usepackage{siunitx}
\usepackage[inline]{enumitem}
\usepackage{xcolor}
\usepackage{float}
\usepackage{pdflscape}
\usepackage{tikz}
\usepackage{pgfplots}
\pgfplotsset{compat=1.17}
\usepackage{subcaption}
\usepackage{hyperref}
\usepackage{url}
\usepackage{natbib}
\usepackage{amsthm}
\usepackage{appendix}
\usepackage{acronym}
\usepackage{courier}
\usepackage{sectsty}
\usepackage{booktabs}
\usepackage{lscape}

\usetikzlibrary{arrows.meta, positioning, shapes.geometric, fit}
\tikzstyle{provider} = [rectangle, rounded corners, draw=black, fill=yellow!30, text centered, minimum height=1.2cm, minimum width=3.5cm]
\tikzstyle{pool} = [circle, draw=black, fill=purple!30, text centered, minimum size=2cm]
\tikzstyle{arrow} = [thick, ->, >=stealth]
\tikzstyle{decision} = [diamond, draw=black, fill=green!30, text centered, minimum height=1.5cm, minimum width=3cm]
\tikzstyle{fee} = [trapezium, trapezium left angle=120, trapezium right angle=60, draw=black, fill=teal!30, text centered, minimum width=2cm, minimum height=1cm]
\tikzstyle{action} = [rectangle, draw=black, fill=gray!20, text centered, minimum height=1.2cm, minimum width=4cm]

\AtBeginEnvironment{tabular}{\scriptsize}
\AtBeginEnvironment{longtable}{\scriptsize}
\AtBeginEnvironment{longtable}{\singlespacing}
\usepackage[labelfont=bf,font={scriptsize},justification=justified,singlelinecheck=false]{caption}

\sectionfont{\normalsize\bfseries}
\subsectionfont{\normalsize\itshape}
\subsubsectionfont{\normalsize\itshape}
\paragraphfont{\normalsize\itshape}

\hypersetup{
    hidelinks,
    colorlinks=true,
    citecolor=blue,
    linkcolor=blue,
    urlcolor=blue,
    pdfmenubar=false,
    pdftoolbar=false
}

\renewcommand{\thefootnote}{\arabic{footnote}}

\newtheorem{definition}{Definition}
\newtheorem{proposition}{Proposition}
\newtheorem{assumption}{Assumption}
\newtheorem{corollary}{Corollary}
\newtheorem{remark}{Remark}
\newtheorem{procedure}{Procedure}

\usepackage{tikz}
\usetikzlibrary{arrows.meta,positioning,fit,calc,shapes.geometric,decorations.pathreplacing}

\begin{document}

\title{Vault as a credit instrument}
\author{
  \scriptsize Anastasiia Zbandut$^{\ast}$\\ 
  \scriptsize\texttt{anastasiia@blockworks.co}
  \and
  \scriptsize Carolina Goldstein\\ 
  \scriptsize\texttt{carolina@blockworks.co}
}

\date{}

\maketitle

\begingroup
\renewcommand{\thefootnote}{\fnsymbol{footnote}}
\setcounter{footnote}{1}
\footnotetext{Corresponding author.}
\endgroup

\setcounter{footnote}{0}

\begin{abstract}
We derive five tractable credit risk metrics for DeFi lending vault depositors, grounded in a formal three level decomposition of vault risk into mechanical loss channels (Level 1), governance quality (Level 2) and smart contract code integrity (Level 3). For Level 1, we show that six structural features of onchain execution (oracle execution divergence, endogenous recovery, full information run dynamics, timelock constrained governance, oracle manipulation and congestion driven liquidation failure) break canonical TradFi analogies and generate depositor loss channels absent from standard credit frameworks. Vault credit risk metrics translate these channels into measurable risk components which are aggregated into a vault credit score. The empirical contribution is an implementable estimation architecture for credit risk metrics, including required onchain data, identification strategies for core parameters, partial identification bounds and a coherent stress scenario methodology. The results have direct implications for vault risk management and for minimum transparency standards necessary for depositor risk assessment.

\end{abstract}

\section{Introduction}
Decentralized finance (DeFi) lending protocols intermediate large volumes of collateralized credit, yet their depositor protection problem remains largely unaddressed by a formal credit risk measurement framework. This paper takes the DeFi lending vault as the unit of analysis and treats it as a credit instrument whose loss distribution is induced by collateral coverage, liquidation and recovery mechanics, and liquidity constraints. We show that DeFi vault depositor risk decomposes into three structurally independent levels that require distinct measurement methodologies. Level 1 comprises the mechanical loss channels arising from collateral execution, recovery endogeneity, funding liquidity fragility, oracle integrity and congestion driven liquidation failure. Standard structured credit frameworks miss these channels because they treat collateral valuation, recovery and settlement as exogenous and friction free. Level 2 comprises governance risk which includes the quality of the parameter setting process and its response latency under stress. Level 3 captures infrastructure and code integrity, assessing whether the vault’s smart contract dependency graph executes according to its intended specification. We derive the condition under which each level is the binding depositor risk concern, show that Level 3 failure probability is superadditive in protocol dependency depth, and, for Level 1, formalize six breakdown conditions for canonical TradFi analogies, derive five tractable metrics (V1-V5) and aggregate them into a vault credit score (VCS).

DeFi lending protocols perform a recognizable banking function through smart contracts deployed on public blockchains. They transform depositor funds into collateralized credit exposure and promise withdrawal rights under rule bound constraints. However, the analytical framework for measuring depositor credit risk in this setting has not evolved at the same pace as the scale of activity. Existing practice substitutes qualitative risk labels, ad hoc scorecards or direct imports of traditional credit metrics. The gap is not only operational, but also methodological. It reflects the absence of a measurement theory grounded in the structural features of a novel execution environment. The one where valuation, enforcement and settlement each operate under frictions that are absent in the institutional structured finance context from which existing methodology was drawn.

The Level 1 analysis develops from a specific structural analogy. DeFi lending vaults share a recognizable skeleton with traditional secured credit instruments since they are overcollateralized, impose threshold based liquidation enforcement and expose depositors to funding liquidity constraints through utilization and withdrawal limits. These similarities motivate the import of structured credit methodology. Six DeFi specific design features break the analogy, however, each breakdown generates a distinct depositor loss channel absent from standard TradFi credit measurement. Valuation uses public oracle networks compared to discretionary pricing agents, introducing manipulation and latency risks. Liquidation executes onchain through automated mechanisms, making realized recovery endogenous to liquidation mass and market depth. Execution requires blockspace whose cost is endogenous to stress through gas pricing and MEV competition. Parameter changes are timelock constrained, so operators cannot intervene at the point of crisis. Additionally, there is no external liquidity backstop meaning withdrawals are mechanically blocked at full utilization and run dynamics unfold under near common information.

Each of the three levels requires distinct methodology. This paper develops Level 1 in full with six formal propositions, five metrics,
and an estimation architecture. Level 2 governance risk introduces additional state variables and behavioral identification challenges that require separate treatment and are developed in a companion paper. Level 3 infrastructure risk is formally characterized here (through the superadditivity proposition and the dominance condition) but its estimation architecture, which requires smart contract security methodology and exploit probability calibration, constitutes a distinct research agenda identified at the close of this paper.

The paper makes four contributions. First, it formalizes a three level decomposition of vault credit risk into mechanical loss channels (Level 1), governance quality (Level 2) and infrastructure and code integrity (Level 3). Thereby, establishing the structural scope boundary between levels and deriving the condition under which each level is the binding depositor risk concern (Proposition~\ref{prop:L3_dominance},  Figure~\ref{fig:three_level}). Second, it provides the first formal, closed form characterization of six precise TradFi to DeFi breakdown conditions via propositions with stated assumptions, converting an informal practitioner recognition of analogy limits into auditable credit measurement primitives. Third, we derive five Level 1 metrics (V1-V5) that quantify the mechanical depositor loss channels implied by these failure modes and aggregate them into a vault credit score (VCS) with explicit weakest-link and additive operators (Figure~\ref{fig:vcs_architecture}). Fourth, we provide an empirically implementable estimation architecture with a data stack sufficient to compute V1-V5, estimation procedures for key parameters, partial identification bounds, a stress scenario methodology and a validation strategy for V1 and V2 (Figures~\ref{fig:estimation_pipeline} and~\ref{fig:data_deps}).

The Level 1 analysis integrates four strands of prior work. The structured credit and CLO methodology literature provides the conceptual foundation for coverage tests, tranche protection logic, and collateral manager constraints \citep{Merton1974, Vasicek1987, McNeilFreyEmbrechts2005, MoodysCLO, SPCLO}. However, it does not provide a formal mapping of those concepts into an onchain execution environment where valuation, enforcement and settlement all operate under distinct frictions. The DeFi liquidation mechanics literature documents liquidation cascades, onchain depth depletion and stressed execution outcomes \citep{GudgeonWernerPerezKnottenbelt2020, PerezWernerXuLivshits2021, QinZhouGamitoJovanovicGervais2021, BertomeuMartinSall2024} but it typically does not embed these mechanisms into an instrument level depositor loss framework with explicit coverage and shortfall functions. The oracle manipulation and TWAP security literature formalizes manipulation vectors and defenses for onchain price feeds \citep{AngerisChitra2020, AngerisKaoChiangNoyesChitra2020, QinZhouLivshits2021}. We add the translation of oracle error into depositor credit losses through structural solvency and recovery mechanics to this strand. Finally, the run dynamics and stablecoin fragility literature analyzes onchain liquidity crises and depositor behavior in crypto intermediaries \citep{DiamondDybvig1983, MorrisShin1998, GoldsteinPauzner2005, BrunnermeierPedersen2009, GortonMetrick2012, KlagesMundtHarzGudgeon2020, AppelGrennan2025}. We apply these findings to provide a unified credit risk architecture connecting funding liquidity constraints to liquidation and execution frictions in a vault balance sheet. A fifth strand, concerning MEV and blockspace competition as endogenous execution costs, is addressed by \citet{DaianGoldfeder2020} and is integrated here through Proposition~\ref{prop:network_congestion} and V5 metric. By synthesizing these strands under a single depositor loss function, the paper provides a unified measurement framework.

The remainder of the paper proceeds as follows. Section~\ref{sec:lit} reviews the related literature, while Section~\ref{sec:theory} defines the DeFi lending environment, formalizes the depositor loss function and derives six TradFi to DeFi failure modes as formal propositions. Section~\ref{sec:metrics} derives credit metrics V1-V5 and the vault credit score VCS. Section~\ref{sec:empirical} specifies the data architecture, estimation procedures, stress scenarios and validation strategy. Section~\ref{sec:discussion} discusses contributions relative to the literature, binding limitations and policy implications. Section~\ref{sec:conclusion} concludes.

\section{Literature review}
\label{sec:lit}

This paper integrates and extends five strands of financial research which are reviewed in turn below.

\paragraph{Structural credit models and structured finance methodology.}\leavevmode\\
The theoretical foundation for coverage based default risk measurement is the \citet{Merton1974} model for company valuation, in which default occurs when the value of a firm’s assets falls below its contractual debt obligations and recovery is endogenously determined by the asset value at the default event. \citet{Vasicek1987} derive portfolio credit loss distributions under a single systematic factor, providing the analytical underpinning for regulatory capital frameworks and the asymptotic loan portfolio models that inform modern LGD estimation. \citet{AroraBohnZhu2006} provide the first out of sample comparative test of the Merton, Vasicek-Kealhofer (VK) and Hull-White reduced form models on 542 US corporate issuers matched with credit default swap data over October 2000-June 2004. Authors find that the VK model achieves an accuracy ratio of 0.801 in discriminating defaulters from nondefaulters versus 0.652 for the basic Merton implementation (a 15 percentage point superiority) and explains substantially more cross-sectional variation in CDS spreads (median R$^{2}$ = 0.48 versus 0.26 for Merton), with the smallest median absolute spread error of 33.28 bps versus 53.99 bps. The evidence that the enriched VK structural specification (empirical distance to default mapping to EDF) materially outperforms the single boundary Merton model directly motivates the paper's adoption of VK type distance to default as the structural foundation for V1. The VK model is calibrated on equity market data for listed firms with observable financial statements, thus, applying its architecture to DeFi vaults requires replacing book leverage with oracle priced collateral ratios and firm asset volatility with collateral return volatility. This structural adaptation is addressed in Section~\ref{sec:empirical}.

\citet{McNeilFreyEmbrechts2005} provide a comprehensive treatment of probability of default, loss given default and copula based joint default dependence structures. \citet{Gordy2000} offers a comparative anatomy of CreditMetrics and CreditRisk$^{+}$, demonstrating that the two dominant portfolio credit risk architectures share a common mathematical skeleton and that differences in simulated tail loss are attributable to the distributional assumptions for the systemic risk factor. His finding that tail percentiles at the 99.5th and 99.97th levels are highly sensitive to the kurtosis of the systemic factor (but not to its mean or variance) provides a methodological precedent for the vault paper's finding that expected shortfall is superlinear in liquidation mass (Proposition~\ref{prop:recovery_endogeneity}) and that stress scenario design must address the shape of the loss distribution. 

\citet{DuffieEcknerHorelSaita2009} extend this concern to the empirical identification of latent common factors, demonstrating in a dataset of 2,793 US public nonfinancial firms over 1979-2004 that observable covariate based portfolio loss models underestimate 99th percentile default counts by a full order of magnitude once a common dynamic frailty factor is incorporated. The frailty factor increases proportional intensity volatility by approximately 40\% beyond observable covariates and its presence is supported by a Bayes factor well above the threshold for very strong evidence. Thus, models without frailty exhibit serially correlated forecast errors ($p = 0.004$) while the frailty model does not ($p = 0.62$). The vault level analog is observable individual position metrics, e.g., health factor distributions, current utilization, realized slippage, cannot capture shared oracle dependencies, common liquidator network capacity and correlated collateral price paths that constitute a latent cross-vault factor not reducible to any single observable. \citet{DuffieEcknerHorelSaita2009} identification architecture is not directly portable since their estimation requires approximately 400,000 firm months of default data unavailable in DeFi and their factor operates at the economic cycle scale compared to block time scale governing onchain liquidation clustering. However, their central finding that observable factor models systematically and severely underestimate portfolio tail losses in the presence of a common latent factor motivates the stress scenario conditioning requirement throughout this paper and provides formal precedent for the superlinearity result in Proposition~\ref{prop:M2_superlinearity}. 

\citet{TreacyCarey2000} document internal credit risk rating practice at the 50 largest US banking organizations, establishing the empirical landscape of PD/LGD decomposition, rating scale design and the through the cycle versus point in time distinction in practitioner rating systems. Their observation that all surveyed banks assign ratings on a point in time basis, i.e., calibrating to current borrower condition, directly parallels the Level 1 design choice in the paper of computing VCS conditional on the current vault state $X_t^{v}$. \citet{CareyHrycay2001}, a companion study, empirically investigate the properties of the two dominant rating quantification methods, i.e., agency grade mapping and credit scoring model quantification, finding that apparently minor methodological variations change implied capital allocation ratios by several percentage points. Their analysis isolates two biases intrinsic to quantification. First, an informativeness bias that causes overoptimism for high risk grades and overpessimism for low risk grades, and, second, a noisy rating assignments bias that compounds this distortion when the same model is used for both grade assignment and quantification. For investment grades, median based methods outperform mean based methods and, for speculative grades, the ranking reverses. Their evidence that actuarial methods require at least twelve years of default observations for speculative grade entities to produce estimates within one-third of the long run average has a direct counterpart in the vault setting. DeFi protocol histories typically spanning fewer than five years, any VCS anchored to historical default frequency by grade cannot rely on stable actuarial calibration and must instead employ structural estimation, precisely the architecture developed in Section~\ref{sec:empirical}. They also document a gaming vulnerability where strategic grade boundary manipulation can cause actual default rates to exceed estimated rates by a factor of two or more, a finding with structural parallels to the curator parameter gaming risk analyzed in the companion paper. 

\citet{CantorPackerCole1997} provide empirical evidence on rating aggregation for split rated corporate bonds, showing that averaging two rating opinions produces strictly lower forecast error (RMSE 56.8 basis points) than any single agency or conservative selection rule (62--63 basis points) across a sample of 4,399 US corporate bond issues from 1983 to 1993. This result is directly relevant to the VCS aggregation designsince in normal regimes where loss channels are not perfectly correlated, the additive operator $\mathrm{VCS}^{\mathrm{add}}$ inherits the efficiency advantage of averaging. However, \citet{CantorPackerCole1997} finding derives from a context with institutional pricing and market discipline for bonds, and does not apply to DeFi vault scenarios in which a single failure mode, e.g., oracle manipulation, gas spike, or utilization ceiling breach, can simultaneously invalidate the readings of all five metrics. For such single point failure scenarios, the multiplicative weakest-link operator $\mathrm{VCS}^{\mathrm{mult}}$ is appropriate. 

In structured finance, rating agency CLO criteria operationalize coverage tests and collateral manager constraints as explicit credit triggers. Moody's market value and cash flow CLO methodologies specify overcollateralization test thresholds, cure periods and deleveraging waterfalls as responses to coverage breaches. S\&P Global Ratings CLO criteria formalize analogous interest coverage and OC test logics with explicit default treatment assumptions. Recent work confirms the discriminating power of the structural distance to default framework over long horizons. \citet{EldomiatyAzzam2025} estimate historical and stochastic distance to default via geoBrownian motion for 109 nonfinancial US firms over 1952-2023, finding that leverage and bankruptcy costs reduce distance to default while firm size amplifies it and that GBM simulated distance to defaults systematically exceed historical realizations across the full panel, supporting stochastic evaluation of distance to default as a forward looking credit risk measure. These firm level determinants, i.e., leverage ratio, tax shield structure and industrial output, do not transfer to DeFi vaults, where the structural analogs of firm assets and equity are oracle priced collateral portfolios and protocol surplus. However, the GBM simulation architecture is conceptually aligned with the stochastic V1 estimation procedure developed in Section~\ref{sec:empirical}. The paper demonstrates that collateral valuation, recovery, and liquidity are endogenously determined in a DeFi execution environment, requiring their systematic redevelopment instead of a recalibration of legacy parameters.

\paragraph{DeFi lending protocols, liquidation mechanics and protocol fragility.}\leavevmode\\
The DeFi lending literature has grown substantially since the emergence of overcollateralized protocols. \citet{GudgeonWernerPerezKnottenbelt2020} analyze interest rate dynamics, liquidity and market efficiency in DeFi loanable fund protocols, providing an economic framework for understanding utilization and borrow rate dynamics that directly informs V3. \citet{PerezWernerXuLivshits2021} study DeFi liquidations as a game theoretic phenomenon, documenting that cascading liquidations can render protocols insolvent when collateral prices decline rapidly and identifying gas competition as a first order determinant of liquidation outcomes. Their result that liquidation viability degrades under congestion is formalized in the paper as Proposition~\ref{prop:network_congestion}. 

\citet{QinZhouGamitoJovanovicGervais2021} provide an empirical study of DeFi liquidations across major protocols, quantifying incentive structures, liquidation risks and instabilities under different design choices, their evidence on endogenous execution quality motivates the impact function in V2. \citet{WernerPerezGudgeon2022} offer a systematic conceptual account of DeFi protocol risk dimensions, providing a taxonomy that contextualizes the paper's level decomposition. \citet{Kirisci2025} reports expert elicitation evidence which prioritizes operational and technical risks above financial and emergent categories across a five dimensional DeFi risk hierarchy, a ranking consistent with the paper's emphasis on execution failure and oracle manipulation as the primary Level 1 loss channels. 

\citet{GudgeonPerezHarzLivshitsGervais2020} develop a formal economic security framework for DeFi lending protocols, providing definitions of overcollateralization, liquidity, and counterparty risk as explicit security constraints and demonstrating\footnote{Applying Monte Carlo stress testing calibrated to Ethereum's price history} that a DeFi lending protocol with \$400 million of initial debt can become undercollateralized within 19 days under an illiquidity decay parameter of $\rho = 0.01$, with the time to undercollateralization strictly decreasing in both system debt and initial illiquidity. Their result that a positively correlated reserve asset provides limited protection against a collateral price decline directly motivates the scenario correlation structure in Section~\ref{subsec:stress_scenarios}, and their formal overcollateralization and liquidity constraints are direct precursors to the vault paper's health factor and ACR definitions. The boundary hitting probability $\mathrm{V3}(v,t;h)$ in \eqref{eq:M3_def} is the onchain operationalization of precisely the question their stress tests address, i.e., how quickly can a withdrawal blocking boundary be reached under adverse price and liquidity conditions, extended here to incorporate utilization dynamics, jump driven withdrawal clustering and the endogenous oracle and execution frictions absent from their stylized model. 

\citet{BertomeuMartinSall2024} develop two aggregate protocol risk measures, i.e., $R_I$, the price decline required to trigger liquidation of a synthetic marginal borrower, and $R_{II}$, the further decline at which collateral is exhausted post liquidation, calibrated to daily Aave and Compound data over April 2021-July 2023. Their central findings are that system fragility ($R_I > 1$) was briefly breached in July 2021 and that their measures detect elevated fragility around the UST collapse and FTX bankruptcy and that daily liquidations respond more strongly to contemporaneous risk measure changes than to lagged changes, indicating that liquidation execution is not timely relative to risk dynamics. \citet{IftikharWeiCartlidge2025} provide cross-version and crosschain panel regression evidence on the effectiveness of automated liquidation risk management in Aave and Compound, estimating fixed effects models on 1{,}456 daily observations from January 2021 to December 2024. Their central finding is a sharp reversal in the liquidation. Performance relationship across protocol generations is as followins, in v2 protocols deployed on Ethereum (L1), a 1\% increase in liquidation volume is associated with a 0.01\% decline in TVL, consistent with cascade driven depositor outflows. While in v3 protocols deployed on Arbitrum (L2), the same 1\% increase is associated with a 0.07\% rise in TVL and a 0.11\% rise in total revenue (both significant at the 1\% level, adjusted $R^2 = 0.61$). The reversal reflects v3's risk parameter architecture, i.e., asset specific collateral caps, isolation mode for volatile assets and enhanced liquidation incentive calibration, combined with L2's lower transaction costs and faster execution throughput, which reduce the probability of failed liquidations under congestion. This constitutes direct empirical support for the proposition that automated liquidation mechanism design is a first order determinant of protocol resilience under stress. This is consistent with the mechanism formalized in Proposition~\ref{prop:recovery_endogeneity}, well parameterized liquidation triggers reduce endogenous recovery shortfalls. Aave v3 isolation mode and efficiency mode have no direct structural counterpart in Morpho Blue's immutable market architecture. The paper's V1 and V2 metrics apply at the vault level aggregate and not to individual dynamic collateral parameters, so this empirical generalization requires architectural qualification.

\citet{ParhizkariIannillo2025} provide the first fully automated, onchain only security risk scoring framework for DeFi projects, introducing four behavioral and code structural metrics, i.e., code vulnerability density weighted by severity and a temporal age decay function, failed transaction ratio, normalized transaction volume variation and new origin EOA transaction fraction, aggregated via the multiplicative complement risk likelihood $= 1 - \prod_{i=1}^{4}(1 - \mathrm{metric}_i)$. Evaluated against 220 confirmed security breaches across seven EVM compatible chains, the model achieves ROC AUC of 0.887, recall of 86.4\%, and precision of 78.5\%, where lending protocols record a category specific attack detection rate of 96.4\%. The result that transaction behavioral signals, i.e., volume variation and new origin fractions, carry correlations above 40\% with the attack label while code vulnerability density contributes orthogonal predictive content empirically validates the vault paper's classification of execution failure and smart contract code integrity as structurally independent Level 3 risk channels. It demonstrates that the node failure probabilities $q_i(h)$ in Proposition~\ref{prop:L3_superadditive} are estimable from behavioral onchain data, making the dominance condition \eqref{eq:L3_dominance} operationally tractable.

\citet{AdamykBenson2025} evaluate six leading DeFi risk tracking platforms, i.e., Chainalysis, Elliptic, Nansen, Dune Analytics, DeBank, and Etherscan, -against seven weighted criteria (compliance features and real time monitoring, 20\% each; data accuracy, usability, and advanced analytics, 15\% each; integration and cost effectiveness, 10\% each), using a mixed methods design combining a survey of 138 practitioners with 12 stakeholder interviews. ANOVA confirms statistically significant performance differences across all criteria ($F_{5,132} = 11.78$, $p < 0.001$, partial $\eta^2 = 0.27$ for overall effectiveness; compliance features: $F = 14.01$, partial $\eta^2 = 0.31$), and user trust regresses significantly on data accuracy ($\hat{\beta} = 0.78$, $p < 0.01$), compliance features ($\hat{\beta} = 0.71$, $p < 0.01$) and real time monitoring ($\hat{\beta} = 0.65$, $p < 0.05$). The central finding is systematic fragmentation across platforms is that Chainalysis and Elliptic lead in compliance and analytics (mean compliance scores 4.51 and 4.45 on a five point scale) but score lowest on cost effectiveness (2.87 and 2.95) and usability (3.21 and 3.08). Nansen leads in real time monitoring (4.52) but lacks compliance depth (3.04), while DeBank and Etherscan achieve cost effectiveness and usability leadership but provide insufficient analytics and compliance for institutional risk management (advanced analytics means 3.20 and 3.09, respectively). The platforms with advanced analytics and real time monitoring are significantly more effective overall than those without ($t(294) = 7.45$, $p < 0.01$). This fragmentation across six specialized tools, documented through a sample of DeFi practitioners, constitutes direct survey based evidence for the measurement gap the paper addresses which is no evaluated platform provides per channel, mechanism grounded decomposition of vault level risk into independently interpretable components simultaneously accessible to non institutional participants.

\citet{ZhangWangZhengCartlidge2026} develop a unified network based framework for ecosystem level systemic risk monitoring in DeFi, constructing rolling shrinkage correlation networks from TVL log returns across 70 protocol categories (DeFiLlama, 2021-2025) and aggregating four spectral fragility metrics, i.e., average correlation strength, maximum eigenvalue, strong edge density and eigenvalue entropy, into a Correlation Fragility Indicator (CFI) via PCA, with the first component explaining 96.73\% of joint variation. The CFI is positively associated with aggregate TVL volatility ($\hat{\beta} = 0.0020$, $p < 0.10$, HAC) but uncorrelated with ETH spot volatility once controls are included, and Granger predicts future TVL growth volatility at horizons of 7, 14, and 30 days (coefficients 0.0028, 0.0017, and 0.0010, respectively; all $p < 0.01$, $n = 239$ rolling windows). Their node level Risk Contribution Score (RCS), which is a counterfactual measure of each protocol category's marginal contribution to the CFI, reveals that structural systemic importance diverges sharply from economic scale. Lending ranks first by TVL but seventh by RCS, insurance ranks thirtieth by TVL but second by RCS and targeted removal of the ten highest RCS categories reduces CFI by 1.072 standard deviations versus 0.157 for random removal. The result is that DeFi systemic fragility is driven by dependency structure and not by balance sheet size, and accumulates through gradual correlation concentration. This result directly corroborates the vault paper's V2 specification, in which curator overlap and shared oracle exposure determine the concentration risk contribution of each borrower cluster and reinforces V5's treatment of tail dependence as a structural property requiring continuous monitoring.

\paragraph{Oracle security, price feed manipulation and TWAP design.}\leavevmode\\
Oracle mechanisms are the primary interface between offchain price discovery and onchain enforcement in DeFi. \citet{AngerisChitra2020} analyze constant function market makers as price oracles and characterize the manipulation resistance of time weighted average price feeds as a function of liquidity depth and the cost of price displacement, establishing formal conditions under which oracle manipulation is economically infeasible. \citet{AngerisKaoChiangNoyesChitra2020} extend this analysis to Uniswap market structure, quantifying how oracle prices respond to large trades and how pool parameters affect oracle accuracy under adversarial conditions. \citet{QinZhouLivshits2021} document flash loan enabled oracle manipulation in DeFi protocols, identifying the conditions under which short term oracle distortions generate economically significant misallocations in downstream credit enforcement. The paper integrates these contributions into a formal credit risk framework through Propositions~\ref{prop:oracle_failure} and~\ref{prop:oracle_manipulation}, which translate oracle error into depositor shortfall, and V4, which provides an operationalizable oracle integrity score distinguishing latency and manipulation channels.

\paragraph{Bank runs, funding liquidity and DeFi fragility.}\leavevmode\\
The theory of run risk in credit intermediaries originates with \citet{DiamondDybvig1983}, who show that maturity transformation with demand deposit contracts generates multiple equilibria in which a run equilibrium coexists with a good one, and that deposit insurance or lender of last resort access can eliminate the run equilibrium. Global games refinements by \citet{MorrisShin1998} and \citet{GoldsteinPauzner2005} demonstrate that under incomplete and heterogeneous private information, a unique equilibrium is selected by a threshold strategy in private signals. An important implication for the paper is that DeFi's near common information environment (where all onchain states are publicly observable) removes the belief heterogeneity that dampens run incentives in Morris-Shin equilibria, sharpening run dynamics relative to traditional intermediaries.

\citet{BrunnermeierPedersen2009} formalize the interaction between funding liquidity and market liquidity, showing that asset liquidations can drive price declines that further tighten funding constraints, producing a liquidity spiral. This mechanism has a direct DeFi analog in the interaction between liquidation induced depth depletion and utilization driven withdrawal constraints, formalized in Propositions~\ref{prop:recovery_endogeneity} and~\ref{prop:liquidity_run}. \citet{GortonMetrick2012} document run dynamics in securitized banking and the repo market, demonstrating that short term secured credit structures are vulnerable to sudden confidence crises even absent fundamental insolvency, a pattern with structural parallels to utilization driven DeFi withdrawal dynamics. \citet{KlagesMundtHarzGudgeon2020} develop risk based models for decentralized stablecoins, providing a theoretical treatment of reflexive collateral dynamics and fragility thresholds in DeFi systems that directly informs the boundary hitting formulation of V3. 

\citet{AppelGrennan2025} provide direct empirical evidence on depositor run dynamics in a crypto shadow bank, linking individual transaction records from the Celsius Chapter~11 bankruptcy filings to onchain blockchain histories for 492,000 depositors and exploiting the June 2022 withdrawal freeze as a quasi-experimental shock in a difference in differences design. Institutional depositors exhibit stranded asset ratios 1.5-4.7 percentage points lower than comparable retail users and cultural heritage (embedded as preference and belief shifters in an extension of the \citet{MorrisShin1998} global games framework) shifts withdrawal probabilities by 3-5 percentage points, with depositors from more individualistic societies exiting earlier and those from high power distance societies deferring to CEO reassurances despite adverse signals. Users with stranded balances subsequently increase high risk token allocations and concentrate portfolios, consistent with gambling for resurrection behavior under reference point dependent utility. Celsius operated as a centralised custodial lender without onchain liquidation enforcement. The cultural monitoring mechanism requires name level depositor identification that is unavailable in pseudonymous vault settings and the absence of smart contract automated liquidation makes its run dynamics structurally closer to a shadow bank than to a utilization constrained DeFi vault. The monitoring incentive findings\footnote{ Institutional sophistication and balance size predict early exit, and that post loss depositors shift toward concentrated risk} are directly relevant to the depositor heterogeneity assumptions underlying the V3 boundary hitting specification.

\paragraph{MEV, blockspace competition and AMM execution.}\leavevmode\\
\citet{DaianGoldfeder2020} introduce the concept of miner extractable value, documenting systematic front-running, sandwich attacks and priority fee competition in decentralized exchange environments and characterizing the adverse selection implications for execution quality under congestion. Their empirical and game theoretic analysis provides the microfoundation for Proposition~\ref{prop:network_congestion}, which formalizes congestion driven execution failure as a wrong way risk channel endogenous to stress. \citet{MilionisMoallemiRoughgardenZhang2022} formalize a continuous time loss channel for AMM liquidity providers, termed loss versus rebalancing, which accumulates from adverse selection by informed traders rather than from discrete liquidation events. This result identifies a sixth loss channel that applies to vaults deploying capital into AMM positions and is outside the scope of V1-V5 as developed here, the scope boundary is stated precisely in the paper's conclusion.

\paragraph{Machine learning approaches to credit risk calibration and liquidation prediction.} \leavevmode\\
Recent work has applied ML to two adjacent problems that contextualise the vault paper's structural approach. \citet{GarciaCespedesMoreno2022} show that gradient boosted tree models can replicate the Monte Carlo simulated loss distribution percentiles of the generalized Vasicek credit risk model for concentrated and multifactor portfolios with $R^2 \approx 100\%$ on held out test samples, using just two inputs the confidence level and a Gaussian copula based concentration estimate $GCR(p)_k = GP(p)_k/ASRF(p)_k$. Their result that a closed form two variable summary statistic captures nearly all variation in the full simulation output is directly relevant to VCS implementation efficiency. Once V1-V5 are structurally defined, ML surrogates trained on simulation output can accelerate real time scoring without sacrificing structural grounding, much as the Gaussian copula concentration estimate accelerates Vasicek loss quantile computation. Their finding that single trees outperform neural networks in the granular case with substantially lower calibration time also suggests that the interpretability, accuracy trade off in complex credit risk settings favors simpler nonlinear models over deep architectures. 

\citet{PalaiokrassasScherrersMakriTassiulas2024} apply ML classification to wallet level liquidation prediction across 23 DeFi protocols and 550,000 Ethereum addresses covering May 2019 to March 2023, achieving AUC scores of 0.847 (CatBoost) and 0.837 (LightGBM) on a six month forward liquidation horizon using 42 monthly behavioral features derived from borrow, deposit, repay, and withdrawal transaction statistics. Their results confirm that transactional behavioral patterns carry predictive content for individual borrower liquidation risk. The unit of analysis is, however, the borrower wallet, predicting whether a given address will cross $\mathrm{HF} < 1$ within six months addresses a different measurement problem from the vault level depositor loss framework developed here. The paper's ACR, recovery, and utilization metrics target the joint distribution of borrower exposures combined with execution constraints that determine whether liquidation proceeds suffice to cover depositor claims, a structural question that AUC based liquidation prediction does not address. The AUC of 0.847 implies residual prediction uncertainty that the present framework attributes structurally to endogenous collateral price dynamics, depth depletion and gas regime shifts rather than to behavioral features.

\citet{GhoshDatta2025} develop the onchain Credit Risk (OCCR) Score, a wallet level probability of default decomposed into five subscores, i.e., historical liquidation rate, current liquidation-at risk simulation based default probability, credit utilization, onchain transaction score, and new credit activity, combined via the weighted linear form $\mathrm{OCCR} = 0.35\hat{s}_{h} + 0.25\hat{s}_{c} + 0.15(1 - \hat{s}_{cu}) - 0.15\hat{s}_{ct} + 0.10\hat{s}_{nc}$. Under standard moment conditions, each subscore is proved consistent and asymptotically normal. Simulation validation on synthetic data yields coverage probabilities of 0.946-0.992, confirming estimator reliability. The OCCR architecture is the closest existing formal analog to the VCS decomposition developed in this paper. Both frameworks decompose composite credit risk into weighted, independently interpretable subscores with established statistical properties and both treat credit utilization and liquidation-at risk as distinct measurement dimensions. The fundamental distinction is that OCCR targets the borrower wallet, i.e., estimating the probability that a specific address will face liquidation, whereas the VCS targets the vault, estimating the probability and magnitude of depositor claim impairment after protocol level liquidations are executed. Wallet level PD estimates of the kind OCCR supplies are inputs to the health factor distribution underpinning V1 and V2 in the present framework. They do not, however, capture oracle divergence, endogenous execution shortfall or cascade induced recovery deficits that determine vault level depositor loss. The OCCR framework assumes independent loan positions, which does not hold in a pooled vault where positions share oracle exposure and correlated liquidation dynamics. Thus, wallet level PD is a necessary but not sufficient input for vault level credit risk measurement.
The ML based and structural measurement approaches are therefore complementary. ML classification can augment real time monitoring of individual borrower risk trajectories, while the structural framework provides mechanism grounded, depositor oriented loss metrics and stress scenario conditioning that borrower level classifiers do not supply.

\section{Theoretical framework}
\label{sec:theory}

\subsection{DeFi lending environment and notation}
DeFi lending protocols implement collateralized borrowing and pooled deposit claims through smart contracts deployed on public blockchains. A vault pools depositor assets (the debt side) and extends loans to borrowers against posted collateral, subject to pre-specified eligibility, collateralization constraints and liquidation rules. Borrowers post collateral assets and incur debt obligations in borrowed assets. If a borrower’s collateral value falls sufficiently, liquidators can repay debt and seize collateral under protocol rules. Deposit claims are typically redeemable subject to available liquidity. When most deposits are lent out (high utilization), withdrawals can be delayed or blocked. Unlike traditional intermediaries, all state variables are publicly observable onchain (up to address attribution) and enforcement is rule bound.

Let $v$ denote a vault with the collateral universe $\mathcal{C}^{v}$ and the debt (borrowed) asset universe $\mathcal{D}^{v}$. For each asset $a$, let $P_{a,t}$ denote the oracle reported unit of account price used in protocol accounting and let $\widetilde{P}_{a,t}$ denote the realized liquidation execution price net of slippage and fees. The distinction is important since $P_{a,t}$ is the valuation input used to compute solvency constraints, while $\widetilde{P}_{a,t}$ is the realization that determines proceeds and recoveries when liquidation is executed.

\begin{definition}[Vault share price and depositor loss equivalence]
\label{def:erc4626}
In a pooled vault implementing the ERC-4626 standard, let $A_t$ denote total underlying assets at time $t$ and let $S_t$ denote total outstanding shares. The share price is:
\begin{equation}
P_t^{\mathrm{share}} \equiv \frac{A_t}{S_t}
\label{eq:share_price}
\end{equation}
Depositor liabilities equal $L_t^{v} = P_t^{\mathrm{share}} \cdot S_t \equiv A_t$ at oracle valued marks. For a depositor entering at share price $P_0^{\mathrm{share}}$, the loss rate at redemption time $T$ is:
\begin{equation}
\ell_T = \left(1 - \frac{P_T^{\mathrm{share}}}{P_0^{\mathrm{share}}}\right)^{\!\!+}
\label{eq:share_loss}
\end{equation}
consistent with Definition~\ref{def:loss_function}, where a vault shortfall $\Delta_t^{v} > 0$ implies that realized liquidation value assets are insufficient to provide redemptions at the current oracle marked share price. Depositors who redeem during or after the shortfall event receive less than the prevailing oracle marked value. It is distinct from a capital loss relative to entry price $P_0^{\mathrm{share}}$, which depends on the depositor's entry timing. Autocompounding vaults accrue yield by increasing $A_t$ faster than deposit inflows increase $S_t$, so $P_t^{\mathrm{share}}$ is strictly increasing absent strategy losses.
\end{definition}

Under ERC-4626 accounting, $L_t^{v}$ is the book value claim at oracle marks, while
$A_t^{v}$ in Definition~\ref{def:loss_function} is the realized liquidation value asset pool available to meet withdrawals. Hence $\Delta_t^{v}=(L_t^{v}-A_t^{v})^+>0$
if and only if $A_t^{v}<A_t^{v,\mathrm{book}}$, where $A_t^{v,\mathrm{book}}\equiv L_t^{v}$ denotes vault assets valued at current oracle marks, i.e., realized execution value assets fall below book value assets. This is precisely the oracle execution divergence formalized in
Proposition~\ref{prop:oracle_failure} and provides the accounting link from share price dynamics to V1.

Let $u\in\mathcal{U}_v$ denote borrower accounts and $C_{u,a,t}$ denote the quantity of collateral asset $a$ held by account $u$ at time $t$. Let $B_{u,d,t}$ denote the quantity of debt asset $d$ owed by $u$. Each collateral asset $a$ has a collateral factor $w_a^{v}\in(0,1]$ that limits borrowing capacity relative to oracle valuation with liquidation thresholds and incentives collected into a vault parameter vector $\Theta_t^{v}$. 
\begin{definition}[Protocol bound vault]
\label{def:protocol_bound}
A vault $v$ is protocol bound if all its strategies are deployed within a single underlying protocol $\Pi$, so that $\mathcal{C}^{v}$ and $\mathcal{D}^{v}$ are subsets of the asset universe of $\Pi$, and all liquidation, oracle, and interest rate mechanics are governed exclusively by $\Pi$'s smart contracts. The dependency graph of $v$ has depth at the vault contract, the underlying protocol contracts and the oracle contracts used to determine collateral prices.
\end{definition}

\begin{definition}[Cross protocol vault]
\label{def:cross_protocol}
A vault $v$ is cross protocol if its strategies deploy capital across multiple distinct protocols $\Pi_1,\ldots,\Pi_k$ or across multiple chains, potentially including AMM liquidity provision, perpetual futures basis trades and tokenized real world asset (RWA) exposures. The dependency graph of $v$ has depth $d\ge 3$ and may include bridges, crosschain messaging layers and external liquidation routes not governed by any single protocol.
\end{definition}

The metrics V1-V5 derived in this paper are for protocol bound vaults, where the collateral set $\mathcal{C}^{v}$, debt set $\mathcal{D}^{v}$, oracle feeds and liquidation engine are governed by a single protocol. For cross protocol vaults, the aggregation of V1-V5 across strategy legs requires assumptions about the dependence of execution failures across protocols and chains that are unlikely to hold in stress, since liquidity depletion and congestion are common shocks. Extension to cross protocol vault measurement involves aggregating per protocol instances under an assumed dependence structure and is a research level extension beyond the present framework. Canonical instances of protocol bound vaults include curated lending vaults built on major lending protocols, e.g., Morpho v1, while canonical instances of cross protocol vaults, e.g., Euler, include multi strategy products spanning lending, AMMs and derivative hedges.

The evolution of onchain states is modeled on a filtration $\{\mathcal{F}_t\}_{t\ge 0}$ representing public information, including onchain events, oracle updates, market data and network conditions. The market/network environment is denoted $\mathcal{M}$.

Several features distinguish this environment from traditional secured lending and structured credit. Price discovery is continuous and public, but mediated by oracle mechanisms that can be stale or manipulated. Liquidation execution is automated and occurs via onchain venues, implying endogenous execution prices and recoveries that depend on trade size and liquidity depth. Execution itself requires blockspace, gas costs and ordering competition are state variables that worsen during stress. Governance constraints can impose timelocks on parameter changes, limiting the ability to respond to fast moving shocks. Finally, there is no lender of last resort, meaning liquidity shortages are resolved (if at all) endogenously through interest rate dynamics, borrower repayments or liquidation.

%Table~\ref{tab:notation} summarizes key notation used throughout.

%\begin{table}[htbp]
%\centering
%\caption{Notation summary for DeFi vault credit risk measurement.}
%\label{tab:notation}
%\begin{tabular}{p{0.22\textwidth} p{0.72\textwidth}}
%\toprule
%\textbf{Symbol} & \textbf{Definition} \\
%\midrule
%$v$ & Vault index (credit instrument). \\
%$u$ & Borrower account index. \\
%$\mathcal{C}^{v}$, $\mathcal{D}^{v}$ & Sets of collateral and debt assets. \\
%$P_{a,t}$, $\widetilde{P}_{a,t}$ & Oracle price and realized execution price for asset $a$. \\
%$C_{u,a,t}$, $B_{u,d,t}$ & Collateral and debt positions (asset units). \\
%$w_a^{v}$ & Collateral factor (LTV ceiling) for collateral asset $a$. \\
%$\Theta_t^{v}$ & Vault parameter vector (thresholds, incentives, caps). \\
%[Notation table — see Appendix~\ref{app:notation}]

\subsection{Vault, loss and the three level structure}
This subsection formalizes the vault as a credit instrument and defines the depositor loss function that anchors the remainder of the paper. The three level decomposition is stated to clarify scope where Level 1 holds $\Theta_t^{v}$ fixed and analyzes mechanical loss channels (the focus of this paper), Level 2 concerns governance risk (companion paper) and Level 3 concerns infrastructure and code integrity (Section~\ref{subsec:level3}). Figure~\ref{fig:three_level} summarizes the decomposition and its scope boundaries.

% ── FIGURE 1: three level Decomposition ─────────────────────────────────
\begin{figure}[htbp]
\centering
\begin{tikzpicture}[
  mybox/.style={
    draw=black, thick, rounded corners=3pt,
    text width=12cm,
    minimum width=13.0cm,
    inner sep=10pt,
    align=justify, font=\small          
  },
  arr/.style={->, >=stealth, semithick},
  brk/.style={
    decorate,
    decoration={brace, amplitude=6pt, raise=5pt},
    thick
  }
]

%% ── Level 1 (bottom, white) ─────────────────────────────────────────────
\node[mybox, fill=white] (L1) at (0,0) {%
  \textbf{Level 1\enspace Mechanical loss channels}\enspace
    \textit{(this paper)}\\[5pt]
  \textbf{P1}\enspace oracle execution divergence
    $\;\Rightarrow\;$ V1 (stress adjusted coverage)\\
  \textbf{P2}\enspace Recovery endogeneity
    $\;\Rightarrow\;$ V2 (volume adjusted expected shortfall)\\
  \textbf{P3}\enspace full information run dynamics
    $\;\Rightarrow\;$ V3 (liquidity stress index)\\
  \textbf{P4}\enspace Timelock/frozen $\Theta_t^v$
    $\;\Rightarrow\;$ frozen-parameter stress methodology\\
  \textbf{P5}\enspace Oracle manipulation/latency
    $\;\Rightarrow\;$ V4 (oracle integrity score)\\
  \textbf{P6}\enspace Congestion wrong way risk
    $\;\Rightarrow\;$ V5 (execution viability)\\[5pt]
  \textit{Metrics V1-V5 and VCS evaluated under fixed
    $\Theta_t^v$ and $\mathcal{I}_t^v = \mathbf{1}$}%
};

%% ── Level 2 (middle, light gray) ───────────────────────────────────────
\node[mybox, fill=gray!10, above=0.45cm of L1] (L2) {%
  \textbf{Level 2\enspace Governance quality}\enspace
    \textit{(companion paper)}\\[5pt]
  The curator selects and updates the parameter vector $\Theta_t^v$
  (LTV, LLTV, collateral caps, liquidation incentives, utilization
  curve). Timelock $\delta^{\mathrm{lock}}$ governs the minimum
  response latency and the critical response window $\Delta^{*v}$
  determines whether governance can adapt before a stress event
  causes mechanical losses. Governance quality also encompasses
  curator incentive alignment, conflict of interest management
  and the risk of strategic parameter miscalibration.\\[4pt]
  \textit{Level 1 outcomes through $\Theta_t^v$,
    treated as a fixed input throughout this paper.}%
};

%% ── Level 3 (top, medium gray) ──────────────────────────────────────────
\node[mybox, fill=gray!20, above=0.45cm of L2] (L3) {%
  \textbf{Level 3\enspace Code integrity}\\[5pt]
  The dependency graph $G^v\!=\!(V^v,E^v)$ comprises all
  critical path nodes including vault contract, underlying protocol contracts, oracle feeds. For cross protocol vaults, also bridge and
  messaging contracts. Each node $i$ carries integrity indicator
  $\mathcal{I}_{t,i}^v\in\{0,1\}$ with the vault wide breach event
  $\mathcal{B}_h^{v,L3}=\{\exists\,s,i:\mathcal{I}_{s,i}^v=0\}$.
  Aggregate failure probability is superadditive in depth
  $k=|V^v|$: $q_{\mathrm{code}}^v(h)\ge
  1-\prod_{i\in V^v}(1-q_i(h))$.
  Cross protocol vaults carry strictly higher Level 3 exposure
  than protocol bound vaults (Corollary~\ref{cor:cross_protocol_L3}).\\[4pt]
  \textit{Level 3 dominates when
    $q_{\mathrm{code}}^v>\mathbb{E}[\ell^{v,L1}]$
    (Proposition~\ref{prop:L3_dominance}).}%
};

%% ── Governance influence arrow ───────────────────────────────────────────
\draw[arr, gray!55] (L2.south) -- (L1.north);
\node[font=\scriptsize\itshape, text=gray!65]
  at ($(L1.north)!0.5!(L2.south) + (3.6,0)$) {sets $\Theta_t^v$};

%% ── Brace spanning all three levels ─────────────────────────────────────
\draw[brk]
  ([xshift=0.75cm]L3.north east) -- ([xshift=0.75cm]L1.south east)
  node[midway, right=5pt, align=center, font=\scriptsize\itshape] {
    Total depositor\\[-2pt]expected loss\\[-2pt]$\mathbb{E}[\ell_T^v\mid\mathcal{F}_t]$
  };

\end{tikzpicture}
\caption{Decomposition of vault credit risk and scope boundaries.
%Level 1 (this paper) formalizes six mechanical loss channels under fixed parameters and functioning smart contracts, yielding five metrics V1-V5 and a Vault Credit Score; Proposition~P4 motivates the frozen-parameter stress test approach rather than a separate metric. Level 2 (companion paper) addresses governance quality and how curator decisions shape the parameter vector $\Theta_t^v$. Level 3 (research agenda) addresses smart contract code integrity via the superadditive failure probability in dependency depth (Propositions~\ref{prop:L3_superadditive}--\ref{prop:L3_dominance}); the dominance condition determines which level is the binding depositor risk concern.
}
\label{fig:three_level}
\end{figure}
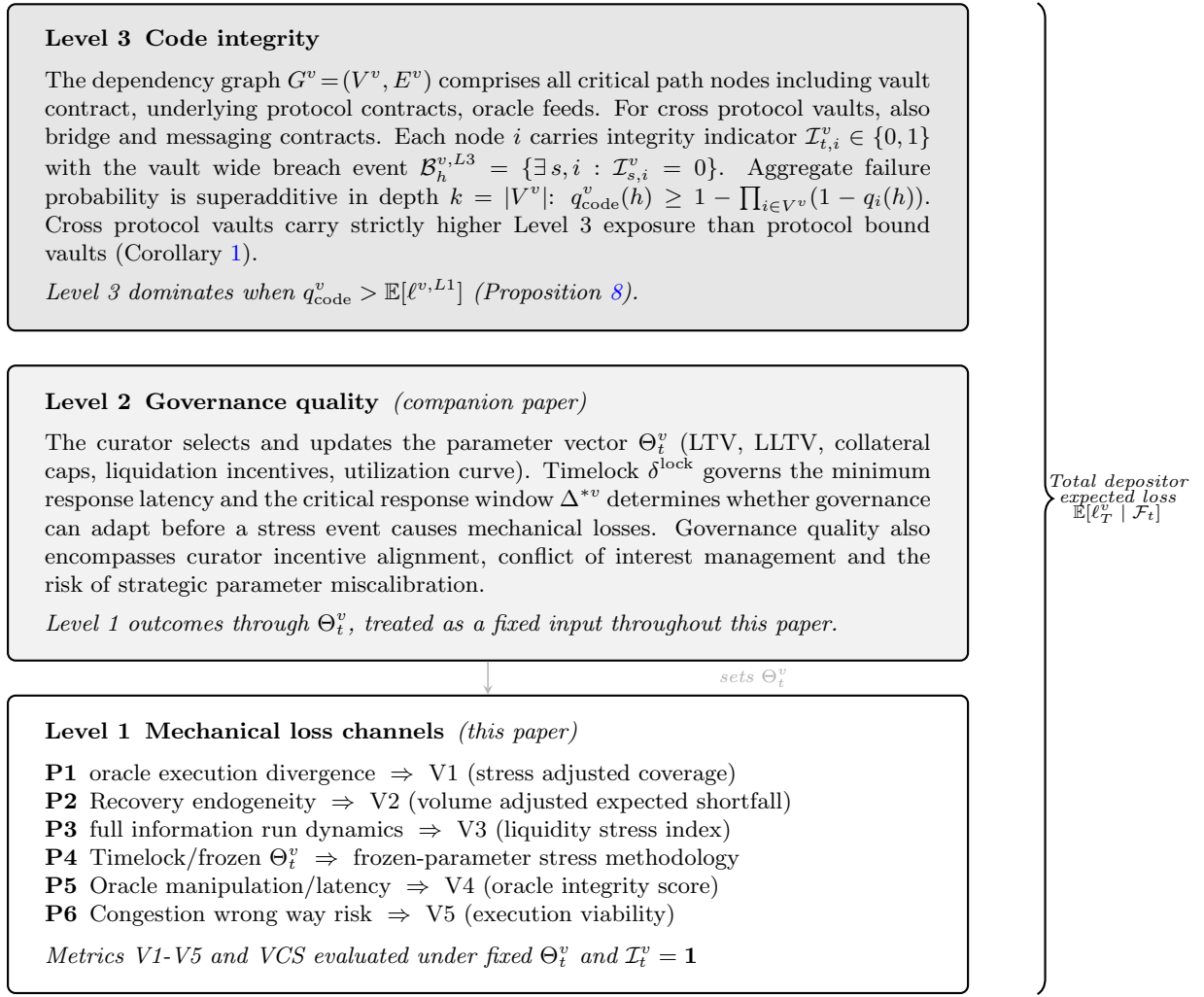

Time is continuous on $t\in[0,T]$, with filtration $\{\mathcal{F}_t\}_{t\ge 0}$ representing publicly observable onchain and market information. A vault $v$ is a rule based lending structure that maps a state vector $S_t^{v}$ into a set of admissible borrower positions, pricing/interest rules, liquidation rules and withdrawal rules. The vault is associated with a finite set of collateral assets $\mathcal{C}^{v}$ and debt (borrowed) assets $\mathcal{D}^{v}$. A curator $k$ selects and updates the vault’s risk parameters and operational dependencies. This paper treats those parameters as fixed inputs and focuses on Level 1 outcomes. This is a point in time evaluation like the internal bank rating systems documented by \citet{TreacyCarey2000}. The VCS is calibrated to the vault's current state and not projected through the cycle default probability. The practical implication is that VCS will exhibit higher time series variation than a through the cycle measure, but provides a more immediate and actionable signal of depositor risk at any given state of the vault.

\begin{definition}[Vault as a credit instrument]
\label{def:level1_state}
The Level 1 credit state at time $t$ is the tuple:
\begin{equation}
X_t^{v} \equiv \left(\mathcal{C}^{v},\mathcal{D}^{v},\{L_{u,t}^{v}\}_{u\in\mathcal{U}_v}, \Theta_t^{v}, \Lambda_t^{v}\right)
\end{equation}
where $\mathcal{U}_v$ is the set of borrower accounts, $L_{u,t}^{v}$ is account $u$'s position (collateral and debt holdings) in asset units, $\Theta_t^{v}$ is the vector of liquidation and solvency parameters (liquidation thresholds, close factors, liquidation incentives, collateral caps) and $\Lambda_t^{v}$ is the vault liquidity state (utilization, available cash, withdrawal queue state). The Level 1 risk of vault $v$ is the distribution of depositor losses induced by the evolution of $X_t^{v}$ under the vault's rule set and market and network dynamics, with depositor loss defined relative to promised withdrawal claims.
\end{definition}

\begin{definition}[Vault loss event and loss function]
\label{def:loss_function}
Let $A_t^{v}$ denote the vault's realized liquidation value assets available to satisfy depositor claims at time $t$ and let $L_t^{v}$ denote the vault's outstanding depositor liabilities including accrued interest. A vault shortfall occurs at time $t$ if:
\begin{equation}
\Delta_t^{v} \equiv \left(L_t^{v}-A_t^{v}\right)^{+} > 0
\end{equation}
and the corresponding loss rate is:
\begin{equation}
\ell_t^{v} \equiv \frac{\Delta_t^{v}}{L_t^{v}} \in [0,1]
\end{equation}
A loss function for horizon $T$ is any mapping $\mathcal{R}_T^{v} = \rho\!\left(\{\ell_t^{v}\}_{t\in[0,T]}\right)$, where $\rho(\cdot)$ is a risk function, e.g., a tail probability, an expected shortfall or a quantile of terminal loss.
\end{definition}

A separate governance state $Y_t^{(v,k)}$ summarizes authority, incentives, process quality and action history for the decision maker that selects and updates $\Theta_t^{v}$. A distinct infrastructure state $\mathcal{I}_t^{v}$ records whether the vault's dependency graph of smart contracts is executing as specified. This paper develops Level 1 metrics and formally delineates the scope boundaries of all three levels. Level 2 (governance) and the Level 3 estimation architecture are developed in companion work.

\begin{definition}[Three level decomposition of vault risk]
\label{def:three_level}
Let $\mathcal{M}$ denote the exogenous market and network environment, $\mathcal{E}^{v}$ the vault rule set (smart contracts implementing accounting, interest, liquidation and withdrawal mechanics), and $\mathcal{I}_t^{v}$ the binary code integrity state vector indicating whether each critical path node in the vault's dependency graph is executing as specified. The three level decomposition asserts:
\begin{equation}
\mathcal{R}_T^{v} = \mathcal{F}\!\left(\mathcal{E}^{v}, \mathcal{M};\, X_0^{v}\right)\cdot\mathbf{1}\!\left\{\mathcal{I}_s^{v}=\mathbf{1}\;\forall s\right\}
+ \mathcal{R}_T^{v,L3}\cdot\mathbf{1}\!\left\{\exists s:\mathcal{I}_s^{v}\neq\mathbf{1}\right\}
\label{eq:three_level}
\end{equation}
where $X_0^{v} = \Psi(u_0, G^{(v,k)}, I^{(v,k)}, \Pi^{(v,k)}, H^{(v,k)})$ and $\mathcal{R}_T^{v,L3}$ is the loss realized when at least one critical path node fails. Level 1 concerns $\mathcal{F}$ with $X_0^{v}$ fixed and $\mathcal{I}_s^{v}=\mathbf{1}$ for all $s$. Level 2 concerns $\Psi$ and the feasibility and latency of updating $X_t^{v}$ through governance. Level 3 concerns $\mathcal{I}_t^{v}$ with whether the contracts execute as specified.
\end{definition}

In DeFi, borrower default is a mechanical liquidation trigger region and recovery is endogenous to execution.

\begin{definition}[Liquidation trigger event (borrower level)]
\label{def:liq_trigger}
For a fixed account $u$ in vault $v$, let $C_{u,t}$ and $B_{u,t}$ denote collateral and debt vectors in asset units, and let $\Theta_t^{v}$ denote the vault's solvency parameters. Define the health factor as the deterministic solvency function $\mathrm{HF}_{u,t}^{v}=\mathrm{HF}(C_{u,t},B_{u,t},P_t;\Theta_t^{v})$, with liquidation enabled when $\mathrm{HF}_{u,t}^{v}<1$. The event $\{\mathrm{HF}_{u,t}^{v}<1\}$ is the liquidation trigger event, meaning an endogenous stopping region condition driven by $P_t$ and by parameter choices in $\Theta_t^{v}$.
\end{definition}

\begin{definition}[Endogenous recovery and shortfall (account level)]
\label{def:recovery}
For a liquidation of account $u$ executed over an interval $[t,t+\delta]$, define realized recovery in unit of account as:
\begin{equation}
\mathcal{R}_{u,t}^{v} \equiv \sum_{a\in\mathcal{C}^{v}} \widetilde{P}_{a,t}\cdot\Delta C_{u,a,t} - \mathrm{Fees}_{u,t}
\end{equation}
where $\Delta C_{u,a,t}$ denotes the amount of collateral seized and liquidated, $\widetilde{P}_{a,t}$ denotes the realized execution price (net of slippage) and $\mathrm{Fees}_{u,t}$ aggregates gas/MEV/protocol fees in unit of account. The account level shortfall is:
\begin{equation}
\Delta_{u,t}^{v} \equiv \left(\sum_{d\in\mathcal{D}^{v}} P_{d,t}\cdot\Delta B_{u,d,t} - \mathcal{R}_{u,t}^{v}\right)^{+}
\end{equation}
with $\Delta B_{u,d,t}$ denoting the debt repaid/covered. The key DeFi specific feature is that $\widetilde{P}_{a,t}$ and $\mathrm{Fees}_{u,t}$ are state dependent and typically worsen in stress.
\end{definition}

\subsection{TradFi to DeFi mapping}
\label{subsec:tradfi_defi}
This subsection is the theoretical core of the paper. Six canonical TradFi credit concepts are mapped into their DeFi vault analogs and for each a formal breakdown condition is derived as a proposition. The breakdown conditions are primitives from which the Level 1 metrics are constructed. Each proposition isolates a structural failure mechanism absent from naive application of traditional credit risk methodology. Table~\ref{tab:mapping_summary} summarizes the mapping and serves as an index for the derivations in Section~\ref{sec:metrics}.

In structured credit, asset coverage is formalized through overcollateralization tests that compare a transaction's collateral value to rated liabilities, with cashflow redirection or deleveraging triggered upon breach. Let $V_t$ denote the marked value of collateral and $D_t$ denote the outstanding debt of a given tranche stack, an archetypal coverage ratio is:
\begin{equation}
\mathrm{CR}_t \equiv \frac{V_t}{D_t}
\end{equation}
The mechanism relies on the structural assumptions that collateral valuation is reliable at the testing frequency, that liquidation is feasible at (or near) the marked value up to specified haircuts and that trigger remedies (cash diversion, amortization) can be implemented without binding settlement frictions. A direct CLO analogue is the OC test logic described in rating agency methodologies and structured finance criteria.
%\citep{MoodysCLO, SPCLO}.

In DeFi lending vaults, the borrower level analog is the health factor $\mathrm{HF}_{u,t}^{v}$ (liquidation triggered when $\mathrm{HF}_{u,t}^{v}<1$), while the vault level analog is an asset coverage ratio defined on liquidation value assets available to depositors (Definition~\ref{def:loss_function} and the shortfall definition therein). Let $A_t^{v}$ denote liquidation value assets and $L_t^{v}$ denote depositor liabilities, the vault asset coverage ratio is:
\begin{equation}
\mathrm{ACR}_t^{v} \equiv \frac{A_t^{v}}{L_t^{v}}
\label{eq:vault_acr}
\end{equation}
The analogy holds because both $\mathrm{CR}_t$ and $\mathrm{ACR}_t^{v}$ are coverage tests. Insolvency is mechanically precluded when the coverage ratio remains above one under the relevant valuation convention.

\begin{assumption}[A0: Execution at or below oracle price]
\label{ass:A0}
In a liquidation event, realized execution prices satisfy $\widetilde{P}_{a,t} \leq P_{a,t}$ almost surely for all collateral assets $a \in \mathcal{C}^{v}$, so that $\varepsilon_{a,t} \equiv 1 - \widetilde{P}_{a,t}/P_{a,t} \in [0,1)$. This holds when collateral is liquidated under competitive onchain depth pressure, which drives execution prices weakly below oracle marks. Assumption~\ref{ass:A0} can be relaxed if routing achieves above-oracle prices; in that case $\varepsilon_{a,t} < 0$ and $\widetilde{\mathrm{ACR}}_t^{v} > \mathrm{ACR}_t^{v}$, which does not require separate treatment for depositor protection.
\end{assumption}

\begin{remark}[Interaction of Assumption~\ref{ass:A0} with execution non viability]
\label{rem:A0_V5_interaction}
Assumption~\ref{ass:A0} conditions on a liquidation event occurring. When Proposition~\ref{prop:network_congestion} applies and liquidation viability $\mathcal{V}_{u,t}$ fails, no execution takes place and $\widetilde{P}_{a,t}$ is undefined. In that case, $\widetilde{P}_{a,t}$ should be interpreted as the counterfactual execution price conditional on viability and V1--V4 metrics quantify the loss conditional on liquidation proceeding. The unconditional probability of execution non viability is separately captured by V5. This separation ensures that V1--V4 measure the quality of execution when it occurs, while V5 measures the probability that execution does not occur at all.
\end{remark}

\begin{proposition}[Oracle execution divergence under stress]
\label{prop:oracle_failure}
Under Definition~\ref{def:recovery} and the vault shortfall definition in Definition~\ref{def:loss_function}, suppose that liquidation value assets $A_t^{v}$ are computed using oracle prices $P_{a,t}$, while realized proceeds are governed by execution prices $\widetilde{P}_{a,t}$. Define the collateral weighted expected oracle deviation:
\begin{equation}
\bar{\varepsilon}_t^{v} \equiv \sum_{a\in\mathcal{C}^{v}} \omega_{a,t}^{v}
\left(1-\frac{\widetilde{P}_{a,t}}{P_{a,t}}\right)
\qquad
\omega_{a,t}^{v} \equiv \frac{P_{a,t}\cdot C_{a,t}^{v}}{\sum_{b\in\mathcal{C}^{v}} P_{b,t}\cdot C_{b,t}^{v}}
\end{equation}
where $C_{a,t}^{v}$ denotes the vault wide collateral quantity in asset $a$. The effective coverage ratio is:
\begin{equation}
\widetilde{\mathrm{ACR}}_t^{v} \equiv \mathrm{ACR}_t^{v}\cdot\big(1-\bar{\varepsilon}_t^{v}\big)
\end{equation}
A vault shortfall can arise despite an apparent oracle based cushion when:
\begin{equation}
\bar{\varepsilon}_t^{v} > 1-\frac{1}{\mathrm{ACR}_t^{v}}
\end{equation}
Credit risk measurement must therefore treat $\bar{\varepsilon}_t^{v}$ as a first order state variable (endogenous to stress) and report coverage in terms of $\widetilde{\mathrm{ACR}}_t^{v}$.
\end{proposition}

The deviation $\bar{\varepsilon}_t^{v}$ is not a static haircut, typically increases in liquidation pressure and network frictions, implying wrong way risk between the severity of liquidation triggers and execution quality (see Remark~\ref{rem:A0_V5_interaction} for the interaction with execution non viability under congestion). When multiple accounts simultaneously cross $\mathrm{HF}_{u,t}^{v}<1$, both market depth and blockspace become scarce, mechanically worsening $\widetilde{P}_{a,t}/P_{a,t}$. This stress endogeneity is absent from standard TradFi coverage tests, where bid--ask spreads and liquidation capacity are often treated through ex ante haircuts rather than as a state dependent execution process. For Level 1 measurement, coverage must be stress conditioned and cannot be summarized by a single point in time oracle mark. Calibrating the function dependence of $\bar{\varepsilon}_t^{v}$ on market and network states should be protocol and venue specific.

In classical portfolio credit models, loss given default (LGD) is frequently treated as an exogenous parameter (or drawn from a distribution) that is independent of the default event conditional on systematic factors, enabling tractable loss distribution calculations. For an exposure with exposure at default $\mathrm{EAD}$, the realized loss is commonly written as:
\begin{equation}
\mathrm{Loss} = \mathrm{EAD}\cdot \mathrm{LGD}\cdot \mathbf{1}\{\tau\le T\}
\end{equation}
where $\tau$ is default time. This construction relies on the structural assumption that recovery is primarily determined by legal and asset liquidation processes not mechanically coupled to contemporaneous default mass in a way that breaks tractability, and that market impact of collateral liquidation is second order at the exposure scale. A canonical asymptotic loan portfolio treatment is \citet{Vasicek1987} single factor framework. \citet{Gordy2000} analysis of CreditMetrics and CreditRisk$^{+}$ further shows that the shape of the systemic factor distribution (not merely its mean and variance) determines whether tail losses under stress are bounded or diverge. This is directly analogous to the vault setting, where the kurtosis of the liquidation mass distribution governs whether expected shortfall scales linearly or superlinearly with aggregate collateral shock severity. \citet{AltmanRestiSironi2004} document that this independence assumption is universal across the major credit VaR architectures. CreditMetrics, CreditRisk$^{+}$, CreditPortfolioView and KMV each treat recovery either as a fixed parameter or as a beta distributed variable with zero correlation to default intensity. This design choice persists despite two bodies of contrary evidence they review. \citet{Merton1974} framework implies analytically that PD and RR are inversely related through firm value (higher leverage increases PD while reducing expected recovery and higher asset volatility does likewise), and empirical studies across US bond markets report negative correlations between aggregate default rates and recovery rates of approximately $-0.22$. When the true negative PD-RR correlation is present but the independence assumption is maintained, expected loss and VaR measures are underestimated by approximately 30 per cent. The paper's Proposition~\ref{prop:recovery_endogeneity} below identifies the DeFi specific mechanism, i.e., execution price impact through AMM depth depletion during liquidation clusters, that induces a structural analog of this negative correlation in an environment where recovery is determined in seconds and where there is no post reorganization process to partially restore creditor value.

In DeFi, recovery is the realized liquidation proceeds net of execution frictions (Definition~\ref{def:recovery}) and, therefore, an endogenous function of liquidation volume and available onchain liquidity. For account $u$ in vault $v$, define the realized recovery rate as:
\begin{equation}
\mathrm{RR}_{u,t}^{v} \equiv
\frac{\mathcal{R}_{u,t}^{v}}{\sum_{d\in\mathcal{D}^{v}} P_{d,t}\cdot\Delta B_{u,d,t}}
\qquad \mathrm{RR}_{u,t}^{v}\in(-\infty,1]
\label{eq:recovery_rate}
\end{equation}
where $\mathcal{R}_{u,t}^{v}$ denotes the realized recovery in unit of account and the denominator is the oracle valued debt repaid in the liquidation window. The analogy holds because $\mathrm{RR}_{u,t}^{v}$ is the DeFi counterpart to $(1-\mathrm{LGD})$. However, unlike in TradFi, $\mathrm{RR}_{u,t}^{v}$ is state dependent and co-varies with liquidation intensity. Negative recovery ($\mathrm{RR}<0$) occurs when execution costs exceed gross liquidation proceeds. The vault shortfall function $\Delta_t^{v} = (L_t^{v} - A_t^{v})^{+}$ absorbs this case automatically since $A_t^{v}$ is then reduced by the excess cost.

\begin{proposition}[Endogenous recovery and liquidation mass]
\label{prop:recovery_endogeneity}
Under Definition~\ref{def:recovery}, suppose collateral liquidations clear through an onchain venue with price impact. Let $Q_{a,t}^{\mathrm{liq}}$ denote aggregate liquidation sell volume in collateral asset $a$ over a short window around $t$ and let $D_{a,t}^{\mathrm{pool}}$ denote the corresponding effective pool depth. Assume execution prices satisfy:
\begin{equation}
\widetilde{P}_{a,t} = P_{a,t} - f\!\left(Q_{a,t}^{\mathrm{liq}}, D_{a,t}^{\mathrm{pool}}\right)
\qquad
\frac{\partial f}{\partial Q}>0,\quad \frac{\partial f}{\partial D}<0
\label{eq:impact_function}
\end{equation}
for a deterministic impact function $f(\cdot,\cdot)$. Let $\mathcal{D}_t^{v} \equiv \sum_{u}\sum_{d\in\mathcal{D}^{v}} P_{d,t}\cdot\Delta B_{u,d,t}$ denote vault wide debt liquidated in the window, then generically:
\begin{equation}
\frac{\partial\,\mathrm{RR}_{u,t}^{v}}{\partial\,\mathcal{D}_t^{v}} < 0
\end{equation}
so that conditional on a liquidation cluster, recoveries deteriorate as liquidation mass increases. Therefore, vault level expected shortfall is superlinear in liquidation mass and credit risk measurement must model the joint distribution of liquidation clustering and depth.
\end{proposition}
Liquidation incentives, e.g., a liquidation bonus $\pi$, are intended to offset execution costs and induce participation.

Liquidation incentives, e.g., a liquidation bonus $\pi$, are intended to offset execution costs and induce participation. In TradFi, analogous incentives for distressed asset sales are set by a trustee or servicer with discretion over timing and route. In DeFi they are fixed protocol parameters embedded in $\Theta_t^v$, applied uniformly regardless of market depth or network congestion at the time of execution. This creates a structural tension since a bonus calibrated to attract liquidators under normal conditions may be insufficient under stress, when gas costs and slippage are elevated, while a bonus generous enough to guarantee participation under stress overcompensates liquidators at the expense of vault solvency. The indifference condition derived below formalizes this trade off.
\begin{proof}[Proof of Proposition~\ref{prop:recovery_endogeneity}]
By \eqref{eq:impact_function}, $\widetilde{P}_{a,t} = P_{a,t} - f(Q_{a,t}^{\mathrm{liq}}, D_{a,t}^{\mathrm{pool}})$ with $\partial f/\partial Q > 0$. An increase in $\mathcal{D}_t^{v}$ increases aggregate liquidation volume $Q_{a,t}^{\mathrm{liq}}$ (more debt is being liquidated, requiring more collateral to be sold), which increases $f(\cdot)$, which in turn strictly decreases $\widetilde{P}_{a,t}$. By \eqref{def:recovery}, this strictly decreases $\mathcal{R}_{u,t}^{v}$, and by \eqref{eq:recovery_rate} this decreases $\mathrm{RR}_{u,t}^{v}$, establishing $\partial \mathrm{RR}_{u,t}^{v}/\partial \mathcal{D}_t^{v} < 0$. For the superlinearity claim, since $A_t^{v}$ is decreasing in each execution deviation $f(Q_{a,t}^{\mathrm{liq}},D_{a,t}^{\mathrm{pool}})$ and the shortfall $\Delta_t^{v}=(L_t^v - A_t^v)^+$ is convex in $A_t^v$, and $A_t^v$ is decreasing linearly in $\mathcal{D}_t^v$ under the linear impact approximation, $\Delta_t^v$ is convex in $\mathcal{D}_t^v$. Therefore, $\mathbb{E}[\Delta_t^v \mid \mathcal{D}_t^v]$ is superlinear in $\mathcal{D}_t^v$ by Jensen's inequality applied to the convex shortfall function.
\end{proof}

However, they also increase collateral seized per unit debt repaid, which can amplify $Q_{a,t}^{\mathrm{liq}}$ and worsen $f(\cdot)$. A liquidator repays debt $P_{d,t}\cdot\Delta B_{u,d,t}$, receives $(1+\pi)$ times that value in collateral at oracle price, and sells the collateral at execution price $\widetilde{P}_{a,t}$. The exact indifference condition for a representative liquidator facing proportional costs is:
\begin{equation}
\pi \;\ge\; \pi^\ast_t \;\equiv\;
\frac{
  \dfrac{g_t\cdot Q_{t}^{\mathrm{gas}}+\mathrm{MEV}_{t}^{\mathrm{cost}}}
        {P_{d,t}\cdot\Delta B_{u,d,t}}
  +\left(1-\dfrac{\widetilde{P}_{a,t}}{P_{a,t}}\right)
}{
  \dfrac{\widetilde{P}_{a,t}}{P_{a,t}}
}
\;=\;
\frac{C_t + \varepsilon_{a,t}}{1 - \varepsilon_{a,t}}
\label{eq:liquidator_indifference}
\end{equation}
where $g_t$ denotes the gas price, $Q_t^{\mathrm{gas}}$ denotes gas units consumed, $\mathrm{MEV}_{t}^{\mathrm{cost}}$ denotes the opportunity and priority cost of inclusion, $C_t \equiv (g_t\cdot Q_{t}^{\mathrm{gas}}+\mathrm{MEV}_{t}^{\mathrm{cost}})/(P_{d,t}\cdot\Delta B_{u,d,t})$ is the total execution cost per unit debt repaid, and $\varepsilon_{a,t} \equiv 1 - \widetilde{P}_{a,t}/P_{a,t}$ is the realised slippage rate, all terms are state dependent. The denominator $(1-\varepsilon_{a,t})$ reflects that the bonus $\pi$ is itself applied to collateral subject to slippage, so the required bonus is strictly larger than the sum of costs and slippage alone. The threshold $\pi^\ast_t$ should be empirically calibrated by chain, liquidation route and market conditions. The design implication is that liquidation incentives must be evaluated jointly with depth and congestion models.

\begin{remark}[Leverage loops and correlated liquidation clustering]
\label{rem:leverage_loop}
A leverage loop is a recursive borrowing strategy in which a borrower deposits collateral asset $a$, borrows asset $b$, converts $b$ back to $a$, and redeposits, repeating $n$ times to amplify yield. For an initial loan to value ratio $\mathrm{LTV}_0$, the total collateral exposure is:
\begin{equation}
\Lambda(\mathrm{LTV}_0,n)=\sum_{j=0}^{n}\mathrm{LTV}_0^j
= \frac{1-\mathrm{LTV}_0^{n+1}}{1-\mathrm{LTV}_0}
\;\xrightarrow{n\to\infty}\; \frac{1}{1-\mathrm{LTV}_0}
\label{eq:leverage_multiplier}
\end{equation}
For $\mathrm{LTV}_0=0.80$, the theoretical maximum leverage multiplier is $5\times$. Because all iterations share the same underlying collateral value, every recursive layer's health factor deteriorates simultaneously when prices fall. The vault wide consequence is that concentrated leverage loop positions generate a liquidation mass $\mathcal{D}_t^{v}$ that is $\Lambda(\mathrm{LTV}_0,n)$ times larger than the initial collateral notional, amplifying the impact channel in Proposition~\ref{prop:recovery_endogeneity} substantially beyond what single layer borrower analysis would indicate. A key monitoring variable is the correlated leverage ratio:
\begin{equation}
\mathrm{CLR}_t^{v} \equiv
\frac{\sum_{a\in\mathcal{C}^{v},\,b\in\mathcal{D}^{v}}
w_{a,b,t}\cdot\mathrm{Var}(r_{a,t}-r_{b,t})}
{\sum_{a,b} w_{a,b,t}\cdot\mathrm{Var}(r_{a,t})}
\label{eq:CLR}
\end{equation}
where $w_{a,b,t}$ is the share of total vault debt backed by collateral $a$ and borrowed against debt asset $b$, $r_{a,t}$ denotes the log return of asset $a$ over the base interval and $\mathrm{Var}(\cdot)$ denotes variance computed over a rolling window of length $W$.\footnote{$\mathrm{CLR}_t^{v}$ is a diagnostic ratio because $\mathrm{Var}(r_{a,t}-r_{b,t}) = \mathrm{Var}(r_{a,t}) + \mathrm{Var}(r_{b,t}) - 2\,\mathrm{Cov}(r_{a,t},r_{b,t})$, the numerator can exceed the denominator whenever collateral and debt assets have low or negative return correlation, which is the economically relevant case for loop strategies. $\mathrm{CLR}_t^{v}$ should be reported in its natural range and interpreted as a severity signal.} When $\mathrm{CLR}_t^{v}$ is high, the collateral to debt price ratio is volatile, i.e., the spread between collateral and debt asset returns is large and subject to widening in stress, creating endogenous liquidation clustering and amplifying realized slippage. Empirical calibration of the loop amplification channel requires joint estimation of $\mathrm{CLR}_t^{v}$ and liquidation clustering in stress.
\end{remark}

In banking and money like liabilities, funding liquidity risk arises from maturity mismatch. Depositors can withdraw on demand while assets are illiquid, generating run equilibria under strategic complementarities. In the canonical \citet{DiamondDybvig1983} setting, a bank invests in a long asset and offers demandable claims. Runs occur when withdrawals exceed liquid reserves, even if the bank is solvent in present value terms. A reduced form representation of the liquidity constraint is $W_t \le \mathrm{Cash}_t$, with $W_t$ denoting aggregate withdrawals and $\mathrm{Cash}_t$ immediately available liquidity. The framework relies on partial information and belief heterogeneity to support multiple equilibria, and on the availability of lender of last resort facilities and deposit insurance that can eliminate run equilibria \citep{MorrisShin1998, GoldsteinPauzner2005}.

In DeFi lending vaults with pooled liquidity, the direct analog is utilization where withdrawals are feasible only to the extent that liquidity is not fully lent out. Let $B_t^{v}$ denote total borrows of the vault and let $D_t^{v}$ denote total deposits, then the utilization ratio is:
\begin{equation}
U_t^{v} \equiv \frac{B_t^{v}}{D_t^{v}}\in[0,1]
\label{eq:utilization}
\end{equation}
so that immediate withdrawal capacity is decreasing in $U_t^{v}$ and withdrawals are mechanically blocked when $U_t^{v}=1$ (up to reserves and protocol specific buffers). The analogy holds because both banking and DeFi vaults issue short term withdrawal claims against longer duration lending assets, creating funding liquidity exposure.

\begin{proposition}[Full information run dynamics without a backstop]
\label{prop:liquidity_run}
Let depositor $i$ in vault $v$ choose a withdrawal stopping time $\tau_i$ adapted to $\{\mathcal{F}_t\}$. Define depositor $i$'s conditional expected recovery rate on withdrawal at $t$ as $\mathbb{E}[\mathrm{RR}_{i,t}^{v}\mid\mathcal{F}_t]$, where $\mathrm{RR}_{i,t}^{v}$ is induced by the vault wide shortfall process (Definition~\ref{def:loss_function}) and the withdrawal rule set. Consider a threshold policy:
\begin{equation}
\tau_i^\ast \equiv \inf\Big\{t:\ \mathbb{E}[\mathrm{RR}_{i,t}^{v}\mid\mathcal{F}_t] < \underline{r}\Big\}
\label{eq:run_stopping}
\end{equation}
for some depositor specific tolerance $\underline{r}\in(0,1)$. If onchain information symmetry implies $\mathcal{F}_t$ is effectively common across depositors, i.e., $\mathcal{F}_t^{(i)}=\mathcal{F}_t^{(j)}$ for all $i,j$, then the withdrawal decision exhibits strategic complementarity and the probability of a run conditional on a negative information shock is not damped by heterogeneous beliefs:
\begin{equation}
\Pr\!\left(\sum_{i} \mathbf{1}\{\tau_i^\ast \le t\} \ge \kappa \cdot D_t^{v} \,\Big|\, \mathcal{F}_t\right)
\ \text{is non degenerate for large }\kappa
\end{equation}
while there is no external liquidity backstop to relax the mechanical constraint when $U_t^{v}\uparrow 1$. Credit risk measurement must treat run risk as an equilibrium object under common information and quantify losses under withdrawal clustering.
\end{proposition}

\begin{remark}[Common information as a worst case bound]
\label{rem:common_info}
Proposition~\ref{prop:liquidity_run} assumes that onchain information symmetry renders $\mathcal{F}_t^{(i)} = \mathcal{F}_t^{(j)}$ for all depositors $i,j$, which is is a worst case modeling assumption. In practice, depositors differ in technical sophistication, monitoring frequency and capacity to interpret onchain signals, partially restoring the belief heterogeneity that \citet{MorrisShin1998} show dampens run severity. In \citet{AppelGrennan2025} findings on the Celsius collapse, institutional depositors exited 1.5-4.7 percentage points earlier than comparable retail users despite identical onchain information access, consistent with informational and analytical heterogeneity. The implication is that Proposition~\ref{prop:liquidity_run}'s run probability is an upper bound on run severity under common information. In practice, heterogeneity in depositor sophistication provides a partial dampening mechanism not captured by the formal proposition. Therefore, the proposition should be interpreted as establishing the structural worst case and motivating V3 as a conservative liquidity stress metric.
\end{remark}

Utilization based interest rate models with steep kinks can partially mitigate runs by increasing the opportunity cost of remaining borrowed liquidity and incentivizing repayment, thereby reducing $B_t^{v}$ and lowering $U_t^{v}$. However, the effectiveness of this mechanism is limited when borrowers face binding constraints, e.g., they are already near liquidation, or when network congestion raises repayment costs precisely when the kink is activated. Moreover, raising borrow rates can worsen borrower solvency and increase liquidation trigger frequency, which feeds back into recovery endogeneity (Proposition~\ref{prop:recovery_endogeneity}). 
%The stabilizing effect is without protocol specific empirical evidence on repayment elasticity under stress. 
Unlike deposit insurance, rate kinks do not create exogenous liquidity, they only reallocate incentives within the existing onchain balance sheet.

Proposition~\ref{prop:liquidity_run} treats all vault assets not currently
lent out as immediately withdrawable. This assumption fails for vaults holding liquid staking tokens when large scale redemptions require traversing the Ethereum beacon chain validator exit queue. Let $Q_{\mathrm{exit}}(t)$ denote the number of validator exit requests queued at time $t$ and $R_{\mathrm{churn}}$ the protocol's maximum exit processing rate in validators per epoch, then the expected exit duration for a new request is:
\begin{equation}
\tau_{\mathrm{exit}}(t)=\frac{Q_{\mathrm{exit}}(t)}{R_{\mathrm{churn}}}
\quad(\text{in epochs})
\label{eq:exit_queue}
\end{equation}
Under stress, $Q_{\mathrm{exit}}(t)$ grows as stakers compete to exit. At
queue depths of $10^5$ validators and base churn, $\tau_{\mathrm{exit}}$ can exceed multiple weeks. During this window, a vault cannot redeem staked collateral at par through the canonical path and must sell at a
secondary market discount or hold an illiquid position while accruing debt
interest. More generally, define the duration bounded withdrawal capacity
rate as:
\begin{equation}
W_{\max}^{v}(t)
=\min\!\left\{D_t^{v}\cdot\!\left(1-U_t^{v}\right),\;
\sum_{k} f_k \cdot W_{\max,k}(t)\right\}
\label{eq:W_max}
\end{equation}
where $f_k$ is the fraction of vault assets in strategy $k$ and
$W_{\max,k}(t)$ is the maximum per period withdrawal rate from strategy $k$ given its exit queue state at time $t$. When the second argument binds, the vault faces duration induced withdrawal gating even if $U_t^{v}<1$. This channel is absent from V3 estimator in Corollary~\ref{cor:M3_estimator} as currently specified. Calibrating $W_{\max,k}(t)$ requires exit queue state from beacon chain APIs and constitutes an important extension for vaults with significant staked asset holdings.

In CLOs and other managed structured credit vehicles, the collateral manager can trade within eligibility criteria, rebalance exposures and actively maintain coverage tests through discretionary sales, reinvestment and hedging, subject to documentation constraints and trustee oversight. Formally, manager discretion can be represented as the ability to choose a trading control $\xi_t$ (holdings adjustments) continuously over time, so that collateral holdings $H_t$ evolve as:
\begin{equation}
dH_t = \xi_t\,dt + \cdots
\end{equation}
where $\xi_t$ can respond rapidly to market conditions. This relies on the structural assumption that the manager can act with low latency relative to the speed of adverse market moves, and that settlement and market access are sufficiently reliable for rebalancing to be executed when needed. Rating methodologies explicitly treat manager behavior and structural mitigants as determinants of expected loss.

In DeFi vaults, parameter changes are typically subject to a timelock duration $\delta^{\mathrm{lock}}$, implying that $\Theta_t^{v}$ is frozen over economically relevant horizons in stress. While a full measurement of governance quality is outside this paper’s scope, the binding implication for Level 1 is that the parameter vector is piecewise constant and may be unable to adapt within a critical response window.

\begin{proposition}[Timelock constrained response window]
\label{prop:curator_lag}
Under a timelock $\delta^{\mathrm{lock}}>0$ such that any parameter update submitted at time $t$ takes effect only at $t+\delta^{\mathrm{lock}}$, for all $s\in[0,\delta^{\mathrm{lock}})$:
\begin{equation}
\Theta_{t+s}^{v} = \Theta_t^{v}
\label{eq:curator_lag}
\end{equation}
Define the critical response window $\Delta^{*v}$ as the maximum horizon over which a parameter change remains protective against an impending breach, a function of collateral volatility, the distribution of gaps to liquidation thresholds and depth conditions. If $\delta^{\mathrm{lock}} \ge \Delta^{*v}$, then governance interventions cannot prevent mechanically induced liquidation clusters and shortfalls during fast stress. Credit risk measurement should evaluate losses under frozen-parameter dynamics over stress horizons shorter than $\delta^{\mathrm{lock}}$, with $\Delta^{*v}$ requiring empirical calibration based on vault specific collateral volatility and liquidation gap distributions.
\end{proposition}

Timelocks mitigate governance takeovers by slowing malicious parameter changes, but they simultaneously reduce the feasible control set during crises, creating an explicit security agility frontier. The appropriate $\delta^{\mathrm{lock}}$ should be calibrated to the speed of adverse collateral moves and liquidation cascades. In volatile collateral regimes, $\Delta^{*}$ may be measured in hours, whereas common onchain timelocks are often 24-72 hours. This mismatch can be first order for solvency, implying a design to pair longer timelocks with pre authorized emergency controls that are themselves constrained, e.g., only tightening risk parameters to avoid creating an unconstrained admin key failure mode. Although formal governance measurement is outside this paper, the frozen parameter constraint is a binding input to Level 1 stress testing.

In secured lending and structured finance, valuation agents (pricing services, trustee marks, dealer quotes) provide reference prices used for collateral valuation, margining and coverage tests. Let $P_t^{\mathrm{val}}$ denote the valuation agent price and $P_t^{\mathrm{mkt}}$ denote the prevailing market clearing price. The valuation risk is often treated as an operational or model risk overlay, assuming deviations are small and correctable through dispute processes and governance. A stylized deviation is:
\begin{equation}
\eta_t^{\mathrm{TradFi}} \equiv P_t^{\mathrm{val}} - P_t^{\mathrm{mkt}}
\end{equation}
with controls (multiple quotes, independent verification) intended to keep $\eta_t^{\mathrm{TradFi}}$ bounded. This relies on the structural assumptions that markets are not adversarial to the valuation process (manipulation costs are prohibited at required scale) and that human discretion can intervene when marks appear unreliable.

In DeFi, the pricing agent is the oracle network. Protocol logic uses $P_{a,t}$ to compute $\mathrm{HF}_{u,t}^{v}$ and to authorize liquidation (Definition~\ref{def:liq_trigger}). Let $P_{a,t}^{\mathrm{true}}$ denote the latent efficient price (unobservable) and define oracle error:
\begin{equation}
\eta_{a,t} \equiv P_{a,t} - P_{a,t}^{\mathrm{true}}
\end{equation}
The analogy holds because both valuation agents and oracles provide reference prices for collateral and triggers. However, DeFi differs because execution is automatic and adversarial incentives can target oracle inputs.

\begin{proposition}[Oracle manipulation induces false solvency]
\label{prop:oracle_manipulation}
Under Definition~\ref{def:liq_trigger} and Definition~\ref{def:loss_function}, let $P_{a,t}^{\mathrm{true}}$ denote the latent efficient price and define oracle error $\eta_{a,t} \equiv P_{a,t} - P_{a,t}^{\mathrm{true}}$. The latent health factor is $\mathrm{HF}^{\ast\,v}_{u,t}\equiv \mathrm{HF}(C_{u,t},B_{u,t},P_t^{\mathrm{true}};\Theta_t^{v})$. Consider two oracle misrepresentation events:
\begin{enumerate}
\item \emph{False solvency:} $\eta_{a,t}>0$ for relevant collateral so that $\mathrm{HF}^{v}_{u,t}>1>\mathrm{HF}^{\ast\,v}_{u,t}$ and liquidation is not triggered despite latent insolvency.
\item \emph{False insolvency:} $\eta_{a,t}<0$ so that $\mathrm{HF}^{v}_{u,t}<1<\mathrm{HF}^{\ast\,v}_{u,t}$ and liquidation is triggered unnecessarily.
\end{enumerate}
Depositor credit risk is directly increased by false solvency through missed liquidation and latent bad debt accumulation. The oracle induced expected shortfall satisfies:
\begin{equation}
\mathbb{E}\!\left[\Delta_t^{v,\mathrm{oracle}}\right]
\ \ge\
\mathbb{E}\!\left[\Delta_t^{v}\cdot
\mathbf{1}\!\left\{\exists u:\ \mathrm{HF}^{v}_{u,t}>1>\mathrm{HF}^{\ast\,v}_{u,t}\right\}\right]
\end{equation}
Therefore, oracle integrity is a first order credit risk dimension and the measurement should be modeled $\eta_{a,t}$ jointly with liquidation dynamics.
\end{proposition}

Oracle error $\eta_{a,t}$ is plausibly correlated with stress regimes because volatility spikes increase both the economic incentive to manipulate and the probability of stale or lagged oracle updates due to network congestion. This creates wrong way risk where the same market move that pushes $\mathrm{HF}_{u,t}^{v}$ toward one may simultaneously degrade the information quality used to compute it. Even if oracle prices are unbiased on average, conditional on stress the joint event of execution deterioration and oracle error can dominate tail losses (the result interacts with Proposition~\ref{prop:oracle_failure}). Empirical estimation of the conditional distribution of $\eta_{a,t}$ in stress requires protocol and oracle specific incident and latency data. False insolvency signals can trigger forced liquidations that deplete pool depth and elevate gas costs, linking oracle risk to run risk through a second round liquidity channel even when no direct depositor shortfall has occurred.

In TradFi, clearing and settlement risk arises because obligations are fulfilled through settlement infrastructure (CCPs, custodians, payment systems) subject to operational outages, margin calls and settlement lags, e.g., T+2, with default funds and loss allocation rules mitigating member failure. A stylized settlement lag constraint can be written as:
\begin{equation}
\text{Delivery occurs at } t+\Delta^{\mathrm{settle}}
\end{equation}
with $\Delta^{\mathrm{settle}}$ and default management processes assumed to be independent of contemporaneous market stress beyond modeled stress scenarios. The framework relies on the structural assumptions that default management resources (auction processes, default funds, central bank liquidity) are available and that settlement finality is institutionally enforced.

In DeFi, liquidation and repayment execution depends on block inclusion, gas pricing and transaction ordering subject to MEV extraction \citep{DaianGoldfeder2020}. Let $g_t$ denote the prevailing gas price, $Q_{u,t}^{\mathrm{gas}}$ the gas units required to execute liquidation or repayment actions for account $u$ and $\mathrm{MEV}_{u,t}^{\mathrm{cost}}$ the incremental cost from ordering and priority competition. The liquidation trigger event is defined as $\mathcal{L}_{u,t}\equiv\{\mathrm{HF}_{u,t}^{v}<1\}$ and the liquidation viability event $\mathcal{V}_{u,t}$ as the event that a rational liquidator can execute profitably at $t$ given costs and expected proceeds. The analogy holds because both systems depend on an execution infrastructure. However, DeFi differs in that infrastructure costs are endogenous and subject to adversarial competition precisely during stress.

\begin{proposition}[Congestion dependent liquidation failure (wrong way risk)]
\label{prop:network_congestion}
Let $\Pi_{u,t}^{\mathrm{liq}}$ denote the expected gross liquidation profit (before gas and MEV) for account $u$ at time $t$, computed under realized execution prices $\widetilde{P}_{a,t}$ (Definition~\ref{def:recovery}). The liquidation viability event is:
\begin{equation}
\mathcal{V}_{u,t} \equiv
\left\{\Pi_{u,t}^{\mathrm{liq}} \ge g_t\cdot Q_{u,t}^{\mathrm{gas}} + \mathrm{MEV}_{u,t}^{\mathrm{cost}}\right\}
\label{eq:viability}
\end{equation}
congestion driven wrong way risk is present when:
\begin{equation}
\Pr\!\left(\mathcal{V}_{u,t}^{c}\mid \mathcal{L}_{u,t}\right) \;>\; \Pr\!\left(\mathcal{V}_{u,t}^{c}\right)
\end{equation}
that is, liquidation becomes less viable precisely when liquidation is needed. Unconditional liquidation success assumptions systematically understate expected shortfall, and credit risk measurement must condition liquidation execution on $g_t$ and MEV states jointly with price shocks.
\end{proposition}

Keeper networks, liquidation bots and relays can reduce latency but do not eliminate the dependence on blockspace scarcity and ordering competition during market wide deleveraging. Protocols can partially internalize congestion risk via gas subsidies, pre funded keeper incentives, or auction designs that broaden participation, but these mechanisms create new attack surfaces and governance burdens. Dutch auction liquidations can reduce race dynamics relative to fixed bonus designs, yet they may fail when congestion prevents timely bidding, making liquidation duration credit relevant. The implication for this paper is that execution mechanism choice is part of $\mathcal{E}^{v}$ and should enter Level 1 measurement directly through viability and realized recovery.

%\subsubsection*{Summary of Failure Modes}
%\label{subsubsec:mapping_summary}

Table~\ref{tab:mapping_summary} TradFi to DeFi mappings.
\begin{table}[htbp]
\centering
\caption{Summary of TradFi to DeFi mappings and formal failure modes. Each row states the TradFi concept, its DeFi analog, the breakdown condition, and the corresponding proposition.}
\label{tab:mapping_summary}
\begin{tabular}{>{\centering\arraybackslash}p{0.22\textwidth}
                >{\centering\arraybackslash}p{0.22\textwidth}
                p{0.36\textwidth}
                >{\centering\arraybackslash}p{0.09\textwidth}}
\toprule
\textbf{TradFi concept} & \textbf{DeFi analog} & \textbf{Failure mode/breakdown condition} & \textbf{Proposition} \\
\midrule
Asset coverage/OC tests & Vault ACR$_t^v$ and health factor $\mathrm{HF}_{u,t}^v$ & Oracle execution divergence $\bar{\varepsilon}_t^v > 1 - 1/\mathrm{ACR}_t^v$ so $\widetilde{\mathrm{ACR}}_t^v < 1$ despite $\mathrm{ACR}_t^v > 1$ & \ref{prop:oracle_failure} \\[4pt]
Exogenous recovery (LGD) & Realized recovery $\mathrm{RR}_{u,t}^v$ & Recovery endogeneity $\partial\mathrm{RR}_{u,t}^v/\partial\mathcal{D}_t^v < 0$. ES is superlinear in liquidation mass & \ref{prop:recovery_endogeneity} \\[4pt]
Funding liquidity/runs & Utilization $U_t^v$ and withdrawal feasibility & Full information run, common $\mathcal{F}_t$ with no backstop, clustered $\tau_i^\ast$ and $U_t^v \uparrow 1$ blocks withdrawals & \ref{prop:liquidity_run} \\[4pt]
Manager discretion & Frozen parameters under timelock $\delta^{\mathrm{lock}}$ & Parameter rigidity $\Theta_{t+s}^v = \Theta_t^v$ for $s < \delta^{\mathrm{lock}} \ge \Delta^{*v}$ & \ref{prop:curator_lag} \\[4pt]
Valuation agent risk & Oracle error $\eta_{a,t}$ & Manipulation/latency $\mathrm{HF}_{u,t}^v > 1 > \mathrm{HF}^{\ast v}_{u,t}$ generates missed liquidations and expected shortfall & \ref{prop:oracle_manipulation} \\[4pt]
Clearing/settlement risk & Gas/MEV driven viability $\mathcal{V}_{u,t}$ & Congestion wrong way risk $\Pr(\mathcal{V}_{u,t}^c \mid \mathcal{L}_{u,t}) > \Pr(\mathcal{V}_{u,t}^c)$ & \ref{prop:network_congestion} \\
\bottomrule
\end{tabular}
\end{table}

\subsection{Infrastructure and code integrity (Level 3)}
\label{subsec:level3}

The six failure modes formalized above each concern valuation, execution, liquidity or governance frictions within a vault whose smart contracts are functioning as specified. A structurally distinct risk class arises when the contracts themselves fail, e.g., through an undiscovered vulnerability, a successful exploit or an unauthorized modification. This risk operates outside the Level 1 measurement space since when a protocol is drained through a code exploit, no coverage ratio, oracle integrity score or gas stress analysis provides warning, because the loss event is a code failure, not a market event. Three features distinguish Level 3 from Levels~1 and~2. First, Level 3 events are catastrophic in magnitude because exploit events typically drain vault assets in full, so $\mathbb{E}[\ell_T^{v,L3}\mid\mathcal{F}_t]\approx 1$, contrasting with the partial shortfalls Level 1 metrics capture. Second, Level 3 risk cannot be mitigated by governance action at the point of failure since the timelock constraint of Proposition~\ref{prop:curator_lag} means curator responses cannot prevent an exploit once submitted. Third, Level 3 risk accumulates in the dependency depth $k\equiv|V^{v}|$ where each additional protocol layer is an independent failure opportunity.

\begin{definition}[Level 3 infrastructure and code integrity state]
\label{def:level3_state}
For vault $v$, the dependency graph $G^{v}=(V^{v},E^{v})$ has node set $V^{v}$ consisting of all critical path smart contracts including the vault contract, the underlying protocol contracts, oracle contracts for each collateral and debt asset, and, for cross protocol vaults (Definition~\ref{def:cross_protocol}), bridge and messaging contracts. The integrity indicator for node $i$ is $\mathcal{I}_{t,i}^{v}=1$ if node $i$ executes as specified at time $t$, and $\mathcal{I}_{t,i}^{v}=0$ if it has been exploited, contains an activated bug, or has been modified contrary to specification. The Level 3 breach event over horizon $h$ is:
\begin{equation}
\mathcal{B}_h^{v,L3}\equiv\left\{\exists s\in[t,t+h],\,\exists i\in V^{v}:\mathcal{I}_{s,i}^{v}=0\right\}
\end{equation}
and the aggregate code failure probability is $q_{\mathrm{code}}^{v}(h)\equiv\Pr(\mathcal{B}_h^{v,L3}\mid\mathcal{F}_t)$.
\end{definition}

\begin{proposition}[Level 3 superadditive failure probability in dependency depth]
\label{prop:L3_superadditive}
For vault $v$ with $k$ critical path nodes, let $q_i(h)\equiv\Pr(\mathcal{I}_{t,i}^{v}=0$ for some $s\in[t,t+h]\mid\mathcal{F}_t)$. Under independence across nodes:
\begin{equation}
q_{\mathrm{code}}^{v}(h) \;\geq\; 1-\prod_{i=1}^{k}\!\left(1-q_i(h)\right)
\label{eq:sc_mult}
\end{equation}
This lower bound is strictly increasing in $k$ and in each $q_i(h)$. Under homogeneous node probabilities $q_i(h)=q$, the marginal contribution of the $(k+1)$-th node is $(1-q)^{k}\cdot q>0$, strictly positive for all $q\in(0,1)$.
\end{proposition}
\begin{proof}
Let $A_i$ denote the failure event of node $i$ over $[t,t+h]$. Under independence, $\Pr(\bigcup_i A_i)=1-\prod_i(1-q_i(h))$, establishing \eqref{eq:sc_mult} as an equality under independence and a lower bound in general, since protocol wide incidents. e.g., shared oracle contracts, consensus layer bugs, correlated governance events, induce positive cross-node dependence that can only increase $\Pr(\bigcup_i A_i)$. The marginal formula follows from $(1-(1-q)^{k+1})-(1-(1-q)^k)=(1-q)^k\cdot q>0$.
\end{proof}

\begin{proposition}[VCS scope boundary, Level 3 dominance condition]
\label{prop:L3_dominance}
The full depositor expected loss satisfies:
\begin{equation}
\mathbb{E}\!\left[\ell_T^{v}\mid\mathcal{F}_t\right] = \bigl(1-q_{\mathrm{code}}^{v}\bigr)\cdot\mathbb{E}\!\left[\ell_T^{v,L1}\mid\overline{\mathcal{B}}^{v,L3}\right] + q_{\mathrm{code}}^{v}\cdot\mathbb{E}\!\left[\ell_T^{v,L3}\right]
\label{eq:full_loss}
\end{equation}
where arguments $\mathcal{F}_t$ are suppressed for readability, $q_{\mathrm{code}}^{v}\equiv q_{\mathrm{code}}^{v}(T-t)$, and $\ell_T^{v,L1}$ is the Level 1 loss measured by the VCS. Since $\mathbb{E}[\ell_T^{v,L3}]\approx 1$, Level 3 dominates total expected depositor loss whenever:
\begin{equation}
q_{\mathrm{code}}^{v}(T-t) \;>\; \mathbb{E}\!\left[\ell_T^{v,L1}\mid\mathcal{F}_t,\,\overline{\mathcal{B}}_{T-t}^{v,L3}\right]
\label{eq:L3_dominance}
\end{equation}
When \eqref{eq:L3_dominance} holds, improving V1-V5 does not materially reduce total expected loss since code risk assessment is the binding depositor protection concern. Remark~\ref{rem:L3_estimation} below provides reference calibrations for $q_i(h)$.
\end{proposition}
\begin{proof}
\eqref{eq:full_loss} follows from the law of total expectation applied to \eqref{eq:three_level}, conditioning on $\mathcal{B}_{T-t}^{v,L3}$. \eqref{eq:L3_dominance} follows by setting the Level 3 term larger than the Level 1 term and substituting $\mathbb{E}[\ell_T^{v,L3}]\approx 1$.
\end{proof}

\begin{corollary}[Cross protocol vaults have strictly higher Level 3 exposure]
\label{cor:cross_protocol_L3}
For a protocol bound vault $v^{\mathrm{b}}$ (Definition~\ref{def:protocol_bound}) and a cross protocol vault $v^{\mathrm{c}}$ (Definition~\ref{def:cross_protocol}), $|V^{v^{c}}|>|V^{v^{b}}|$ by construction. Under identical node failure probabilities $q_i(h)=q$:
\begin{equation}
q_{\mathrm{code}}^{v^{c}}(h) \;\geq\; q_{\mathrm{code}}^{v^{b}}(h)
\end{equation}
with strict inequality whenever $q>0$. Additional protocol layers, bridge contracts and messaging contracts are structural Level 3 risk increments that cannot be diversified away through collateral or governance quality improvements.
\end{corollary}

\begin{remark}[Estimating $q_i(h)$ from onchain signals]
\label{rem:L3_estimation}
Estimation of $q_i(h)$ for individual contracts requires smart contract security methodology outside the scope of this paper. As reference calibrations, behavioral and code structural exploit prediction models achieve AUC of 0.887 on 220 confirmed security breaches with a lending protocol detection rate of 96.4\% \citep{ParhizkariIannillo2025}, demonstrating that $q_i(h)$ is estimable from onchain signals and that the dominance condition \eqref{eq:L3_dominance} is computationally tractable. A deployed protocol level estimate from Credora places the Morpho Protocol annual smart contract exploit probability at 0.13\% \citep{CredoraMethodology} which is a reference Level 3 point estimate for a well audited, single protocol vault architecture. For cross protocol vaults with bridge dependencies, empirical exploit frequency is substantially higher. Therefore, Level 3 assessment should be treated as a prerequisite for depositor risk evaluation in such vaults, prior to any Level 1 analysis.
\end{remark}

\section{Measurement framework}
\label{sec:metrics}

Level 1 measurement constructs tractable statistics that quantify the depositor loss channels implied by Propositions~\ref{prop:oracle_failure}-\ref{prop:network_congestion}, holding the parameter vector $\Theta_t^{v}$ fixed over the evaluation horizon. Continuous time definitions are used because vault dynamics and failure modes are naturally stated as state contingent events (liquidation trigger regions, execution viability, boundary hitting for utilization) on $\{\mathcal{F}_t\}$. Discrete time estimators are required for operational implementation because the relevant inputs are observed onchain at transaction, block or hourly/daily frequency. Therefore, each is defined as a function $\mathrm{V}j(v,t)=\Phi_j(X_t^{v},\mathcal{M})$ and paired with a discrete estimator designed for practical computation from public data. The five metrics are constructed to be non redundant where V1 targets valuation to execution divergence, V2 targets recovery endogeneity to liquidation mass, V3 targets funding liquidity fragility via utilization dynamics, V4 targets oracle integrity and V5 targets congestion/MEV driven execution failure. The subsection on aggregation defines how these five dimensions are combined into a vault credit score VCS that is interpretable for monitoring and risk limits.

\subsection{Stress adjusted asset coverage}
\label{subsec:V1}

V1 metric represents stress adjusted asset coverage and is motivated by Proposition~\ref{prop:oracle_failure}, which shows that oracle based coverage $\mathrm{ACR}_t^{v}$ can overstate solvency when execution prices diverge from oracle prices in stress. Traditional coverage tests treat valuation haircuts as exogenous, DeFi requires a state dependent correction because liquidation proceeds are generated by endogenous onchain execution. Therefore, V1 quantifies worst case (or expected worst) effective coverage $\widetilde{\mathrm{ACR}}_t^{v}$ over a stress scenario set, providing a solvency buffer that is robust to oracle execution divergence.

Let $\mathrm{ACR}_t^{v}$ be the oracle based vault coverage ratio in \eqref{eq:vault_acr}. For each collateral asset $a\in\mathcal{C}^{v}$, define the instantaneous execution deviation:
\begin{equation}
\varepsilon_{a,t} \equiv 1-\frac{\widetilde{P}_{a,t}}{P_{a,t}}\in[0,1)
\end{equation}
and define the vault collateral weight $\omega_{a,t}^{v}$ as in Proposition~\ref{prop:oracle_failure}. The collateral weighted oracle execution deviation is:
\begin{equation}
\bar{\varepsilon}_t^{v} \equiv \sum_{a\in\mathcal{C}^{v}} \omega_{a,t}^{v}\cdot\varepsilon_{a,t}
\end{equation}
Let $Z_t$ denote a vector of stress indicators measurable in $\mathcal{F}_t$, e.g., realized volatility of $P_{a,\cdot}$, aggregate liquidation volume, and congestion proxies, and let $\mathcal{S}$ be a stress scenario set represented by constraints $Z_t\in\mathcal{Z}(s)$ for $s\in\mathcal{S}$. Define the stress conditional effective coverage:
\begin{equation}
\mathrm{V1}(v,t;s) \equiv \widetilde{\mathrm{ACR}}_t^{v}\big|_{s}
\equiv \mathrm{ACR}_t^{v}\Big(1-\bar{\varepsilon}_t^{v}\big|_{Z_t\in\mathcal{Z}(s)}\Big)
\label{eq:M1_def}
\end{equation}
and define the vault level stress adjusted coverage as the expected worst coverage over $\mathcal{S}$:
\begin{equation}
\mathrm{V1}(v,t)\equiv \inf_{s\in\mathcal{S}} \ \mathbb{E}\!\left[\mathrm{V1}(v,t;s)\mid\mathcal{F}_t\right]
\end{equation}

\begin{assumption}[A1-1]
\label{ass:M1-1}
(\emph{Execution deviation is measurable and scenario conditionable.})
\end{assumption}
The process $\varepsilon_{a,t}$ is $\mathcal{F}_t$ measurable through observed liquidation transactions or executable quotes, and admits a well defined conditional distribution given $Z_t\in\mathcal{Z}(s)$. This is plausible when historical liquidation events and contemporaneous DEX state are observable. It fails when execution prices are unobservable or when stress proxies omit relevant drivers.

\begin{assumption}[A1-2]
\label{ass:M1-2}
(\emph{Scenario set $\mathcal{S}$ is economically meaningful.})
\end{assumption}
The scenario constraints $\mathcal{Z}(s)$ reflect economically relevant tail states of $(P_{a,t},D_{a,t}^{\mathrm{pool}},G_t)$ and it fails under structural breaks in liquidity provision or oracle design.

\begin{proposition}[V1: coverage bound and dominance]
\label{prop:M1_coverage_bound}
For any vault $v$ and time $t$, $\mathrm{V1}(v,t)\le \mathrm{ACR}_t^{v}$. Moreover, if $\bar{\varepsilon}_t^{v}=0$ almost surely under all $s\in\mathcal{S}$, then $\mathrm{V1}(v,t)=\mathrm{ACR}_t^{v}$.
\end{proposition}
\begin{proof}
Assumption~\ref{ass:A0} ensures $\varepsilon_{a,t}\in[0,1)$, which implies $\bar{\varepsilon}_t^{v}\in[0,1)$ and hence $(1-\bar{\varepsilon}_t^{v})\le 1$ pointwise. Therefore, $\mathrm{V1}(v,t;s)\le \mathrm{ACR}_t^{v}$ for each $s$ and taking conditional expectations and the infimum preserves the inequality. If $\bar{\varepsilon}_t^{v}=0$ almost surely in each scenario, then $\mathrm{V1}(v,t;s)=\mathrm{ACR}_t^{v}$ for all $s$.
\end{proof}

\begin{proposition}[V1 concentration monotonicity under convex impact]
\label{prop:M1_concentration}
Let $H(\omega)\equiv \sum_{a\in\mathcal{C}^{v}} \omega_{a,t}^{v}\cdot\varepsilon_{a,t}$ denote the collateral weighted deviation at time $t$. Suppose $\varepsilon_{a,t}$ is nonnegative and, conditional on the scenario, is a convex function of the portfolio weight $\omega_{a,t}^{v}$ in a neighborhood of the current allocation (reflecting increasing marginal price impact under concentrated liquidations). Then $\mathrm{V1}(v,t)$ is weakly decreasing under any mean preserving increase in weight concentration.
\end{proposition}
\begin{proof}
Under convexity of $\varepsilon_{a,t}$ in $\omega_{a,t}^{v}$, Jensen's inequality implies that $\mathbb{E}[H(\omega)\mid\mathcal{F}_t,Z_t\in\mathcal{Z}(s)]$ increases under mean preserving concentration (a spread in the weight distribution toward the high impact collateral). Since $\mathrm{V1}(v,t;s) = \mathrm{ACR}_t^{v}(1 - \bar{\varepsilon}_t^{v}|_s)$ is affine decreasing in $\bar{\varepsilon}_t^{v}$ by \eqref{eq:M1_def}, the conditional expectation $\mathbb{E}[\mathrm{V1}(v,t;s)\mid\mathcal{F}_t]$ weakly decreases for each $s$ and the infimum over $s$ preserves the ordering.
\end{proof} 
%The concavity condition requires empirical verification. 

\begin{corollary}[V1 discrete estimator]
\label{cor:M1_estimator}
With a fixed estimation window length $W$ (e.g., $W=30$ days) and a stress set $\mathcal{S}$ constructed from historical tail events, e.g., worst $N$ daily drawdowns. For each collateral asset $a$, estimate the stressed execution deviation $\widehat{\varepsilon}_{a,t}(s)$ as the median realized slippage from liquidation events or executable DEX quotes in scenario $s$ over $[t-W,t]$. Using onchain snapshots, compute collateral weights $\widehat{\omega}_{a,t}^{v}$ from oracle valued collateral quantities and then compute:
\begin{equation}
\widehat{\mathrm{V1}}(v,t;s)=\widehat{\mathrm{ACR}}_t^{v}
\left(1-\sum_{a\in\mathcal{C}^{v}} \widehat{\omega}_{a,t}^{v}\cdot\widehat{\varepsilon}_{a,t}(s)\right)
\end{equation}
reporting $\widehat{\mathrm{V1}}(v,t)=\min_{s\in\mathcal{S}} \widehat{\mathrm{V1}}(v,t;s)$. Required data inputs are (i) vault liabilities $L_t^{v}$ and oracle valued assets for $\widehat{\mathrm{ACR}}_t^{v}$, (ii) liquidation transaction data to infer realized slippage, (iii) DEX pool depth snapshots to define scenarios and (iv) oracle price feed history. 
%The choice of $W$, scenario construction, and stress quantile are \unverified.
\end{corollary}

Collateral eligibility, caps and collateral factors $w_a^{v}$ are governance controlled inputs that shape collateral composition and concentration. A complete assessment of how effectively parameters are set and updated is deferred to a companion paper on governance risk. For Level 1 purposes, these parameters are treated as fixed inputs and their sensitivity is evaluated through stress scenarios in Section~\ref{sec:empirical}.

Within V1, the oracle based coverage ratio $\mathrm{ACR}_t^{v}$ is computed using oracle prices and collateral factors. A distinct but related static buffer exists between a market's maximum borrowing ratio $\mathrm{LTV}_0$ (governing initial borrowing capacity) and its liquidation threshold $\mathrm{LLTV}$ (above which liquidation is enabled). For a fully leveraged position, the price move required to trigger liquidation from initial leverage is:
\begin{equation}
\frac{p_t^{\mathrm{trigger}}}{p_0}
= \frac{\mathrm{LTV}_0}{\mathrm{LLTV}}
\label{eq:trigger_price}
\end{equation}
so the buffer $(\mathrm{LLTV}-\mathrm{LTV}_0)$ directly controls the
percentage price decline that can be absorbed before forced liquidation. A
monitoring frequency adequacy condition requires:
\begin{equation}
\mathrm{LLTV}-\mathrm{LTV}_0
\;\ge\; k_q\cdot\frac{\sigma_a}{\sqrt{f_{\mathrm{mon}}}}
\label{eq:buffer_adequacy}
\end{equation}
where $\sigma_a/\sqrt{f_{\mathrm{mon}}}$ approximates the one standard deviation price move over a single monitoring interval when liquidation eligibility is checked at frequency $f_{\mathrm{mon}}$, the quantile multiplier $k_q$ scales this to the desired tail coverage. Under a GBM approximation, this is a lower bound and fat tailed crypto returns require substantially larger buffers in practice. The buffer adequacy condition \eqref{eq:buffer_adequacy} provides a quantitative criterion for assessing whether governance set LTV/LLTV parameters are plausible for a given collateral's volatility and monitoring regime. Formal governance measurement of parameter setting quality is deferred to the companion paper.

\subsection{Volume adjusted recovery and expected shortfall}
\label{subsec:V2}

V2 metric represents volume adjusted recovery and expected shortfall, and is motivated by Proposition~\ref{prop:recovery_endogeneity}, which establishes that realized recovery rates deteriorate with liquidation mass due to price impact and limited depth. Traditional credit modeling often treats LGD as exogenous and independent of the default event \citep{AltmanRestiSironi2004}. DeFi requires a recovery model that conditions on liquidation clustering and observable depth. V2 operationalizes this by making liquidation volume an explicit conditioning variable in expected shortfall, capturing superlinear amplification when liquidations cluster.

Let $\Delta_t^{v}$ denote the vault shortfall (Definition~\ref{def:loss_function}) and let $Q_{a,t}^{\mathrm{liq}}$ be aggregate liquidation volume in collateral asset $a$ over an infinitesimal window around $t$, and let $D_{a,t}^{\mathrm{pool}}$ denote the corresponding effective onchain depth. Assume the linear impact approximation:
\begin{equation}
f\!\left(Q_{a,t}^{\mathrm{liq}},D_{a,t}^{\mathrm{pool}}\right)=\lambda_a\cdot\frac{P_{a,t}\cdot Q_{a,t}^{\mathrm{liq}}}{D_{a,t}^{\mathrm{pool}}}
\qquad \lambda_a>0
\end{equation}
where $\lambda_a$ is an asset and venue specific impact coefficient. The total liquidation notional is defined as:
\begin{equation}
Q_t^{v} \equiv \sum_{a\in\mathcal{C}^{v}} P_{a,t}\cdot Q_{a,t}^{\mathrm{liq}}
\end{equation}
The volume adjusted expected shortfall is defined as:
\begin{equation}
\mathrm{V2}(v,t)\equiv \mathbb{E}\!\left[\Delta_t^{v} \,\big|\, \mathcal{F}_t,\, Q_t^{v}\right]
\label{eq:M2_shortfall}
\end{equation}
with $\Delta_t^{v}$ computed using realized execution prices $\widetilde{P}_{a,t}=P_{a,t}-f(\cdot)$ and realized recovery rates \eqref{eq:recovery_rate}.

\begin{assumption}[A2-1]
\label{ass:M2-1}
(\emph{Paraimpact approximation.}) \end{assumption}
Execution price impact is approximated by a linear function of normalized liquidation notional, with coefficient $\lambda_a$. This can fail under concentrated liquidity AMMs, multi venue routing and crisis nonlinearities with the consequence of tail underestimation.

\begin{assumption}[A2-2]
\label{ass:M2-2}
(\emph{Depth observability.}) 
\end{assumption}
Effective depth $D_{a,t}^{\mathrm{pool}}$ is observable from onchain liquidity curves or executable quotes. This fails under fragmented or opaque routing. 

\begin{assumption}[A2-3]
\label{ass:M2-3}
(\emph{Convex shortfall in execution deviations.})
\end{assumption}
The shortfall function $\Delta_t^{v}$ is convex in the vector $\{f(Q_{a,t}^{\mathrm{liq}}, D_{a,t}^{\mathrm{pool}})\}_{a\in\mathcal{C}^{v}}$. This holds under standard pool accounting where assets are valued at execution prices and the shortfall is the positive part of the deficit. It fails if hedging or protocol level insurance partially absorbs losses in a concave manner.

\begin{proposition}[V2 superlinearity under linear impact]
\label{prop:M2_superlinearity}
Under the linear impact approximation and Assumption~\ref{ass:M2-3}, holding $(P_{a,t},D_{a,t}^{\mathrm{pool}},\Theta_t^{v})$ fixed, the mapping $Q_t^{v}\mapsto \mathrm{V2}(v,t)$ is convex.
\end{proposition}
\begin{proof}
Under the linear impact approximation, each execution deviation $f(Q_{a,t}^{\mathrm{liq}}, D_{a,t}^{\mathrm{pool}}) = \lambda_a Q_{a,t}^{\mathrm{liq}}/D_{a,t}^{\mathrm{pool}}$ is affine in $Q_{a,t}^{\mathrm{liq}}$. Under Assumption~\ref{ass:M2-3}, $\Delta_t^{v}$ is convex in the vector of deviations, and hence (by composition with an affine map) convex in $Q_{a,t}^{\mathrm{liq}}$ for each $a$. Since $Q_t^{v} = \sum_a P_{a,t} Q_{a,t}^{\mathrm{liq}}$ is a linear aggregation, $\Delta_t^{v}$ is convex in $Q_t^{v}$. Conditional expectation preserves convexity, so $\mathrm{V2}(v,t) = \mathbb{E}[\Delta_t^{v}\mid\mathcal{F}_t, Q_t^{v}]$ is convex in $Q_t^{v}$.
\end{proof}

\begin{proposition}[V2 monotonicity in impact coefficient]
\label{prop:M2_monotone}
If $\lambda_a \le \lambda_a'$ for all $a\in\mathcal{C}^{v}$, then $\mathrm{V2}(v,t;\lambda)\le \mathrm{V2}(v,t;\lambda')$ almost surely, where $\mathrm{V2}(v,t;\lambda)$ denotes \eqref{eq:M2_shortfall} evaluated under impact coefficients $\{\lambda_a\}$.
\end{proposition}
\begin{proof}
Higher $\lambda_a$ weakly decreases $\widetilde{P}_{a,t}$ for any given $(Q_{a,t}^{\mathrm{liq}},D_{a,t}^{\mathrm{pool}})$, which weakly decreases realized assets $A_t^{v}$ and weakly increases $\Delta_t^{v}=(L_t^{v}-A_t^{v})^+$. Taking conditional expectations preserves the inequality. 
\end{proof}

\begin{corollary}[V2 discrete estimator]
\label{cor:M2_estimator}
Fix an estimation window $W$ and discretize time at $\Delta$, e.g., hourly. For each asset $a$, estimate $\lambda_a$ by regressing realized slippage from liquidation events (or executable quote slippage) on normalized liquidation notional:
\begin{equation}
\widehat{\varepsilon}_{a,\tau} \approx \widehat{\lambda}_a\cdot\frac{P_{a,\tau}\cdot Q_{a,\tau}^{\mathrm{liq}}}{D_{a,\tau}^{\mathrm{pool}}}
\qquad \tau\in[t-W,t]
\end{equation}
where $\widehat{\varepsilon}_{a,\tau}=1-\widetilde{P}_{a,\tau}/P_{a,\tau}$ is inferred from onchain execution. Then compute stressed shortfall paths by simulating liquidation volumes $Q_{a,\tau}^{\mathrm{liq}}$ under a stress scenario set, e.g., historical tail price moves, generating:
\begin{equation}
\widehat{\mathrm{V2}}(v,t) = \frac{1}{N}\sum_{n=1}^N \widehat{\Delta}_{t}^{v}(n)
\end{equation}
where $\widehat{\Delta}_{t}^{v}(n)$ is the simulated shortfall under scenario path $n$. Required inputs are liquidation event logs, pool depth snapshots $D_{a,\tau}^{\mathrm{pool}}$, oracle prices $P_{a,\tau}$ and vault wide liabilities $L_\tau^{v}$. The stress simulation requires an assumption about how $D_{a,\tau}^{\mathrm{pool}}$ evolves under the same price shock that generates liquidation clustering. A standard conservative assumption is that depth contracts proportionally with realized volatility, consistent with the empirical pattern that LP withdrawals accelerate as price moves exceed historical norms. Implementation requires selecting reference DEX pools per collateral asset. These should be reviewed at least monthly, since pool depth can migrate across venues following incentive program changes.\footnote{Pool selection bias is a material estimation risk for $\widehat{\lambda}_a$ since a thin or illiquid venue systematically overstates the impact coefficient, while an aggregated quote including routing optimization understates it. Conservative practice is to use worst case single venue depth and report sensitivity across alternative venue assumptions.}
\end{corollary}

Liquidation incentives and collateral caps influence liquidation clustering and the effective liquidation notional $Q_t^{v}$, which enters V2. Formal evaluation of how these parameters are chosen and updated is deferred to a companion paper. For Level 1, the parameter configuration is fixed and V2 is computed under stress scenarios to quantify sensitivity to misspecification.

\subsection{Liquidity stress index}
\label{subsec:V3}

V3 metric represents liquidity stress index and is motivated by Proposition~\ref{prop:liquidity_run}, which highlights that utilization can induce blocked withdrawals and run dynamics under common information, without an external liquidity backstop (see Remark~\ref{rem:common_info} for the qualification that common information is a worst case bound). Utilization is a state variable but not a forward looking risk measure. It does not quantify the probability that the withdrawal blocking boundary is reached within a horizon under stress.
Therefore, V3 is defined as the stress conditional probability that
utilization $U_t^v$ reaches the withdrawal blocking boundary $U^{\max}$
within a horizon $h$, computed over a scenario set $\mathcal{S}$ that
conditions on adverse collateral price paths and jumps calibrated
to the vault's historical stress episodes.
This construction operationalizes funding liquidity fragility as a
forward looking depositor risk object since it captures not only the current
level of utilization but the speed and volatility of the path toward the
boundary, the amplifying feedback between liquidation activity and
withdrawal demand, and the absence of any external backstop that would
contain the dynamic before the boundary is reached.

For a fixed horizon $h>0$, let $U_t^{v}$ be the utilization ratio in \eqref{eq:utilization}. Define the liquidity stress index as the stress conditional boundary hitting probability:
\begin{equation}
\mathrm{V3}(v,t;h) \equiv \Pr\!\left(\sup_{s\in[t,t+h]} U_s^{v} \ge 1 \,\Big|\, \mathcal{F}_t,\ Z_t\in\mathcal{Z}(\mathcal{S})\right)
\label{eq:M3_def}
\end{equation}
where $Z_t$ is a vector of stress indicators and $\mathcal{Z}(\mathcal{S})$ denotes the union of stress constraints defining the scenario set $\mathcal{S}$. To represent clustered withdrawals under common information, assume a reduced form utilization dynamics:
\begin{equation}
dU_t^{v} = \mu(U_t^{v},t)\,dt + \sigma_U(U_t^{v},t)\,dW_t + dJ_t
\end{equation}
where $\mu(\cdot)$ and $\sigma_U(\cdot)$ are drift and diffusion terms and $J_t$ is a pure jump process capturing run like outflows (or sudden borrow demand spikes) that are conditionally activated in stress.

\begin{assumption}[A3-1]
\label{ass:M3-1}
(\emph{Reduced form utilization dynamics.})\end{assumption}
Utilization admits a well defined semimartingale representation with a jump component $J_t$ that captures clustered net flows. This fails if utilization is dominated by deterministic mechanics not captured by the specification.

\begin{assumption}[A3-2]
\label{ass:M3-2}
(\emph{Stress conditioning is stable.})\end{assumption}
Conditioning on $Z_t\in\mathcal{Z}(\mathcal{S})$ isolates elevated run propensity. It fails when stress is driven by unobserved coordination shocks. 

\begin{proposition}[V3: utilization monotonicity]
\label{prop:M3_utilization_monotone}
Assume $\mu(U,t)$ is nondecreasing in $U$ and $\sigma_U(U,t)$ is nondecreasing in $U$ on $[0,1)$, and that the jump intensity of $J_t$ is nondecreasing in $U$ under stress. Then $\mathrm{V3}(v,t;h)$ is nondecreasing in the current utilization $U_t^{v}$.
\end{proposition}
\begin{proof}
Higher initial utilization increases drift toward the boundary and increases volatility under the stated monotonicity conditions. For the diffusion component, the result follows from stochastic comparison where under nondecreasing drift $\mu(U,t)$ and diffusion $\sigma_U(U,t)$ in $U$, the solution process starting from higher $U_t^v$ dominates in the first order stochastic sense\footnote{See \citet{Ikeda1989} Ch.\ VI, Theorem 1.1}, implying a higher boundary hitting probability. For the jump component, under a conditional jump intensity $\lambda(U_t^v)$ that is nondecreasing in $U$, the result follows from stochastic ordering for jump diffusions (Decreusefond and Moyal 2008) where higher intensity increases the probability of crossing the boundary within horizon $h$. 
\end{proof}

\begin{proposition}[V3 stress amplification via liquidation feedback]
\label{prop:M3_stress_amplification}
Suppose that in stress, liquidation trigger clustering increases utilization through a repayment delay or liquidity drain channel such that $dJ_t$ has higher jump intensity when the measure of triggered accounts $\sum_u \mathbf{1}\{\mathcal{L}_{u,t}\}$ increases. Then $\mathrm{V3}(v,t;h)$ is nondecreasing in the conditional expectation $\mathbb{E}[\sum_u \mathbf{1}\{\mathcal{L}_{u,t}\}\mid\mathcal{F}_t]$.
\end{proposition}
\begin{proof}
An increase in expected triggers raises the conditional intensity of negative net flow jumps in $J_t$, which increases the boundary hitting probability by stochastic comparison. 
\end{proof} 

\begin{corollary}[V3 discrete estimator]
\label{cor:M3_estimator}
Fix discretization $\Delta$, e.g., one hour, and horizon $h=m\Delta$. Using historical utilization observations $\{U_{t-j\Delta}^{v}\}_{j=0}^{J}$ over window $W=J\Delta$, estimate a reduced form transition model for increments $\Delta U_{t}^{v}$, allowing for a jump events. Under stress conditioning, simulate $N$ paths forward for $m$ steps to compute:
\begin{equation}
\widehat{\mathrm{V3}}(v,t;h)=\frac{1}{N}\sum_{n=1}^N \mathbf{1}\left\{\max_{j=1,\ldots,m} U_{t+j\Delta}^{v}(n)\ge 1\right\}
\end{equation}
Required inputs are time series of total borrows $B_t^{v}$ and deposits $D_t^{v}$ to compute $U_t^{v}$, stress proxies $Z_t$, and liquidation trigger counts.
\end{corollary}

Interest rate curve parameters and caps influence utilization dynamics and thus V3. Governance quality is analyzed elsewhere. In this paper, these parameters enter as fixed inputs and V3 is evaluated under stress scenarios to quantify liquidity fragility conditional on the current configuration.

\subsection{Oracle integrity score}
\label{subsec:V4}

V4 metric represents oracle integrity score and is motivated by Proposition~\ref{prop:oracle_manipulation}, which shows that oracle error can induce false solvency and false insolvency, making oracle integrity a first order driver of depositor losses. Unlike discretionary valuation in TradFi, oracle inputs directly drive automated enforcement, so errors can translate into missed liquidations and latent bad debt. V4 quantifies the expected shortfall attributable to oracle imperfections and normalizes it by liabilities, yielding a score on $[0,1]$.

Let $\Delta_t^{v,\mathrm{oracle}}$ denote oracle induced shortfall as in Proposition~\ref{prop:oracle_manipulation}. The oracle integrity score is defined as:
\begin{equation}
\mathrm{V4}(v,t)\equiv 1-\frac{\mathbb{E}\!\left[\Delta_t^{v,\mathrm{oracle}}\mid\mathcal{F}_t\right]}{L_t^{v}}
\qquad \mathrm{V4}(v,t)\in[0,1]
\end{equation}
Decompose oracle risk into latency and manipulation components by defining $\mathrm{V4a}(v,t)$ for latency and $\mathrm{V4b}(v,t)$ for manipulation, with
\begin{equation}
\mathrm{V4}(v,t)=\mathrm{V4a}(v,t)\cdot \mathrm{V4b}(v,t)
\end{equation}
interpreting the product as a conservative combination under potential dependence. For latency, assume a lag model:
\begin{equation}
P_{a,t}=P_{a,t-\delta_a^{\mathrm{oracle}}}^{\mathrm{true}}
\qquad \delta_a^{\mathrm{oracle}}\ge 0
\label{eq:M4a_latency}
\end{equation}
where $\delta_a^{\mathrm{oracle}}$ is the oracle staleness for asset $a$. For manipulation, define an economic feasibility condition based on a reference market depth proxy $\mathcal{D}_{a,t}^{\mathrm{ref}}$ and a targeted oracle deviation magnitude $\Delta P_{a,t}$:
\begin{equation}
\text{Manipulation feasible at }t \ \text{if}\
\mathrm{Benefit}_{t}(\Delta P_{a,t}) \;>\; \mathrm{Cost}_{t}(\Delta P_{a,t};\mathcal{D}_{a,t}^{\mathrm{ref}})
\label{eq:M4b_manipulation}
\end{equation}
where benefit and cost are measured in unit of account and depend on liquidation boundary geometry and the liquidity of the reference market used by the oracle.

\begin{assumption}[A4-1]
\label{ass:M4-1}
(\emph{Latency model adequacy.})\end{assumption}
Oracle staleness can be summarized by a lag $\delta_a^{\mathrm{oracle}}$ as in \eqref{eq:M4a_latency} but fails for complex aggregation/filtering mechanisms.

\begin{assumption}[A4-2]
\label{ass:M4-2}
(\emph{Manipulation feasibility is economically characterizable.})\end{assumption}
A meaningful cost and benefit model exists for oracle manipulation, capturing the dominant attack vector in public protocols where economic incentives are transparently encoded. This assumption fails without reliable reference depth and attack cost measurement, and does not capture attacks motivated by non economic goals such as adversarial actors targeting a competing protocol regardless of direct profit. Therefore, non economic manipulation should be treated as a tail event in V4b sensitivity analysis.

Assumptions~\ref{ass:M4-1} and~\ref{ass:M4-2} presuppose that oracle error
$\eta_{a,t}$ is estimable from observable data and that its conditional
distribution given $\mathcal{F}_t$ has mean close to zero in normal regimes.

\begin{remark}[RWA collateral and systematic oracle bias]
\label{rem:rwa_oracle}
For real world asset (RWA) collateral (including tokenized private credit, real estate backed instruments and tokenized money market fund claims) Assumption~\ref{ass:M4-2} fails structurally. Define the offchain information arrival time $\tau_{\mathrm{info}}$ (when issuer level impairment information becomes available) and the onchain update time $\tau_{\mathrm{chain}}$ (when this information is reflected in the oracle NAV feed). The latency window $\delta^{\mathrm{NAV}}=\tau_{\mathrm{chain}}-\tau_{\mathrm{info}}\ge 0$
generates a systematic, directional oracle bias:
\begin{equation}
\mathbb{E}\!\left[\eta_{a,t}\mid \tau_{\mathrm{info}}\le t<\tau_{\mathrm{chain}}\right]
= P_{a,t}^{\mathrm{oracle}} - P_{a,t}^{\mathrm{true}}
> 0
\label{eq:nav_bias}
\end{equation}
that is, the oracle systematically overstates collateral value by $\eta_{a,t}$ in unit of account terms during the latency window.
Positions that appear solvent under oracle prices may be insolvent under true prices, generating unrecognized expected shortfall
\begin{equation}
\mathbb{E}\!\left[\Delta_t^{v,\mathrm{NAV}}\right]
\;\ge\;
f^{\mathrm{RWA}}\cdot L_t^{v}\cdot
\frac{\delta^{\mathrm{NAV}}\cdot
\left|\frac{\mathrm{d}P_{a}^{\mathrm{true}}}{\mathrm{d}t}\right|^{\!-}}
{P_{a,t}}
\label{eq:nav_shortfall}
\end{equation}
where $\delta^{\mathrm{NAV}} = \tau_{\mathrm{chain}} - \tau_{\mathrm{info}}$ is the NAV latency window and $\left|\mathrm{d}P_a^{\mathrm{true}}/\mathrm{d}t\right|^{-}$ is the magnitude of the price drift during a declining episode. The consequence is that Proposition~\ref{prop:PI1_oracle_bound} (which bounds oracle error using reference price spreads) does not hold for RWA assets. The identified set for $\eta_{a,t}$ is not bounded by reference price deviations because the reference prices themselves are stale when offchain information exists. For vaults with material RWA collateral, V4 must be augmented with an offchain NAV monitoring discipline that separately tracks $\delta^{\mathrm{NAV}}$ and the expected impairment pace, and the worst case oracle bias should be treated as a floor on V4.
\end{remark}

\begin{remark}[Market hours oracle staleness for tokenized assets]
\label{rem:mkt_hours}
For collateral assets backed by instruments traded during fixed market hours (tokenized equities, REITs, and certain fixed income instruments), the oracle feed can be frozen at the last traded price during market closures. The oracle error $\eta_{a,t}$ during a closed market period of duration $h$ is not mean zero. Economic value continues to evolve through after hours events, so the oracle price $P_{a,t}$ is simply the prior close. Under a GBM approximation for the underlying, the expected absolute gap at market open is:
\begin{equation}
\mathbb{E}\!\left[|P_{a,t_{\mathrm{open}}}^{\mathrm{market}}
- P_{a,t_{\mathrm{close}}}^{\mathrm{oracle}}|\right]
\approx P_{a,t_{\mathrm{close}}}\cdot\sigma_a\cdot\sqrt{h/365}
\label{eq:gap_risk}
\end{equation}
where $\sigma_a$ is calendar time annualised volatility and $h$ is measured in calendar days.\footnote{If $\sigma_a$ is calibrated as a 252 day trading volatility, it must be rescaled by $\sqrt{252/365}$ to obtain the calendar time equivalent.} In practice, fat-tailed returns and the concentration of macro news in market closed periods make \eqref{eq:gap_risk} a lower bound. The credit risk implication is that the oracle staleness $\delta_a^{\mathrm{oracle}}$ in Procedure~P2 must be defined as the maximum of the typical update frequency and the market closure duration for RWA backed assets, the latter can dominate. Consequently, the LTV/LLTV buffer for any collateral with periodic market closures must embed the expected gap at open as an additional coverage buffer over and above intraday volatility considerations.
\end{remark}

%\paragraph{Properties.}
%\begin{proposition}[V4a: volatility--latency sensitivity]
%\label{prop:M4a_volatility_latency}
%Assume the latent log price for collateral asset $a$ satisfies $d\ln P_{a,t}^{\mathrm{true}}=\mu_a\,dt+\sigma_a\,dW_t$ with $\sigma_a>0$ and that oracle pricing follows \eqref{eq:M4a_latency}. Then the variance of oracle error satisfies
%\begin{equation}
%\mathrm{Var}\!\left(\eta_{a,t}\mid\mathcal{F}_{t-%\delta_a^{\mathrm{oracle}}}\right)
%\ \text{is increasing in}\ \delta_a^{\mathrm{oracle}}\sigma_a^2,
%\end{equation}
%and therefore latency induced expected oracle shortfall is weakly increasing in $\delta_a^{\mathrm{oracle}}\sigma_a$ under monotone liquidation-trigger sensitivity to price.
%\end{proposition}
%\begin{proof}[Sketch]
%Under GBM, $\ln P_{a,t}^{\mathrm{true}}-\ln P_{a,t-\delta}^{\mathrm{true}}$ has variance $\sigma_a^2\delta$. Since $P_{a,t}$ is lagged, oracle error inherits dispersion increasing in $\sigma_a^2\delta$. Mapping dispersion to expected shortfall depends on liquidation geometry and is empirical. \unverified
%\end{proof}

\begin{proposition}[V4 boundedness]
\label{prop:M4_boundedness}
For any vault $v$ and time $t$, $\mathrm{V4}(v,t)\in[0,1]$, with $\mathrm{V4}(v,t)=1$ if and only if $\mathbb{E}[\Delta_t^{v,\mathrm{oracle}}\mid\mathcal{F}_t]=0$.
\end{proposition}
\begin{proof}
$\Delta_t^{v,\mathrm{oracle}}\ge 0$ by definition and cannot exceed $L_t^{v}$, so the normalized expected oracle shortfall $\mathbb{E}[\Delta_t^{v,\mathrm{oracle}}\mid\mathcal{F}_t]/L_t^{v}$ lies in $[0,1]$, and, therefore, $\mathrm{V4}(v,t) = 1 - \mathbb{E}[\Delta_t^{v,\mathrm{oracle}}\mid\mathcal{F}_t]/L_t^{v} \in [0,1]$. 
\end{proof}

\begin{corollary}[V4 discrete estimator]
\label{cor:M4_estimator}
Using discrete timestamps of oracle updates, estimate $\widehat{\delta}_a^{\mathrm{oracle}}(t)$ as the time elapsed since the last oracle update for asset $a$ at time $t$. Estimate realized volatility $\widehat{\sigma}_a(t)$ over a rolling window $W$ from high frequency oracle prices or reference spot prices. Under \eqref{eq:M4a_latency}, approximate latency risk by computing a false solvency probability proxy proportional to $\widehat{\delta}_a^{\mathrm{oracle}}(t)\widehat{\sigma}_a^2(t)$ and map it into $\widehat{\mathrm{V4a}}(v,t)$ via a monotone transform calibrated to historical near miss liquidations. For manipulation risk, extract oracle architecture parameters (TWAP window, data sources, reference venues) and estimate reference depth $\widehat{\mathcal{D}}_{a,t}^{\mathrm{ref}}$, compute a feasibility indicator from \eqref{eq:M4b_manipulation} using conservative lower bounds for $\mathrm{Cost}_t$ and report:
\begin{equation}
\widehat{\mathrm{V4}}(v,t)=\widehat{\mathrm{V4a}}(v,t)\cdot \widehat{\mathrm{V4b}}(v,t)
\end{equation}
Required inputs are oracle update logs, price histories, oracle mechanism documentation and incident/attack datasets.
\end{corollary}

\begin{remark}[Oracle immutability and forward looking quality requirements]
\label{rem:oracle_immut}
In some lending protocol designs, oracle addresses are immutable per market since they are set at market deployment and cannot be changed without migrating the entire market. Canonical instances include fixed oracle market designs such as Morpho Blue. In such settings, the vault cannot respond to oracle quality deterioration between deployment and the end of its operating horizon, even if governance timelocks would otherwise permit parameter updates. The consequence for V4 is that initial oracle quality must be assessed not only for current fitness but for expected fitness over the full anticipated horizon. If the oracle reference market's liquidity is expected to decline or if the oracle's computational logic becomes outdated, the expected future V4 score is systematically worse than its current value but cannot be corrected through ordinary governance. Formally, let $Q_a(t)$ denote an oracle quality index for asset $a$ (higher is better) and let $Q^{\min}$ be a minimum acceptable quality. A conservative time 0 oracle adequacy condition, derived from a first order delta approximation, is:
\begin{equation}
Q_a(0) \;\ge\; Q^{\min} + \frac{\sup_{t\in[0,T]}\mathrm{Var}(Q_a(t))^{1/2}}{\bigl|\partial \mathrm{V4}/\partial Q_a\bigr|_0}
\label{eq:oracle_irreversibility}
\end{equation}
where $\sup_{t\in[0,T]}\mathrm{Var}(Q_a(t))^{1/2}$ is the worst case standard deviation of oracle quality variation over the vault's horizon, and $|\partial \mathrm{V4}/\partial Q_a|_0$ is the sensitivity of V4 to oracle quality evaluated at $t=0$.\footnote{This conservative bound follows from a first order delta approximation where the expected decline in V4 attributable to oracle quality deterioration is at most the sensitivity of V4 to $Q_a$, multiplied by the standard deviation of quality variation. Dimensions are consistent, $[Q_a(0)]$ and $[Q^{\min}]$ are in units of oracle quality; the ratio is $[\mathrm{Std}(Q_a)]/[1/\mathrm{quality}] = [\mathrm{quality}]$.} Immutable oracle protocols require a more conservative initial oracle selection and cannot rely on governance responsiveness to compensate for post deployment quality decline.
\end{remark}

Oracle selection and parameterization are governance controlled inputs, but this article holds them fixed and measures the implied oracle integrity exposure. A companion paper addresses governance quality, including how oracles are selected and updated. For Level 1, V4 is evaluated under stress scenarios to quantify loss sensitivity to oracle latency and manipulation feasibility.

\subsection{Execution viability under stress}
\label{subsec:V5}

V5 metric represents execution viability under stress and is motivated by Proposition~\ref{prop:network_congestion}, which shows that liquidation viability can fail with higher probability precisely when liquidation triggers occur, due to congestion and MEV competition. Traditional settlement risk frameworks assume institutional default management resources. DeFi relies on decentralized execution where costs are endogenous to stress. V5 quantifies the stress conditional probability that a triggered liquidation is economically viable and estimates the expected shortfall contribution of execution failure.

Let $\mathcal{L}_{u,t}=\{\mathrm{HF}_{u,t}^{v}<1\}$ be the liquidation trigger event and let $\mathcal{V}_{u,t}$ be the liquidation viability event in \eqref{eq:viability}. Define the vault level execution viability rate as the conditional viability probability averaged over the distribution of triggered accounts:
\begin{equation}
\mathrm{V5}(v,t)\equiv \mathbb{E}\!\left[\Pr\!\left(\mathcal{V}_{u,t}\mid \mathcal{L}_{u,t},\mathcal{F}_t\right)\right]
\label{eq:M5_viability_rate}
\end{equation}
where the outer expectation is taken over accounts $u$ weighted by their contribution to liquidation mass. Define the stress conditional execution shortfall contribution as:
\begin{equation}
\mathrm{V5\text{-}ES}(v,t)\equiv
\mathbb{E}\!\left[\Delta_t^{v} \mid \mathcal{L}_{u,t}, \mathcal{V}_{u,t}^{c}, \mathcal{F}_t\right]\cdot
\Pr\!\left(\mathcal{V}_{u,t}^{c}\mid \mathcal{L}_{u,t},\mathcal{F}_t\right)
\label{eq:M5_stress_es}
\end{equation}
interpreted as expected shortfall attributable to execution failure conditional on liquidation need.

\begin{assumption}[A5-1]
\label{ass:M5-1}
(\emph{Rational liquidator participation.})
\end{assumption} Liquidators participate when the net profit condition in \eqref{eq:viability} holds  and it fails if participation is constrained by non economic frictions.

\begin{assumption}[A5-2]
\label{ass:M5-2}
(\emph{Stress correlation is stable.})
\end{assumption} Gas price $g_t$ and MEV competition are positively correlated with stress measures driving liquidation triggers. This fails under structural changes in chain congestion regimes. 

\begin{proposition}[V5 correlation decay of viability]
\label{prop:M5_correlation_decay}
Let $S_t$ denote a scalar stress measure and suppose $\mathrm{Corr}(g_t,S_t)=\rho_g>0$ under stress, while $\Pi_{u,t}^{\mathrm{liq}}$ is decreasing in $S_t$ through worsening execution prices. Then $\mathrm{V5}(v,t)$ is weakly decreasing in $\rho_g$ holding marginal distributions fixed.
\end{proposition}
\begin{proof}
Higher $\rho_g$ increases the probability that high gas coincides with high stress, which raises the probability that costs exceed liquidation profits on the trigger set $\mathcal{L}_{u,t}$. Formally, under log-supermodularity of the joint density of $(g_t, S_t)$, the FKG inequality\footnote{See \citet{Müller2000} Theorem 3.1} implies that the conditional probability $\Pr(\mathcal{V}_{u,t}^c \mid \mathcal{L}_{u,t})$ is nondecreasing in $\rho_g$, holding marginal distributions fixed. Therefore $\mathrm{V5}(v,t) = 1 - \Pr(\mathcal{V}_{u,t}^c \mid \mathcal{L}_{u,t})$ is weakly decreasing in $\rho_g$. A complete proof requires specification of the joint distribution of $(g_t, S_t)$, the result holds under log-supermodularity of that density. We treat this as a structural proposition supported by standard positive dependence theory. 
\end{proof}

\begin{proposition}[V5 incentive expansion of viability region]
\label{prop:M5_incentive_expansion}
Suppose the liquidation incentive parameter $\pi$ enters liquidator profit additively so that $\Pi_{u,t}^{\mathrm{liq}}(\pi)$ is increasing in $\pi$. Then $\mathrm{V5}(v,t)$ is nondecreasing in $\pi$, and $\lim_{\pi\to\infty}\mathrm{V5}(v,t)=1$ provided costs remain finite.
\end{proposition}
\begin{proof}
Increasing $\pi$ expands the state set satisfying \eqref{eq:viability} by shifting the profitability threshold downward. Monotone convergence then yields $\mathrm{V5}(v,t)\to 1$ as $\Pi_{u,t}^{\mathrm{liq}}$ diverges while costs remain bounded. The practical relevance of large $\pi$ is limited by the V2 trade off. A higher bonus increases collateral seized per unit debt repaid, amplifying liquidation mass $Q_t^{v}$ and worsening execution prices through the channel in Proposition~\ref{prop:recovery_endogeneity}. 
\end{proof}

\begin{corollary}[V5 discrete estimator]
\label{cor:M5_estimator}
Using a discrete time grid, e.g., hourly, estimate $\rho_g$ as the empirical correlation between gas prices and absolute collateral return magnitudes over a rolling window $W$, conditioning on stress states. Estimate $\Pr(\mathcal{V}_{u,t}^{c}\mid \mathcal{L}_{u,t})$ as the fraction of liquidation triggered positions that remain unliquidated beyond a protocol specific maximum execution window, then compute:
\begin{equation}
\widehat{\mathrm{V5}}(v,t)=1-\widehat{\Pr}\!\left(\mathcal{V}_{u,t}^{c}\mid \mathcal{L}_{u,t}\right)
\end{equation}
and estimate $\widehat{\mathrm{V5\text{-}ES}}(v,t)$ via scenario simulation using delayed liquidation frequency and realized shortfalls. Required inputs are gas price time series, trigger and completion logs, and MEV proxies.
\end{corollary}

Liquidation incentive and close factor parameters influence $\Pi_{u,t}^{\mathrm{liq}}$ and thus V5. Governance quality is analyzed in a companion paper. Here these parameters are fixed inputs, and stress scenario analysis is used to evaluate V5 under plausible gas and depth regimes.

\subsection{Vault credit score}
\label{subsec:VCS}

Aggregation of V1-V5 into a single scalar inevitably requires normalization choices. The purpose of a composite score is monitoring and comparability, not to replace the underlying mechanism level diagnostics. Let $\hat{m}_j(v,t)\in[0,1]$ denote a normalized score derived from $\mathrm{V}j$ (with higher values indicating lower risk) using monotone transformations calibrated to acceptable risk tolerances. The aggregated score is:
\begin{equation}
\mathrm{VCS}(v,t)\equiv \phi\!\left(\hat{m}_1(v,t),\hat{m}_2(v,t),\hat{m}_3(v,t),\hat{m}_4(v,t),\hat{m}_5(v,t)\right)
\qquad \phi:\,[0,1]^5\to[0,1]
\end{equation}

A weakest link aggregation is multiplicative:
\begin{equation}
\mathrm{VCS}^{\mathrm{mult}}(v,t)\equiv \prod_{j=1}^{5}\hat{m}_j(v,t)
\label{eq:VCS_mult}
\end{equation}
An alternative is a weighted additive score:
\begin{equation}
\mathrm{VCS}^{\mathrm{add}}(v,t)\equiv \sum_{j=1}^{5}\omega_j\cdot\hat{m}_j(v,t)
\qquad \omega_j\ge 0,\ \sum_{j=1}^{5}\omega_j=1
\label{eq:VCS_add}
\end{equation}
where weights $\{\omega_j\}$ require empirical calibration and are, therefore, based on backtesting.

\begin{proposition}[Weakest-link dominance]
\label{prop:aggregation_dominance}
If $\hat{m}_j(v,t)=0$ for any $j\in\{1,\ldots,5\}$, then $\mathrm{VCS}^{\mathrm{mult}}(v,t)=0$ regardless of $\hat{m}_{\ell}(v,t)$ for $\ell\neq j$.
\end{proposition}
\begin{proof}
Immediate from \eqref{eq:VCS_mult}, the product contains $\hat{m}_j(v,t)$ as a factor. 
\end{proof}

\begin{proposition}[Multiplicative score is conservative]
\label{prop:aggregation_amgm}
If $\omega_j = 1/5$ for all $j$ and all $\hat{m}_j(v,t) \in [0,1]$, then $\mathrm{VCS}^{\mathrm{mult}}(v,t) \leq \mathrm{VCS}^{\mathrm{add}}(v,t)$, with equality if and only if $\hat{m}_1 = \hat{m}_2 = \hat{m}_3 = \hat{m}_4 = \hat{m}_5$.
\end{proposition}
\begin{proof}
By the AM--GM inequality, $\bigl(\prod_{j=1}^{5}\hat{m}_j\bigr)^{1/5} \leq \frac{1}{5}\sum_{j=1}^{5}\hat{m}_j$, with equality if and only if all $\hat{m}_j$ are equal. Therefore:
\[
\mathrm{VCS}^{\mathrm{mult}} = \prod_{j=1}^{5}\hat{m}_j \;\leq\; \left(\frac{1}{5}\sum_{j=1}^{5}\hat{m}_j\right)^{\!5} \;\leq\; \frac{1}{5}\sum_{j=1}^{5}\hat{m}_j \;=\; \mathrm{VCS}^{\mathrm{add}},
\]
where the first inequality is AM--GM raised to the fifth power, and the second uses $x^5 \leq x$ for $x \in [0,1]$. 
\end{proof}

The five metrics are not independent since oracle error interacts with coverage and liquidation timing (V1 and V4), depth depletion interacts with execution viability (V2 and V5) and congestion can exacerbate utilization stress and liquidation delays (V3 and V5). These interactions imply that a fully coherent vault loss model is a joint simulation problem on $(X_t^{v},\mathcal{M})$ and not a separable sum of marginal risks. Additive scoring can mask compounding tail dependence. The practical role of $\mathrm{VCS}^{\mathrm{add}}$ is monitoring under a chosen weighting scheme, while $\mathrm{VCS}^{\mathrm{mult}}$ provides a conservative weakest-link summary that penalizes single point failure. The empirical literature on composite credit assessment offers a relevant parallel. \citet{CantorPackerCole1997} demonstrate for split rated corporate bonds that averaging two rating opinions yields strictly lower forecast error than relying on either agency alone or on the more conservative rating. This result motivates $\mathrm{VCS}^{\mathrm{add}}$ for routine monitoring, where metrics reflect partially independent information channels. However, \citet{Gordy2000} shows that when the systemic factor driving joint defaults is fat-tailed, models that appear similar on mean and variance can diverge sharply at high percentiles of the loss distribution, exactly the regime that matters for depositor protection. In DeFi, the equivalent of a fat-tailed systemic factor is a correlated stress event (simultaneous price shock, depth depletion, and gas congestion) that degrades all metrics jointly. In such a regime, $\mathrm{VCS}^{\mathrm{mult}}$ is the appropriate conservative operator. 

\begin{corollary}[VCS practical implementation checklist]
\label{cor:VCS_checklist}
To compute $\mathrm{VCS}^{\mathrm{mult}}(v,t)$ and $\mathrm{VCS}^{\mathrm{add}}(v,t)$ in practice, the following data inputs are required, at hourly or daily frequency unless stated otherwise:
\begin{enumerate}
\item Vault balance sheet: $L_t^{v}$ (depositor liabilities) and the composition of collateral and debt exposures sufficient to compute $\mathrm{ACR}_t^{v}$ and utilization $U_t^{v}$.
\item Oracle data: price feed values $P_{a,t}$ for all $a\in\mathcal{C}^{v}\cup\mathcal{D}^{v}$, plus oracle update timestamps to estimate $\delta_a^{\mathrm{oracle}}$.
\item Liquidation data: event level logs indicating liquidation triggers, execution completion times, repaid debt, seized collateral, and realized execution prices (or sufficient information to infer $\widetilde{P}_{a,t}$ and $\widehat{\varepsilon}_{a,t}$).
\item Onchain liquidity: DEX pool state or executable quote curves to measure $D_{a,t}^{\mathrm{pool}}$ and reference depth $\mathcal{D}_{a,t}^{\mathrm{ref}}$ for oracle manipulation feasibility proxies.
\item Network conditions: gas price time series $G_t$, priority fees/builder payments as MEV proxies, and block level congestion indicators.
\item Stress scenario definitions: a specification of $\mathcal{S}$ and horizon parameters $(W,h)$ for estimators in corollaries~\ref{cor:M1_estimator}--\ref{cor:M5_estimator}, each requiring calibration to the vault’s collateral and chain environment.
\end{enumerate}
\end{corollary}

A complete assessment of depositor risk depends on all three levels of the decomposition in Definition~\ref{def:three_level}, mechanical loss channels (this paper), governance quality (a companion paper), and infrastructure and code integrity (Section~\ref{subsec:level3}). The VCS reported here is explicitly a Level 1 structural score conditional on $\mathcal{I}_t^{v}=\mathbf{1}$ (functioning contracts) and the current parameter configuration. Whether Level 3 or Level 1 is the binding risk concern for a given vault is determined by the dominance condition \eqref{eq:L3_dominance}.

% ── FIGURE 2: VCS Architecture ───────────────────────────────────
\begin{figure}[htbp]
\centering
\begin{tikzpicture}[
  propbox/.style={
    draw=black, thick, rounded corners=2pt,
    minimum width=3.4cm, minimum height=0.7cm,
    inner sep=5pt, align=center, font=\scriptsize
  },
  metricbox/.style={
    draw=black, thick, rounded corners=2pt,
    minimum width=3cm, minimum height=0.7cm,
    inner sep=5pt, align=center, font=\small\bfseries,
    fill=white
  },
  normbox/.style={
    draw=black, thick, rounded corners=2pt,
    minimum width=2.4cm, minimum height=0.7cm,
    inner sep=5pt, align=center, font=\scriptsize,
    fill=gray!8
  },
  vcsbox/.style={
    draw=black, very thick, rounded corners=3pt,
    minimum width=2.4cm, minimum height=1.5cm,
    inner sep=6pt, align=center, font=\small,
    fill=gray!15
  },
  arr/.style={->, >=stealth, semithick},
  darr/.style={->, >=stealth, semithick, dashed, gray!70},
  xsep/.style={font=\scriptsize\itshape, text=gray!80}
]

%% ── Column 1: Six propositions ──────────────────────────────────────────
%  Spacing increased from 1.0 to 1.3 to match wider column
\foreach \i/\lbl/\desc in {
  1/{P1}/{oracle execution divergence},
  2/{P2}/{Recovery endogeneity},
  3/{P3}/{full information runs},
  4/{P4}/{Timelock constraint},
  5/{P5}/{Oracle manipulation/latency},
  6/{P6}/{Congestion wrong way risk}
}{
  \node[propbox, fill=gray!5] (P\i) at (0, {-(\i-1)*1.3}) {\textbf{\lbl}\\\desc};
}

\node[font=\scriptsize\bfseries, above=0.15cm of P1] {\textit{Failure mode (Proposition)}};

%% ── Column 2: Five metrics + stress-methodology node ────────────────────
%  Spacing 1.3 units; ST sits between V3 (y=-2.6) and V4 (y=-3.9)
\node[metricbox] (V1) at (4.5,  0.0) {V1\\{\scriptsize Coverage}};
\node[metricbox] (V2) at (4.5, -1.3) {V2\\{\scriptsize Recovery}};
\node[metricbox] (V3) at (4.5, -2.6) {V3\\{\scriptsize Liquidity}};
\node[normbox, minimum width=3cm] (ST) at (4.5, -3.9)
  {{\scriptsize Frozen-param stress}};
\node[metricbox] (V4) at (4.5, -5.2) {V4\\{\scriptsize Oracle}};
\node[metricbox] (V5) at (4.5, -6.5) {V5\\{\scriptsize Execution}};

\node[font=\scriptsize\bfseries, above=0.15cm of V1] {\textit{Level 1 metric}};

%% ── Arrows P→─────────────────────────────────────────────────────
\draw[arr]  (P1.east) -- (V1.west);
\draw[arr]  (P2.east) -- (V2.west);
\draw[arr]  (P3.east) -- (V3.west);
\draw[darr] (P4.east) -- (ST.west);
\draw[arr]  (P5.east) -- (V4.west);
\draw[arr]  (P6.east) -- (V5.west);

%% ── Column 3: Normalization ─────────────────────────────────────────────
%  Height covers V1 (y=0) to V5 (y=-6.5): span = 6.5, centre = -3.25
%  Add half-box height (~0.35) top and bottom → height ≈ 7.2cm
\node[normbox, minimum height=7.2cm, minimum width=1.6cm] (NORM)
  at (7.6, -3.25) {\rotatebox{90}{\small Normalize $\hat{m}_j\in[0,1]$}};

\foreach \m in {V1,V2,V3,V4,V5}{
  \draw[arr] (\m.east) -- (NORM.west |- \m.east);
}

%% ── Column 4: VCS ───────────────────────────────────────────────────────
%  Centred at same y as NORM
\node[vcsbox] (VCS) at (11.3, -3.25) {
  \textbf{VCS}\\[3pt]
  {\scriptsize $\mathrm{VCS}^{\mathrm{mult}} = \prod_j \hat{m}_j$}\\[2pt]
  {\scriptsize (weakest link)}\\[6pt]
  {\scriptsize $\mathrm{VCS}^{\mathrm{add}} = \sum_j \omega_j \hat{m}_j$}\\[2pt]
  {\scriptsize (additive)}
};

\node[font=\scriptsize\bfseries, above=0.15cm of VCS] {\textit{Score}};

\draw[arr] (NORM.east) -- (VCS.west);

\end{tikzpicture}
\caption{Vault credit score architecture.}
\label{fig:vcs_architecture}
\end{figure}
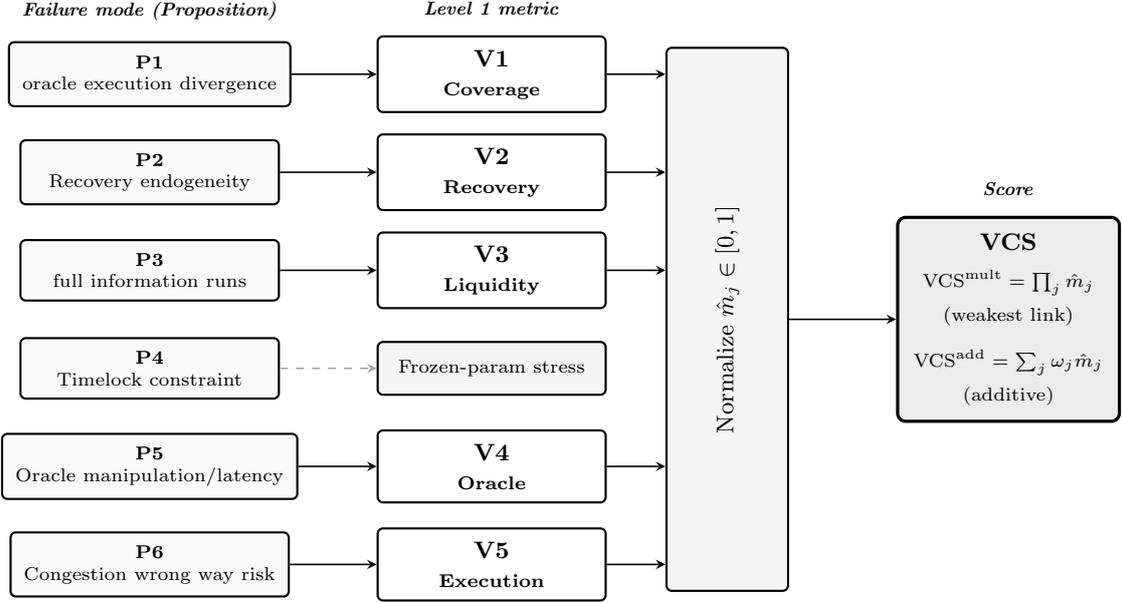
% ─────────────────────────────────────────────────────────────────────────

\section{Empirical implementation}
\label{sec:empirical}

\subsection{Data architecture}
Implementation of V1-V5 requires a three layer data stack covering onchain protocol state, oracle/reference prices, and market microstructure plus network conditions. Layer~1 contains vault accounting and liquidation logs sufficient to reconstruct $L_t^{v}$, $B_t^{v}$, $D_t^{v}$, and execution outcomes. Layer~2 contains oracle prices and update timestamps and a reference price proxy $P_{a,t}^{\mathrm{ref}}$ for bounding oracle error. Layer~3 contains DEX depth snapshots, swap logs for reconstructing $\widetilde{P}_{a,t}$ and gas/priority fee data. A prerequisite collateral evaluation (Remark~\ref{rem:collateral_prereq} and Table~\ref{tab:collateral_taxonomy}) is required before the data architecture can be applied to a given vault. Figure~\ref{fig:estimation_pipeline} provides an overview of the full estimation pipeline from raw data to VCS and Figure~\ref{fig:data_deps} maps data dependencies for each metric.

% ── FIGURE 3: Estimation Pipeline ────────────────────────────────────────
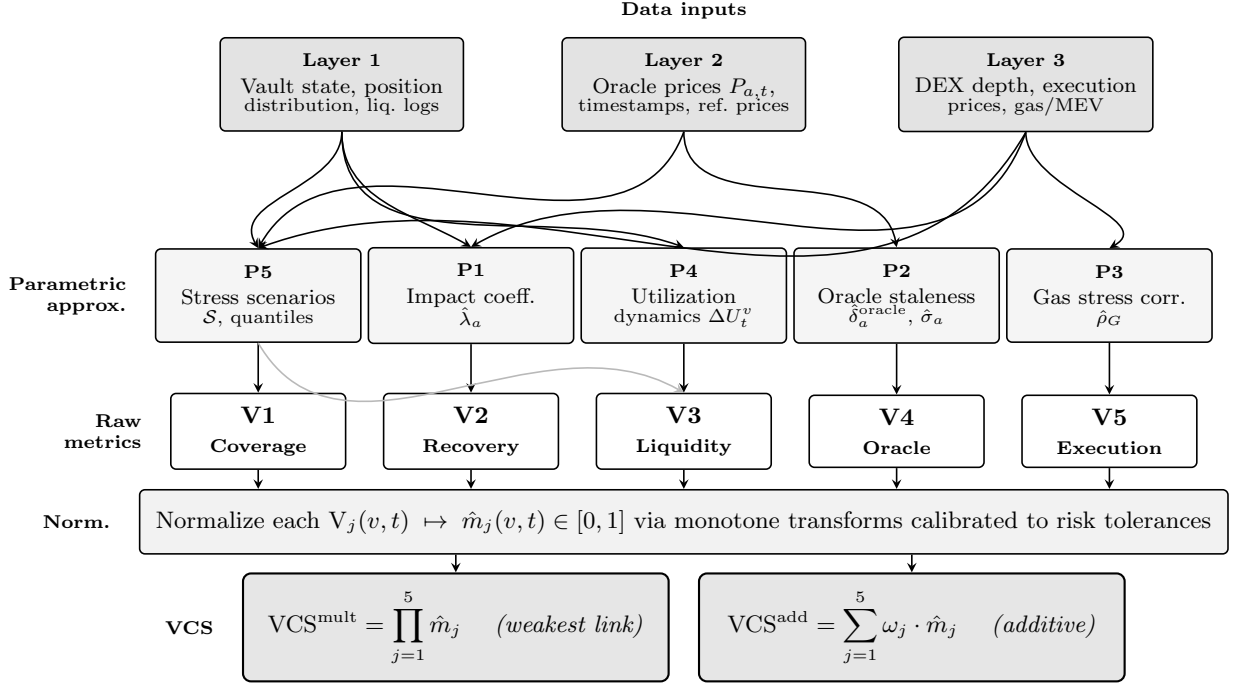
\begin{figure}[htbp]
\centering
\begin{tikzpicture}[
  databox/.style={
    draw=black, semithick, rounded corners=2pt,
    minimum width=3.2cm, minimum height=1.1cm,
    inner sep=6pt, align=center, font=\scriptsize,
    fill=gray!22
  },
  procbox/.style={
    draw=black, semithick, rounded corners=2pt,
    minimum width=2.7cm, minimum height=1.1cm,
    inner sep=6pt, align=center, font=\scriptsize,
    fill=gray!8
  },
  rawbox/.style={
    draw=black, semithick, rounded corners=2pt,
    minimum width=2.3cm, minimum height=0.85cm,
    inner sep=5pt, align=center, font=\small\bfseries,
    fill=white
  },
  normstep/.style={
    draw=black, semithick, rounded corners=2pt,
    minimum width=13.5cm, minimum height=0.85cm,
    inner sep=6pt, align=center, font=\small,
    fill=gray!10
  },
  vcsstep/.style={
    draw=black, thick, rounded corners=3pt,
    minimum width=5.6cm, minimum height=1.0cm,
    inner sep=6pt, align=center, font=\small,
    fill=gray!20
  },
  arr/.style={->, >=stealth, semithick},
  arrgray/.style={->, >=stealth, semithick, gray!55},
  lbl/.style={font=\scriptsize\bfseries, align=right}
]

%% ── Row 1: Data inputs (y = 0) ──────────────────────────────────────────
%  Nudged slightly wider (±4.5) to balance the wider procedure row below
\node[databox] (D1) at (-4.5, 0)
  {\textbf{Layer 1}\\[2pt]\footnotesize Vault state, position\\[-1pt]distribution, liq.\ logs};
\node[databox] (D2) at (0, 0)
  {\textbf{Layer 2}\\[2pt]\footnotesize Oracle prices $P_{a,t}$,\\[-1pt]timestamps, ref.\ prices};
\node[databox] (D3) at (4.5, 0)
  {\textbf{Layer 3}\\[2pt]\footnotesize DEX depth, execution\\[-1pt]prices, gas/MEV};

\node[font=\scriptsize\bfseries, above=0.12cm of D2] {Data inputs};

%% ── Row 2: Procedures (y = -2.8) ────────────────────────────────────────
%  Spacing increased from 2.4 to 2.8 units: x = ±5.6, ±2.8, 0
\node[procbox] (P5b) at (-5.6, -2.8)
  {\textbf{P5}\\[2pt]\footnotesize Stress scenarios\\[-1pt]$\mathcal{S}$, quantiles};
\node[procbox] (P1b) at (-2.8, -2.8)
  {\textbf{P1}\\[2pt]\footnotesize Impact coeff.\\[-1pt]$\hat\lambda_a$};
\node[procbox] (P4b) at (0, -2.8)
  {\textbf{P4}\\[2pt]\footnotesize Utilization\\[-1pt]dynamics $\Delta U_t^v$};
\node[procbox] (P2b) at (2.8, -2.8)
  {\textbf{P2}\\[2pt]\footnotesize Oracle staleness\\[-1pt]$\hat\delta_a^{\mathrm{oracle}}$, $\hat\sigma_a$};
\node[procbox] (P3b) at (5.6, -2.8)
  {\textbf{P3}\\[2pt]\footnotesize Gas stress corr.\\[-1pt]$\hat\rho_G$};

\node[lbl, left=0.25cm of P5b] {Parametric\\approx.};

%% ── Row 3: Raw metrics (y = -4.6) ───────────────────────────────────────
%  Same x-positions as procedures so arrows go straight down
\node[rawbox] (rV1) at (-5.6, -4.6) {V1\\{\scriptsize Coverage}};
\node[rawbox] (rV2) at (-2.8, -4.6) {V2\\{\scriptsize Recovery}};
\node[rawbox] (rV3) at (0,    -4.6) {V3\\{\scriptsize Liquidity}};
\node[rawbox] (rV4) at (2.8,  -4.6) {V4\\{\scriptsize Oracle}};
\node[rawbox] (rV5) at (5.6,  -4.6) {V5\\{\scriptsize Execution}};

\node[lbl, left=0.25cm of rV1] {Raw\\metrics};

%% ── Row 4: Normalization (y = -5.8) ─────────────────────────────────────
\node[normstep] (NORM) at (0, -5.8)
  {Normalize each $\mathrm{V}_j(v,t)\;\mapsto\;\hat{m}_j(v,t)\in[0,1]$
    via monotone transforms calibrated to risk tolerances};

\node[lbl, left=0.25cm of NORM.west] {Norm.};

%% ── Row 5: VCS (y = -7.2) ───────────────────────────────────────────────
\node[vcsstep] (VCSMUL) at (-3.0, -7.2)
  {$\mathrm{VCS}^{\mathrm{mult}} = \displaystyle\prod_{j=1}^5 \hat{m}_j$
   \quad\textit{(weakest link)}};
\node[vcsstep] (VCSADD) at (3.0, -7.2)
  {$\mathrm{VCS}^{\mathrm{add}} = \displaystyle\sum_{j=1}^5 \omega_j\cdot\hat{m}_j$
   \quad\textit{(additive)}};

\node[lbl, left=0.25cm of VCSMUL.west] {VCS};

%% ── Arrows: Data → Procedures ───────────────────────────────────────────
\draw[arr] (D1.south) to[out=270, in=120] (P5b.north);
\draw[arr] (D1.south) to[out=270, in=150] (P1b.north);
\draw[arr] (D1.south) to[out=270, in=160] (P4b.north);
\draw[arr] (D2.south) to[out=270, in=100] (P2b.north);
\draw[arr] (D2.south) to[out=250, in=60]  (P5b.north);
\draw[arr] (D3.south) to[out=270, in=30]  (P3b.north);
\draw[arr] (D3.south) to[out=250, in=30]  (P1b.north);
\draw[arr] (D3.south) to[out=240, in=20]  (P5b.north);

%% ── Arrows: Procedures → Metrics (straight down) ────────────────────────
\draw[arr] (P5b.south) -- (rV1.north);
\draw[arr] (P1b.south) -- (rV2.north);
\draw[arr] (P4b.south) -- (rV3.north);
\draw[arr] (P2b.south) -- (rV4.north);
\draw[arr] (P3b.south) -- (rV5.north);
\draw[arrgray] (P5b.south) to[out=300, in=150] (rV3.north);

%% ── Arrows: Metrics → Normalization ─────────────────────────────────────
\foreach \m in {rV1, rV2, rV3, rV4, rV5}{
  \draw[arr] (\m.south) -- (NORM.north -| \m.south);
}

%% ── Arrows: Normalization → VCS ─────────────────────────────────────────
\draw[arr] (NORM.south -| VCSMUL.north) -- (VCSMUL.north);
\draw[arr] (NORM.south -| VCSADD.north) -- (VCSADD.north);

\end{tikzpicture}
\caption{Estimation pipeline for vault credit score.}
\label{fig:estimation_pipeline}
\end{figure}

\begin{table}[htbp]
\centering
\caption{Level 1 data architecture. }
\label{tab:data_architecture_L1}
\begin{tabular}{>{\centering\arraybackslash}p{0.24\textwidth}
>{\centering\arraybackslash}p{0.03\textwidth} 
>{\centering\arraybackslash}p{0.22\textwidth} 
>{\centering\arraybackslash}p{0.12\textwidth} 
>{\centering\arraybackslash}p{0.14\textwidth} 
>{\centering\arraybackslash}p{0.18\textwidth}}
\toprule
\textbf{Data item} & \textbf{Layer} & \textbf{Source} & \textbf{Granularity} & \textbf{Depth required} & \textbf{Gaps/issues} \\
\midrule
$L_t^{v}$, $B_t^{v}$, $D_t^{v}$ & 1 & Contract state/events/subgraph & hourly--daily & $\ge 1$ year & Accounting heterogeneity\\
$(C_{u,a,t},B_{u,d,t})$ distribution & 1 & State reconstruction from events & daily--hourly & $\ge 6$ months & Storage/attribution cost \\
$\Theta_t^{v}$ (parameter state) & 1 & Governance execution tx and state & event/block level & full history & Proxy upgrades\\
Liquidation logs & 1 & Events and tx traces & tx level & $\ge 6$ months & Routing inference\\
Oracle prices $P_{a,t}$ and updates & 2 & Oracle contract events/state & block--hourly & $\ge 1$ year & Vendor heterogeneity\\
Reference prices $P_{a,t}^{\mathrm{ref}}$ & 2 & CEX OHLCV/independent TWAP & hourly--daily & $\ge 2$ years & Proxy error\\
DEX depth $D_{a,t}^{\mathrm{pool}}$ & 3 & Pool state (reserves/ticks) & block--hourly & $\ge 6$ months & Fragmented liquidity\\
Execution prices $\widetilde{P}_{a,t}$ & 3 & Swap logs and liquidation traces & tx level & $\ge 6$ months & Multi-hop routing\\
Gas prices $G_t$ and congestion & 3 & Block headers/receipts & block--hourly & $\ge 1$ year & crosschain comparability \\
MEV proxies & 3 & Priority fees/builder data & block level & $\ge 6$ months & Proxy validity\\
$\mathrm{HD}(a)$, $\mathrm{CM}(a)$ (rehypothecation metadata) & 1 & Contract code review/protocol docs & static & per asset & Recursive structure tracing\\
\bottomrule
\end{tabular}
\end{table}

% ── FIGURE 4: Data-Dependency Matrix ──────────────────────────────
%
%  Coordinate system:
%    grid left  = 5.0
%    column j centre: cx(j) = 5.9 + (j-1)*1.8   j = 1..5
%    grid right = 14.0
%    row i centre: ry(i) = -(i-1)*0.9             i = 1..6
%    grid top = 0.45,  grid bottom = -4.95
%
%  Gap chain (all verified clear):
%    row label east (4.0)  →  0.33 cm  →  layer tag left (4.33)
%    layer tag right (4.77) →  0.23 cm  →  grid left (5.0)
%
\begin{figure}[htbp]
\centering
\begin{tikzpicture}

%% ── Column headers ──────────────────────────────────────────────────────
\foreach \j/\name in {1/V1, 2/V2, 3/V3, 4/V4, 5/V5}{
  \node[font=\small\bfseries]
    at ({5.9 + (\j-1)*1.8}, 0.75) {\name};
}

%% ── Row labels + layer tags ─────────────────────────────────────────────
\foreach \i/\lyr/\lbl in {
  1/1/{Vault balance sheet},
  2/2/{Oracle prices and update timestamps},
  3/1/{Liquidation logs},
  4/3/{DEX pool depth and execution prices},
  5/3/{Gas prices and MEV proxies},
  6/2/{Reference prices}
}{
  \node[font=\small, text width=3.8cm, align=right, anchor=east]
    at (4.0, {-(\i-1)*0.9}) {\lbl};
  \node[font=\scriptsize\bfseries, text=gray!60, anchor=center]
    at (4.55, {-(\i-1)*0.9}) {L\lyr};
}

%% ── Grid: outer border ──────────────────────────────────────────────────
\draw[black, thick] (5.0, 0.45) rectangle (14.0, -4.95);

%% ── Grid: internal horizontal rules ────────────────────────────────────
\foreach \i in {1,...,5}{
  \draw[gray!40, thin] (5.0, {0.45 - \i*0.9}) -- (14.0, {0.45 - \i*0.9});
}

%% ── Grid: internal vertical rules ──────────────────────────────────────
\foreach \j in {1,...,4}{
  \draw[gray!40, thin] ({5.0 + \j*1.8}, 0.45) -- ({5.0 + \j*1.8}, -4.95);
}

%% ── Dependency dots ─────────────────────────────────────────────────────
%    row 1  Vault balance:   V1 V2 V3 V4 V5
%    row 2  Oracle prices:   V1 V2    V4 V5
%    row 3  Liq. logs:       V1 V2 V3 V4 V5
%    row 4  DEX depth:       V1 V2    V4 V5
%    row 5  Gas/MEV:                   V5
%    row 6  Ref. prices:     V1          V4
%
\foreach \i/\j in {
  1/1, 1/2, 1/3, 1/4, 1/5,
  2/1, 2/2,      2/4, 2/5,
  3/1, 3/2, 3/3, 3/4, 3/5,
  4/1, 4/2,      4/4, 4/5,
                 5/5,
  6/1,           6/4
}{
  \fill[black] ({5.9 + (\j-1)*1.8}, {-(\i-1)*0.9}) circle (0.07cm);
}

\end{tikzpicture}
\caption{Data dependencies of Level 1 metrics. A filled circle
indicates that the requires the corresponding data series for
computation. Layer labels L1-L3 correspond to the three layer data architecture in Table~\ref{tab:data_architecture_L1}.}
\label{fig:data_deps}
\end{figure}
% ─────────────────────────────────────────────────────────────────────────
\label{rem:collateral_prereq}
Metrics V1-V5 can be computed only when the data objects they require (DEX
depth, oracle update logs, liquidation traces, gas data) are available and
interpretable for each collateral asset $a\in\mathcal{C}^{v}$.
Table~\ref{tab:collateral_taxonomy} organizes the prerequisite collateral
assessment into five dimensions, each of which maps to a specific or
estimation procedure. An asset that fails technical or oracle quality
requirements may render one or more metrics non computable or subject to partial
identification bounds wider than those in
Propositions~\ref{prop:PI1_oracle_bound} and~\ref{prop:PI2_liquidation_bound}.
In that case, the corresponding should be assigned a conservative
worst case score.

\begin{table}[htbp]
\centering
\caption{Collateral due diligence taxonomy.}
\label{tab:collateral_taxonomy}
\begin{tabular}{p{0.18\textwidth} p{0.33\textwidth} p{0.26\textwidth} p{0.17\textwidth}}
\toprule
\textbf{Dimension} & \textbf{Key question} & \textbf{Formal risk object} & \textbf{Primary metric} \\
\midrule
Technical and admin risk &
Are admin keys capable of arbitrary minting, freezing, or non timelock upgrades? Has the contract system been audited by reputable firms? &
Contract level governance integrity, upgrade risk to $\mathcal{E}^{v}$ &
V5 (execution disruption)\\
Market microstructure &
What is DEX depth $D_{a,t}^{\mathrm{pool}}$ across major pairs? What slippage at 1\%, 3\%, 5\% price impact notionals? Does depth hold under stress? &
$D_{a,t}^{\mathrm{pool}}$ in V2, $\lambda_a$ in Procedure~P1 &
V2 (impact function), V1 (stress slippage) \\
Oracle quality &
Is a robust, decentralized oracle available? Are there reference market cross-checks? Is the oracle susceptible to TWAP manipulation via thin reference markets? &
$\eta_{a,t}$, $\delta_a^{\mathrm{oracle}}$, manipulation cost in V4 &
V4a (latency), V4b (manipulation) \\
Correlation and co-movement &
What is $\mathrm{CLR}_t^{v}$ for this asset and the vault's debt set? Does the collateral to debt price ratio exhibit stress time widening episodes? &
$\mathrm{CLR}_t^{v}$ (Remark~\ref{rem:leverage_loop}), liquidation clustering &
V2 (amplified liquidation mass) \\
Rehypothecation and recursive leverage &
Is this asset itself a tokenized claim on a leveraged vault? What is the unwinding depth? &
Rehypothecation depth $\mathrm{HD}(a)$, cascade multiplier $\mathrm{CM}(a)$ (Remark~\ref{rem:rehypothecation}) &
V1 (cascade inflated shortfall), V2 (amplified liquidation mass) \\
\bottomrule
\end{tabular}
\end{table}

\begin{remark}[Rehypothecation depth and cascade shortfall amplification]
\label{rem:rehypothecation}
For any collateral asset $a\in\mathcal{C}^{v}$ that is itself a share token
of another vault $v'$, such as a yield bearing stablecoin backed by a
lending strategy or a tokenized structured product, define the
\emph{rehypothecation depth} $\mathrm{HD}(a)$ recursively: $\mathrm{HD}(a)=1$
for a DeFi native primitive asset, and $\mathrm{HD}(a)=1+\mathrm{HD}(a')$ for
share tokens whose underlying asset is $a'$. A shock of fractional magnitude
$\epsilon$ to the base layer asset propagates through the chain of vaults,
generating a vault level shortfall at vault $v$ proportional to the cascade
multiplier $\mathrm{CM}(a)=\Lambda^{\mathrm{HD}(a)-1}$, where $\Lambda$ is the
leverage multiplier in \eqref{eq:leverage_multiplier}. This formula applies a uniform leverage multiplier $\Lambda$ across all $\mathrm{HD}(a)-1$ layers, in heterogeneous structures where each layer $j$ has its own $\Lambda_j$, the cascade multiplier generalises to $\mathrm{CM}(a)=\prod_{j=1}^{\mathrm{HD}(a)-1}\Lambda_j$, which should be used when per layer LTV data are available. For credit risk measurement, $\mathrm{HD}(a)$ and $\mathrm{CM}(a)$ are required metadata for any collateral asset with $\mathrm{HD}(a)>1$, they should be included in the data architecture (Table~\ref{tab:data_architecture_L1}) and used to scale V1 and V2 inputs conservatively when the base layer shock cannot be directly observed.
Data sources for $\mathrm{HD}$ tracing are protocol documentation and smart
contract code review; partial observability of recursive structures is a known
data gap (see Table~\ref{tab:data_architecture_L1}).
\end{remark}

\begin{remark}[Scope and governance inputs]
This article restricts implementation to Level 1 structural measurement. Full implementation of a joint vault and governance score would additionally require offchain governance documentation, fee structures, and conflict disclosures (those inputs are analyzed in the companion paper). For this paper, governance enters only through the fixed parameter vector $\Theta_t^{v}$ and the timelock constraint \eqref{eq:curator_lag}, which determines whether parameters can be treated as constant over stress horizons.
\end{remark}

\subsection{Estimation of core parameters}
This subsection provides operational estimators for the parameters required by V1-V5 including impact coefficients $\lambda_a$, oracle staleness and latency volatility exposure, gas stress correlation $\rho_G$, utilization dynamics and jumps, and scenario inputs.

\begin{procedure}[Estimation of price impact coefficients $\lambda_a$ for V2]
\label{proc:P1_lambda}
Let $\widehat{\varepsilon}_{a,\tau}=1-\widetilde{P}_{a,\tau}/P_{a,\tau}$ denote realized execution deviation for collateral asset $a$ at liquidation timestamp $\tau$, and let $D_{a,\tau}^{\mathrm{pool}}$ denote contemporaneous depth. Use liquidation events as quasi-exogenous sell flow observations and estimate $\lambda_a$ via:
\begin{equation}
\widehat{\varepsilon}_{a,\tau}=\hat{\lambda}_a\cdot \frac{Q_{a,\tau}^{\mathrm{liq}}}{D_{a,\tau}^{\mathrm{pool}}}
+\hat{\beta}\cdot\mathrm{CLR}_{\tau^-}^{v}+u_{\tau}
\label{eq:lambda_regression}
\end{equation}
where $Q_{a,\tau}^{\mathrm{liq}}$ is liquidation volume in asset units and $u_{\tau}$ is an error term. In vaults with a material share of leverage loop borrowing, an additional
diagnostic specification should include $\mathrm{CLR}_{\tau^-}^{v}$
(Remark~\ref{rem:leverage_loop}) as a covariate for
$\widehat{\varepsilon}_{a,\tau}$, since correlated collateral composition can
amplify realized slippage beyond what the single asset depth ratio $Q/D$
predicts. The coefficient $\hat{\beta}$ estimates the marginal contribution of
collateral correlation to impact and can be used to scale conservative V2
estimates under stress. Inference must be heteroskedasticity robust and clustered/HAC due to temporal clustering of liquidations. If sample sizes are thin, implement partial pooling with fixed effects.
\end{procedure}

\begin{procedure}[Estimation of oracle staleness $\delta_a^{\mathrm{oracle}}$ and latency volatility risk for V4a]
\label{proc:P2_oracle}
Estimate $\delta_a^{\mathrm{oracle}}(t)$ directly from oracle update logs as elapsed time since last update at time $t$. Estimate realized volatility $\widehat{\sigma}_a(t)$ from a reference price series $P_{a,t}^{\mathrm{ref}}$ over rolling window $W$ using an hourly close to close estimator as:
\begin{equation}
\widehat{\sigma}_a^2(t)=\frac{1}{W-1}\sum_{j=1}^{W-1}\left(\ln P_{a,t-j\Delta}^{\mathrm{ref}}-\ln P_{a,t-(j+1)\Delta}^{\mathrm{ref}}\right)^2
\end{equation}
with $\Delta=1$ hour. Under a GBM latency model, approximate the false solvency probability that oracle error exceeds a threshold $\bar{\eta}>0$ as:
\begin{equation}
\Pr\!\left(\eta_{a,t}>\bar{\eta}\right)\approx \Phi\!\left(-\frac{\bar{\eta}}{\widehat{\sigma}_a(t)\cdot\sqrt{\delta_a^{\mathrm{oracle}}(t)}}\right)
\label{eq:oracle_false_solvency_prob}
\end{equation}
where $\Phi(\cdot)$ is the standard normal CDF.
\end{procedure}

\begin{procedure}[Estimation of gas stress correlation $\rho_G$ for V5]
\label{proc:P3_rhog}
Let $G_t$ denote hourly average gas price and let $S_t$ denote a scalar stress measure, e.g., absolute collateral return aggregated across $\mathcal{C}^{v}$. Estimate the stress conditional correlation over a rolling window $W$ by restricting to the top $q$ percentile of $S_t$:
\begin{equation}
\hat{\rho}_G(t)=\mathrm{Corr}\!\left(G_s,S_s \,\big|\, S_s\ge \widehat{Q}_{S}(q)\right)_{s\in[t-W,t]}
\label{eq:rho_G_estimator}
\end{equation}
where $\widehat{Q}_{S}(q)$ is the empirical $q$-quantile.
\end{procedure}

\begin{procedure}[Estimation of utilization dynamics and jump risk for V3]
\label{proc:P4_util}
Estimate hourly utilization $U_t^{v}=B_t^{v}/D_t^{v}$ and compute increments $\Delta U_t^{v}$. Estimate drift and diffusion over window $W$ and identify jump candidates via $|\Delta U_t^{v}|>n\cdot\widehat{\sigma}_U$ with $n$ calibrated to the vault's observed utilization volatility regime. Compute $\widehat{\mathrm{V3}}(v,t;h)$ by Monte Carlo simulation using the fitted diffusion and empirical jump distribution.
\end{procedure}

\begin{procedure}[Point estimation inputs for stress scenario construction $\mathcal{S}$]
\label{proc:P5_sceninputs}
Compute empirical quantiles of collateral drawdowns over horizons aligned to stress windows and compute empirical depth and gas quantiles over identified stress periods. These estimates parameterize the historical and parascenario classes in Section~\ref{subsec:stress_scenarios}.
\end{procedure}

\begin{procedure}[Yield attribution and sustainability assessment]
\label{proc:P6_yield}
For protocol bound lending vaults within the scope of V1-V5, $Y_h^{\mathrm{basis}} = 0$ by construction. The vault does not take derivative positions and any funding rate carry embedded in collateral tokens, e.g., synthetic stablecoins, belongs to the collateral issuer's structure rather than to the vault yield. The $Y_h^{\mathrm{basis}}$ component is retained in the decomposition for completeness and for future application to cross protocol vaults incorporating delta neutral strategies.
Observed vault yield is a composite of heterogeneous sources with materially
different risk and sustainability profiles. Define the yield decomposition over
horizon $h$:
\begin{equation}
Y_h = Y_h^{\mathrm{organic}} + Y_h^{\mathrm{incentive}}
+ Y_h^{\mathrm{basis}} + Y_h^{\mathrm{arb}}
\label{eq:yield_decomp}
\end{equation}
where $Y_h^{\mathrm{organic}}$ denotes yield from structural sources such as lending rates, AMM fees and staking rewards. $Y_h^{\mathrm{incentive}}$ denotes yield from governance token emission programs (typically unsustainable when programs end or token prices fall). $Y_h^{\mathrm{basis}}$ denotes funding rate carry from delta neutral strategies (sustainable subject to hedge integrity) and
$Y_h^{\mathrm{arb}}$ denotes temporary yield from pricing inefficiencies (by
definition transitory). A vault's yield is sustainable at target level
$Y^{\mathrm{target}}$ if $Y_h^{\mathrm{organic}} \ge Y^{\mathrm{target}}-\epsilon$ for a conservative buffer $\epsilon>0$. Vaults where
$Y_h^{\mathrm{organic}} \ll Y^{\mathrm{target}}$ are functionly dependent on
incentive programs, a sudden collapse in $Y_h^{\mathrm{incentive}}$ can trigger
a withdrawal run as depositors reassess expected yield, implying that incentive
program endings should be modeled as jump drivers in V3's utilization dynamics
alongside price shocks. To estimate $Y_h^{\mathrm{organic}}$ versus
$Y_h^{\mathrm{incentive}}$, decompose total observed yield into protocol rate
income (from lending rate logs and pool fee records) and emission income (from
onchain emission distribution events). The ratio
$Y_h^{\mathrm{organic}}/Y_h$ is an indicator of structural yield quality and
should be reported alongside VCS as a forward looking depositor risk signal.
\end{procedure}

\begin{remark}[Pooling versus protocol specific estimation]
For protocols with limited liquidation histories, cross protocol pooling of $\lambda_a$ and stress conditional $\rho_G$ estimates can stabilize inference but introduces model risk due to heterogeneity in liquidation mechanisms, routing, and chain conditions. Pooling is conservative when pooled samples overweight high impact venues, but can be optimistic if pooled samples underrepresent tail episodes relevant to the target vault. A principled pooling strategy requires hierarchical modeling, thus in practice, pooled estimates should be accompanied by sensitivity ranges.
\end{remark}

\subsection{Partial identification and bounds}
Two key objects cannot be point identified from standard onchain observables: (i) oracle error relative to the latent efficient price and (ii) liquidation failure probabilities from observed liquidation outcomes alone. Therefore, the framework uses identified set reasoning to provide conservative bounds.

\begin{proposition}[Partial identification of oracle error bounds from reference spreads]
\label{prop:PI1_oracle_bound}
Let $P_{a,t}^{\mathrm{ref}}$ be an observable reference price proxy and define the spread $S_{a,t}\equiv P_{a,t}-P_{a,t}^{\mathrm{ref}}$. Write $P_{a,t}^{\mathrm{ref}}=P_{a,t}^{\mathrm{true}}+\zeta_{a,t}$, where $\zeta_{a,t}$ is reference proxy error, so $S_{a,t}=\eta_{a,t}-\zeta_{a,t}$. Oracle error variance is not point identified. If $\zeta_{a,t}$ is mean zero and independent of $\eta_{a,t}$, then:
\begin{equation}
\mathrm{Var}(\eta_{a,t}) = \mathrm{Var}(S_{a,t})+\mathrm{Var}(\zeta_{a,t}) \ \ge \ \mathrm{Var}(S_{a,t})
\end{equation}
so $\mathrm{Var}(S_{a,t})$ provides a lower bound. The sign of systematic oracle bias is identified from $\mathbb{E}[S_{a,t}]$, and lead lag structure from cross correlations of $(P_{a,t},P_{a,t}^{\mathrm{ref}})$. Latency risk components of V4 can therefore be bounded from below using the reference spread, with tightness determined by the magnitude of $\mathrm{Var}(\zeta_{a,t})$.
\end{proposition}

\begin{proposition}[Partial identification of lower bound on liquidation failure probability]
\label{prop:PI2_liquidation_bound}
Let $\mathcal{L}_{u,t}$ be the liquidation trigger event and $\mathcal{V}_{u,t}$ be the viability event. Liquidation outcomes observe execution events but do not cleanly observe counterfactual non executions because self repayment prior to execution is observationally similar to failed liquidation. Let $\tau_{\max}>0$ be a maximum execution window and define the observable event $\mathcal{U}_{u,t}(\tau_{\max})$ that an account remains undercollateralized for longer than $\tau_{\max}$ without liquidation. Then
\begin{equation}
\Pr(\mathcal{V}_{u,t}^{c}\mid \mathcal{L}_{u,t}) \ \ge\ \Pr(\mathcal{U}_{u,t}(\tau_{\max})\mid \mathcal{L}_{u,t})
\end{equation}
so the right-hand side provides an identified lower bound. The bound is tighter for larger $\tau_{\max}$ but becomes susceptible to survivorship and intervention bias.
\end{proposition}

\begin{remark}[Partial identification and what it implies operationally]
The identified bounds are decision relevant even when point identification fails. For oracle integrity (V4), a lower bound on oracle error variance implies a conservative upper bound on the attainable integrity score under any plausible latent price model. For execution viability (V5), the lower bound on failure probability implies that any point estimate based solely on executed liquidations is optimistic unless it incorporates undercollateralization persistence..
\end{remark}

\subsection{Stress scenario construction}
\label{subsec:stress_scenarios}
Stress scenario construction provides the scenario set $\mathcal{S}$ used in V1 and V3 and informs V5 stress conditioning. The central requirement is internal consistency, a scenario $s$ must specify a joint realization of collateral price shocks, depth contractions, utilization stress and gas congestion.
Historical scenarios are constructed from empirical tail events in the vault’s history or in comparable protocols. Include event $s$ if collateral weighted drawdown exceeds a historical quantile threshold $q_H$, and parameterize the scenario by $(P_{a,t}^{(s)},D_{a,t}^{(s)},G_t^{(s)})$ observed in that episode (recommended $q_H=0.9$ over trailing two years).

Parascenarios specify a grid over price drawdown magnitudes and depth reductions. For each asset $a$, select drawdown shocks $\Delta P_a$ and depth shocks $\Delta D_a/D_a$ on a grid and evaluate $\mathrm{V1}(v,t;s)$ and $\mathrm{V3}(v,t;h;s)$ over the grid. Depth reduction baselines such as 50\% must be calibrated to the vault's observed depth regime.
Adversarial scenarios define a robust worst case within a feasibility set $\mathcal{Z}^{\mathrm{adv}}$, e.g., shocks bounded by $k$ standard deviations:
\begin{equation}
s^\ast = \arg\min_{s\in\mathcal{Z}^{\mathrm{adv}}} \ \mathrm{V1}(v,t;s)
\end{equation}
Cross-consistency requires that scenarios specify a joint quadruple $(P^{(s)},D^{(s)},U^{(s)},G^{(s)})$ that is dynamically attainable. A minimal consistency condition is that the implied utilization state is compatible with the balance sheet and liquidation mechanics under the specified price and depth shocks:
\begin{equation}
U_t^{(s)} \le \frac{B_t^{v}+\Delta B_t^{(s)}}{D_t^{v}-\Delta D_t^{(s)}}
\label{eq:scenario_consistency}
\end{equation}
where $\Delta B_t^{(s)}$ and $\Delta D_t^{(s)}$ denote scenario implied increments from borrow demand and deposit outflows, as induced by the stress dynamics and require protocol specific modeling of borrow and withdrawal elasticities.

\begin{remark}[Tail scenario coverage and unknown unknowns]
Historical and parascenarios omit structural breaks such as novel collateral types, chain level outages and unprecedented venue fragmentation. Adversarial scenarios partially address this by optimizing within $\mathcal{Z}^{\mathrm{adv}}$, but outcomes outside that set require expert judgment and should be treated as residual risk. In particular, oracle manipulation scenarios depend on oracle mechanism constraints and cannot be reduced to price shocks alone.
\end{remark}

\subsection{Validation and backtesting}
Validation tests whether V1 and V2 have predictive content for realized shortfalls and whether estimated parameters behave consistently with their mechanism interpretation. Because depositor shortfalls are rare, validation must leverage cross-vault variation within common stress episodes as well as within vault time series behavior. Remark~\ref{rem:strat_validation} specifies the additional validation requirements for cross protocol vaults with strategy specific risks outside the V1-V5 scope.

For V1, define a backtest at stress episode start time $t_i$ with horizon $h$ (aligned to $\Delta^{*v}$ or a fixed horizon such as 24 hours). Let $\mathbf{1}\{\Delta_{t_i+h}^{v}>0\}$ indicate realized shortfall. Use $\widehat{\mathrm{V1}}(v,t_i;s_i)$ as a predictor in a calibration regression:
\begin{equation}
\Pr\!\left(\Delta_{t_i+h}^{v}>0\right)=\mathrm{logit}^{-1}\!\left(\alpha_0+\alpha_1\,\widehat{\mathrm{V1}}(v,t_i;s_i)\right)
\label{eq:backtest_regression}
\end{equation}
and test $\alpha_1<0$ as the directional implication. For V2, relate predicted volume conditioned shortfall to realized shortfall severity in liquidation clusters by regressing realized loss rates on predicted $\widehat{\mathrm{V2}}$ under matched scenarios.

Unit validation of the critical response window $\widehat{\Delta}^{*v}$ is relevant because it governs whether parameters are effectively frozen over stress horizons (Proposition~\ref{prop:curator_lag}) and because it anchors horizon choice in stress testing. A direct diagnostic compares the frequency with which shortfalls occur within $\widehat{\Delta}^{*v}$ hours of stress onset to the frequency with which parameter changes (if any) complete within that window; the latter is not used for scoring in this paper but provides a timing sanity check for stress horizons:
\begin{equation}
\widehat{\mathrm{Gap}}^{v}\equiv
\frac{1}{N_e}\sum_{i=1}^{N_e}\mathbf{1}\{\Delta_{t_i+\widehat{\Delta}^{*v}}^{v}>0\}
-\frac{1}{N_e}\sum_{i=1}^{N_e}\mathbf{1}\{\ell_i\le \widehat{\Delta}^{*v}\}
\end{equation}
where $\ell_i$ is an observed protocol specific action latency when available.

\begin{remark}[Backtesting constraints and tail event scarcity]
DeFi lending protocols have short histories and a limited number of extreme stress episodes, limiting statistical power for out of sample tests. This data scarcity problem has a direct TradFi analog. \citet{TreacyCarey2000} document that even the 50 largest US banking organizations lacked sufficient loan level performance data by internal grade to empirically calibrate PD/LGD estimates, relying instead on mappings to agency grades and judgmental adjustments. A complementary calibration scarcity result from the latent factor literature is provided by \citet{DuffieEcknerHorelSaita2009}. Even with 25 years of monthly default data on 2,793 firms, the mean reversion rate $\kappa$ of their Ornstein--Uhlenbeck frailty process (the parameter governing the persistence of the common latent factor) was difficult to identify precisely ($\hat{\kappa} = 0.018$, standard error $= 0.004$), with substantial posterior uncertainty documented in their Bayesian analysis. This parameter is directly analogous to the persistence of the jump intensity in V3's utilization dynamics. The identification challenge in a 25 year corporate dataset is compounded many times over in DeFi where protocol histories rarely exceed five years, structural breaks are frequent, and failed vaults are systematically underrepresented in available datasets. Validation relies on pooling across vaults and chains within the same stress episodes and on mechanism level diagnostics, e.g., whether estimated $\lambda_a$ increases in stress regimes. Survivorship bias is material if failed vaults are missing from datasets. Curated panels that include protocol failures are a prerequisite for unbiased validation. A formal power analysis is deferred pending construction of a curated multi protocol panel that includes failures and near misses.
\end{remark}

\begin{remark}[Validation limitations for strategy specific risks]
\label{rem:strat_validation}
The present validation framework targets V1 (coverage stress test accuracy) and
V2 (impact function predictiveness) against observed liquidation outcomes.
Strategy specific risks that fall outside the protocol bound scope of V1-V5 (including impermanent loss for AMM-deployed strategies, auto-deleveraging events
in derivative hedged positions, and NAV latency gaps in RWA strategies) cannot
be directly backtested against the present metrics. For cross protocol vaults
incorporating these strategies, supplementary validation would require
strategy level P\&L attribution against the decomposition in Procedure~P6 and
separate empirical analysis of basis divergence and execution failure in the
relevant derivative venues.
\end{remark}

\section{Discussion}
\label{sec:discussion}

\subsection{Contribution to existing literature}
The paper occupies a specific gap in the literature where prior work provides either qualitative decompositions of DeFi risk categories or empirical analyses of individual mechanisms, but neither a formal instrument level loss function nor a mechanism grounded set of metrics derivable from it. The core contribution is a mechanism level decomposition precise enough to generate tractable metrics and estimation procedures rather than narrative diagnoses.

Relative to structured credit and CLO methodology, the principal addition is not the idea of coverage or triggers, but the formal decomposition of where these concepts fail in an onchain execution environment. Table~\ref{tab:mapping_summary} provides a systematic classification of six breakdown points that do not arise in institutional structured finance because valuation, settlement and default management operate under different frictions. In particular, Proposition~\ref{prop:network_congestion} formalizes congestion driven wrong way risk as an execution feasibility constraint that is endogenous to stress, a channel that cannot be captured by static haircuts and that has no direct analogue in standard OC/IC testing logic. The result is a structured finance like framework that remains faithful to the instrument logic of depositor protection while being native to DeFi execution.

The paper also provides the first formal integration of infrastructure and code integrity as a structurally distinct third layer in vault credit risk measurement. The three level decomposition of Definition~\ref{def:three_level} and equation~\eqref{eq:three_level} establishes that vault depositor losses arise from three structurally independent sources whose appropriate measurement methodologies are non overlapping with mechanical execution channels (Level 1), governance quality (Level 2) and smart contract code integrity (Level 3). Proposition~\ref{prop:L3_superadditive} establishes that Level 3 failure probability is superadditive in dependency depth, accumulating geometrically as protocol layers are added. Proposition~\ref{prop:L3_dominance} derives the dominance condition under which Level 3 risk exceeds Level 1 risk in expected value terms, providing a criterion for when VCS monitoring is the binding depositor protection measure and when it must be supplemented by independent code risk assessment. Corollary~\ref{cor:cross_protocol_L3} gives the protocol bound versus cross protocol vault distinction formal credit risk content. Additional dependency layers are structural Level 3 risk increments regardless of collateral or oracle configuration quality. Prior DeFi risk frameworks typically catalog smart contract risk as a qualitative category without deriving its formal relationship to the instrument level loss function. This paper, to the authors' knowledge, is the first to derive the dominance condition that determines which of the three levels constitutes the binding credit concern for any specific vault.

Relative to DeFi empirical work on liquidation mechanics and protocol fragility, the paper provides an instrument level loss mapping. Empirical research has shown that aggregate lending protocol states can signal fragility, with $R_I$ detecting stress episodes around the UST collapse and FTX bankruptcy, and that liquidation effectiveness degrades under elevated risk conditions, with lagged liquidation responses indicating execution does not keep pace with deterioration \citep{BertomeuMartinSall2024}. The paper explains why such degradation is structurally expected since execution is endogenous to liquidation mass (Proposition~\ref{prop:recovery_endogeneity}) and can fail under congestion (Proposition~\ref{prop:network_congestion}). It further provides metrics (V2 and V5) that operationalize these channels as estimable objects rather than descriptive regularities.

Relative to oracle manipulation and security analyses, the paper’s contribution is to integrate oracle imperfections into depositor loss measurement. Oracle research typically focuses on attack feasibility and protocol design defenses \citep{AngerisChitra2020, AngerisKaoChiangNoyesChitra2020, QinZhouLivshits2021}. Here, Proposition~\ref{prop:oracle_manipulation} and V4 treat oracle error as a direct driver of depositor shortfall through false solvency and delayed liquidation, and PI-1 formalizes the resulting identification limits. This integration is necessary because in DeFi, pricing inputs are not merely accounting marks, they directly trigger automated enforcement.

Relative to run dynamics and liquidity crisis analyses, the paper formalizes how utilization and common information can sharpen withdrawal clustering (Proposition~\ref{prop:liquidity_run}) and defines V3 as a boundary hitting probability that can be estimated from onchain time series. The emphasis is not that runs exist, but that run propensity is measurable as a forward looking probability conditional on the current utilization state and stress regimes. This provides a direct bridge from liquidity crisis narratives to a credit risk statistic.

Finally, emerging work on the oracle infrastructure underlying tokenized RWAs identifies the structural barriers to reliable offchain to onchain data transmission and the trust efficiency trade off that renders purely decentralized oracles infeasible for assets whose value is determined by offchain processes \citep{DuleyGambacortaGarrattKooWilkens2023}. The Financial Stability Board identifies oracle third party reliance as a core financial stability vulnerability of tokenized assets at scale, alongside liquidity mismatch, leverage and operational fragility \citep{FSB2024Tokenisation}. The legal and accountability dimensions of DeFi's departure from traditional enforceability mechanisms (particularly the absence of legal recourse when smart contract enforcement produces economically incorrect outcomes) are analyzed in \citet{ZetzscherArnerBuckley2020}. None of these contributions formalizes the specific credit risk channel by which NAV latency windows generate systematic, directional depositor shortfall exposure in DeFi lending vaults. Remarks~\ref{rem:rwa_oracle} and~\ref{rem:mkt_hours} in the paper provide the first such formal treatment, deriving the expected shortfall bound in equation~\eqref{eq:nav_shortfall} and identifying the conditions under which the partial identification result of Proposition~\ref{prop:PI1_oracle_bound} fails for RWA collateral.

\subsection{Limitations}
The first binding limitation is paradependence in the most stress relevant channels. V2 relies on a linear impact approximation and V4a relies on a GBM style latency mapping and V3 relies on a reduced form diffusion with jumps representation of utilization dynamics. These forms are most likely to fail exactly in crisis regimes as concentrated liquidity AMMs exhibit highly nonlinear marginal prices, liquidation routing can create discontinuities and crypto returns exhibit jumps and regime dependent volatility. The consequence is that V2 and V4a can be poorly calibrated when they matter most. The mitigation is explicit sensitivity analysis over alternative function forms and, where data allow, semi-paraor nonparaestimation of $f(Q,D)$ and latency to error mappings.

Leverage loop concentration provides a second channel through which the linear
impact model can systematically understate tail losses. When
$\mathrm{CLR}_t^{v}$ is high (Remark~\ref{rem:leverage_loop}), the collateral to debt price ratio is volatile and liquidation
clustering is amplified beyond what single asset depth ratios predict and the
effective impact coefficient becomes an increasing function of correlated
leverage concentration rather than a fixed constant. Conservative V2 estimation
should therefore use a stress adjusted coefficient
$\lambda_a\cdot(1+\gamma\cdot \mathrm{CLR}_t^{v})$, where $\gamma>0$ is a loop
amplification factor.

The second binding limitation is non stationarity. Rolling window estimators for $\lambda_a$, $\rho_G$ and utilization dynamics assume stability over the estimation window. In DeFi, structural breaks are frequent: liquidation mechanisms are upgraded, oracle architectures change, collateral sets evolve, and congestion regimes shift. The consequence is that pre-break parameter estimates can become misleading post break, and simple updating is not appropriate. The mitigation is structural versioning and event-study re-estimation following upgrades, with conservative degradation of confidence in post break estimates until sufficient post break data are observed.

A specific form of non stationarity that deserves separate treatment is oracle
immutability. For vaults on protocols where oracle addresses are set at market
deployment (Remark~\ref{rem:oracle_immut}), the effective staleness of the oracle design can only increase as reference markets evolve, while no governance action can correct it within the existing market. This is structural non stationarity in V4 that cannot be addressed by rolling re-estimation and requires instead conservative initial provisioning of oracle quality buffers. Additionally, the magnitude of required buffer depends on protocol lifetime and oracle market liquidity trajectory.

The third binding limitation is sparse tail event history and limited backtesting power. Severe depositor shortfalls are rare and DeFi protocol histories are short. The consequence is that stress scenario quantiles, depth shock baseline, and calibration thresholds are necessarily conservative but not statistically validated, and type I/type II error rates are unknown. It also implies that validation must rely on cross-sectional pooling across vaults and chains and on mechanism level diagnostics rather than purely on outcome frequency calibration. Simulation based validation can complement empirical backtesting, but it depends on the correctness of microstructure primitives. This actuarial challenge is not unique to the theoretical framing of this paper but is an industry recognized structural barrier. \citet{ZhouZhang2026} identify an analogous actuarial conundrum as one of eight primary impediments to DeFi insurance market development, concluding from structured focus group discussions with practitioners that DeFi insured perils lack the stationary, independent distributions that conventional underwriting requires and that Monte Carlo simulation is the only currently feasible pricing approach, with continuous time models applicable only after sufficient price history has accumulated. Their quantitative AHP evidence assigns an overall priority weight of 0.087 to the actuarial conundrum, below the liquidity (0.166) and regulation (0.149) conundrums, and consistent with the vault paper's finding that run risk (V3) and governance risk (Level 2) are the most actionable near term concerns. This independent practitioner consensus confirms that the vault paper's reliance on structural estimation and conservative bounding which is correct response to a structural property of DeFi as a nascent financial system.

\begin{remark}[Practical deployment consequences for Level 1]
Given these limitations, V1-V5 and VCS should be deployed as conservative diagnostics. Where point identification fails (PI-1 and PI-2), risk limits should be set using worst case values within identified bounds. Where structural breaks occur, model versions must be treated as distinct and comparisons of VCS across versions should be conditioned on changes in $\mathcal{E}^{v}$. Finally, the scarcity of tail data implies that conservative stress scenarios and sensitivity ranges are not optional, they are integral to making the framework robust.
\end{remark}

\subsection{Policy implications}
The framework implies a minimum transparency standard for vault risk measurement. Table~\ref{tab:data_architecture_L1} and Corollary~\ref{cor:VCS_checklist} show that without liquidation logs, oracle update timestamps, depth snapshots and gas data, V1-V5 cannot be computed in an auditable manner. Partial identification results PI-1 and PI-2 (Propositions~\ref{prop:PI1_oracle_bound} and \ref{prop:PI2_liquidation_bound}) imply that even with full onchain transparency, certain objects remain only partially identified. Nonetheless, non disclosure worsens the identified set by preventing even conservative bounding. A direct implication is that non disclosure should be treated conservatively, missing data should map to worst case score components within the feasible identified bounds.

An additional implication follows from the architecture of modern curated vaults. When vault governance separates parameter authority into at minimum a curator role (adjusting caps, LTV, and allocations) and a sentinel role (emergency only risk tightening with real time authority), the relevant disclosure obligation is at minimum threefold: (i) the static parameter configuration $\Theta_t^{v}$ and its history (for V1-V5 computation), (ii) the governance architecture documenting which actions are subject to timelocks and which are sentinel eligible (for assessing the effective critical response window $\Delta^{*v}$ and the applicability of frozen-parameter stress tests) and (iii) the oracle selection and configuration for each collateral
market, including whether oracles are immutable (for V4 estimation). For RWA collateral, a fourth disclosure obligation is required, real time NAV and impairment status from issuers, given the systematic bias identified in Remarks~\ref{rem:rwa_oracle} and \ref{rem:mkt_hours}. Non disclosure of any of these items should be treated as a worst case input to the relevant. Monitoring
tools that have emerged for DeFi vaults, including curator action trackers for major lending protocols and protocol decentralization dashboards \citep{WernerPerezGudgeon2022}, represent early stage infrastructure for making these disclosures auditable and machine readable. The paper's VCS framework provides the formal methodology that such tools should implement to move from descriptive monitoring to credit relevant risk measurement.

Governance design standards, incentive alignment requirements and systemic risk implications depend on Level 2 governance measurement and multivault extensions, which are outside the scope of this paper and developed in a companion paper.

\begin{remark}[Disclosure as an input to credit measurement]
Disclosure is not a governance preference in this framework but it is an input required to compute the objects that define depositor credit risk. Because V1-V5 depend on execution prices, depth, oracle staleness and congestion, a vault that does not provide or enable reconstruction of these data makes credit risk measurement infeasible. From a depositor protection perspective, infeasibility must be treated as adverse information rather than ignored. The relevant policy lever is, therefore, minimum machine readable disclosure.
\end{remark}

\section{Conclusion}
\label{sec:conclusion}
This paper develops a formal, auditable Level 1 credit risk measurement framework for DeFi lending vaults by treating the vault as a credit instrument and anchoring risk measurement in a depositor loss function. The central reason such a framework was missing is simple. While DeFi vaults resemble traditional secured credit instruments, the mapping from valuation to realized outcomes is altered by oracle mediated pricing, onchain execution, congestion costs, governance constraints and the absence of external liquidity support. Without a mechanism consistent framework, applying traditional credit tools can produce false confidence precisely when depositor losses are most likely.

The theoretical foundation is a set of six failure modes that identify the precise breakdown points of naive TradFi analogies. Oracle execution divergence shows that oracle based solvency can coexist with execution insolvency when liquidation prices are endogenously worse. Endogenous recovery shows that recoveries deteriorate with liquidation mass, implying superlinear loss amplification in stress. Full information runs show that common information and mechanical withdrawal constraints can sharpen run dynamics without a liquidity backstop. Timelock constrained governance shows that parameters can be frozen over economically relevant horizons, motivating stress testing under fixed inputs. Oracle manipulation and latency show that incorrect price inputs can induce false solvency and delayed liquidation, directly increasing expected shortfall. Congestion wrong way risk shows that liquidation may become economically non viable precisely when needed because gas and MEV costs spike with stress.

Building on these failure modes, the paper derives five tractable metrics (V1-V5) and aggregates them into a vault credit score. V1 corrects oracle based coverage for stress dependent execution wedges. V2 conditions expected shortfall on liquidation mass and depth, making recovery endogeneity explicit. V3 quantifies liquidity fragility as the probability of hitting the utilization boundary within a horizon. V4 measures oracle integrity via latency and manipulation exposure, acknowledging partial identification of the latent true price. V5 measures the stress conditional viability of liquidation execution under gas and MEV constraints. Together, these metrics enable a depositor protective measurement discipline that is transparent about its assumptions and can be implemented with onchain data.

%Several important extensions are identified as distinct research agendas that fall outside the formal scope of Level 1 measurement. Vaults deploying capital into AMM liquidity pools face a sixth loss channel --- loss versus rebalancing --- that is continuous rather than event driven and accumulates from informed order flow rather than from liquidation mechanics. Vaults holding correlated leverage loop positions face amplified liquidation clustering that is endogenous to the collateral's leverage multiplierand correlated leverage ratio, both of which should be tracked as first order diagnostics for V2 (Remark~\ref{rem:leverage_loop}). And vaults accepting RWA collateral face systematic oracle bias during NAV latency windows that violates the partial identification bound derived for DeFi native collateral (Remark~\ref{rem:rwa_oracle}). These extensions do not invalidate V1-V5 for protocol bound lending vaults; they specify the conditions under which the framework requires augmentation before application to more complex vault architectures.

The framework’s limitations are: paraapproximations for impact and oracle latency may fail in crisis regimes, DeFi non stationarity implies that estimators must be versioned and re-estimated after upgrades and tail event scarcity limits statistical power for traditional backtesting. These limitations motivate conservative stress scenario design, sensitivity analysis and explicit reporting of identified bounds. The empirical implementation section, therefore, emphasizes data architecture, identification strategies and partial identification results as first class components of credit risk measurement.

Finally, this paper isolates vault risk at Level 1 which is a mechanical loss channels under fixed parameters and functioning smart contracts. A complete risk assessment requires all three levels of the decomposition. Level 2 governance risk, i.e., the quality of the process that sets and updates parameters, is developed in a companion paper. Level 3 infrastructure and code integrity risk, i.e., whether the vault's dependency graph executes as specified, is formalized in Section~\ref{subsec:level3} via Propositions~\ref{prop:L3_superadditive}--\ref{prop:L3_dominance} and Corollary~\ref{cor:cross_protocol_L3}, establishing the superadditivity of code failure probability in dependency depth and the dominance condition that determines when Level 3 is the binding depositor concern. Advancing DeFi credit measurement ultimately requires data public goods including standardized liquidation routing datasets, depth curve archives, oracle update registries and stress episode panels across chains. The stakes are not limited to DeFi markets. As tokenized collateral and onchain intermediation expand, the ability to measure credit risk in rule bound, execution driven systems will become a prerequisite for robust financial intermediation.

% --- PAPER 1, VERSION 1.1 COMPLETE.
% TITLE: The Vault as a Credit Instrument: Structural Risk
% Measurement for DeFi Lending
% INTEGRATED: academic reference memo additions 1-11
% READY FOR REVIEW ---

% --- PAPER 1, VERSION 2.0 REVISED.
% TITLE: The Vault as a Credit Instrument: Structural Risk
% Measurement for DeFi Lending
% REVISION: Pre-submission pass, all 8 passes, skill academic-finance-revision v2.0.0
% READY FOR SUBMISSION PENDING AUTHOR REVIEW OF FLAGGED ITEMS ---

\appendix

\section{Notation Summary}
\label{app:notation}

\begin{table}[htbp]
\centering
\caption{Principal notation used throughout the paper.}
\label{tab:notation}
\begin{tabular}{p{0.22\textwidth} p{0.72\textwidth}}
\toprule
\textbf{Symbol} & \textbf{Definition} \\
\midrule
$v$ & Vault index (credit instrument). \\
$u$ & Borrower account index within vault $v$. \\
$t, T$ & Time index; $T$ is the evaluation horizon. \\
$\mathcal{C}^{v}$, $\mathcal{D}^{v}$ & Sets of collateral and debt (borrowed) assets for vault $v$. \\
$P_{a,t}$ & Oracle-reported unit-of-account price for asset $a$ at time $t$. \\
$\widetilde{P}_{a,t}$ & Realized liquidation execution price for asset $a$ net of slippage and fees. \\
$P_{a,t}^{\mathrm{true}}$ & Latent efficient price (unobservable). \\
$\eta_{a,t}$ & Oracle error: $\eta_{a,t} \equiv P_{a,t} - P_{a,t}^{\mathrm{true}}$. \\
$C_{u,a,t}$ & Quantity of collateral asset $a$ held by account $u$ at time $t$. \\
$B_{u,d,t}$ & Quantity of debt asset $d$ owed by account $u$ at time $t$. \\
$w_a^{v}$ & Collateral factor (LTV ceiling) for asset $a$ in vault $v$. \\
$\Theta_t^{v}$ & Vault parameter vector: liquidation thresholds, close factors, incentives, caps. \\
$\Lambda_t^{v}$ & Vault liquidity state: utilization, available cash, withdrawal queue. \\
$X_t^{v}$ & Level 1 credit state tuple (Definition~\ref{def:level1_state}). \\
$\mathcal{E}^{v}$ & Vault rule set (smart contracts). \\
$\mathcal{I}_{t,i}^{v}$ & Binary integrity indicator for node $i$ at time $t$. \\
$\mathcal{F}_t$ & Public information filtration at time $t$. \\
$\mathcal{M}$ & Exogenous market and network environment. \\
$L_t^{v}$ & Total depositor liabilities (book value) at time $t$. \\
$A_t^{v}$ & Realized liquidation value assets available to satisfy claims. \\
$\Delta_t^{v}$ & Vault shortfall: $\Delta_t^{v} \equiv (L_t^{v} - A_t^{v})^+$. \\
$\ell_t^{v}$ & Loss rate: $\ell_t^{v} \equiv \Delta_t^{v}/L_t^{v} \in [0,1]$. \\
$\mathrm{HF}_{u,t}^{v}$ & Health factor for account $u$; liquidation triggered when $\mathrm{HF} < 1$. \\
$\mathrm{ACR}_t^{v}$ & Vault asset coverage ratio: $\mathrm{ACR}_t^{v} \equiv A_t^{v}/L_t^{v}$. \\
$\widetilde{\mathrm{ACR}}_t^{v}$ & Stress-effective coverage ratio accounting for execution deviation. \\
$\bar{\varepsilon}_t^{v}$ & Collateral weighted oracle execution deviation. \\
$U_t^{v}$ & Utilization ratio: $U_t^{v} \equiv B_t^{v}/D_t^{v} \in [0,1]$. \\
$B_t^{v}$, $D_t^{v}$ & Total borrows and total deposits for vault $v$. \\
$\mathrm{RR}_{u,t}^{v}$ & Realized recovery rate for account $u$ (equation~\eqref{eq:recovery_rate}). \\
$Q_{a,t}^{\mathrm{liq}}$ & Aggregate liquidation sell volume in collateral asset $a$. \\
$D_{a,t}^{\mathrm{pool}}$ & Effective onchain pool depth for asset $a$. \\
$f(Q,D)$ & Price impact function: $\partial f/\partial Q > 0$, $\partial f/\partial D < 0$. \\
$\lambda_a$ & Linear price impact coefficient for asset $a$. \\
$\mathcal{L}_{u,t}$ & Liquidation trigger event: $\{\mathrm{HF}_{u,t}^{v} < 1\}$. \\
$\mathcal{V}_{u,t}$ & Liquidation viability event (equation~\eqref{eq:viability}). \\
$g_t$ & Prevailing gas price at time $t$. \\
$\mathrm{MEV}_{u,t}^{\mathrm{cost}}$ & MEV and priority competition cost for account $u$. \\
$\delta^{\mathrm{lock}}$ & Governance timelock duration. \\
$\Delta^{*v}$ & Critical response window (Proposition~\ref{prop:curator_lag}). \\
$G^{v}=(V^{v},E^{v})$ & Smart contract dependency graph for vault $v$. \\
$k \equiv |V^{v}|$ & Dependency depth (number of critical path nodes). \\
$q_{\mathrm{code}}^{v}(h)$ & Code failure probability over horizon $h$ (Definition~\ref{def:level3_state}). \\
$\mathcal{B}_h^{v,L3}$ & Level 3 breach event over horizon $h$. \\
$\mathrm{V1}$--$\mathrm{V5}$ & Five Level 1 credit risk metrics (Sections~\ref{subsec:V1}--\ref{subsec:V5}). \\
$\mathrm{VCS}$ & Vault Credit Score: aggregation of V1-V5 (Section~\ref{subsec:VCS}). \\
$\hat{m}_j(v,t)$ & Normalized score for $j$, $\in [0,1]$ (higher = lower risk). \\
$\mathcal{S}$, $\mathcal{Z}(s)$ & Stress scenario set and associated state constraints. \\
$\mathrm{CLR}_t^{v}$ & Correlated leverage ratio (equation~\eqref{eq:CLR}). \\
$\mathrm{HD}(a)$ & Rehypothecation depth for collateral asset $a$. \\
$\mathrm{CM}(a)$ & Cascade multiplier for asset $a$. \\
\bottomrule
\end{tabular}
\end{table}

%\bibliographystyle{apalike}
%\bibliography{literature}

\end{document}